\documentclass[preprint]{elsarticle}

\usepackage{amsmath,amssymb,amsfonts,amsthm,bbm}
\usepackage{graphicx}
\usepackage{comment}
\usepackage{textcomp}
\usepackage[makeroom]{cancel}
\usepackage{indentfirst}
\usepackage{ulem}
\usepackage{tikz}
\usepackage{epstopdf} %converting to PDF
\usepackage{epsfig}
\usepackage{xcolor}
\usepackage{algorithm,algpseudocode}
\usepackage[linesnumbered,ruled,vlined,algo2e]{algorithm2e}
\usepackage{hyperref}
\usepackage{etoolbox}

\usepackage{natbib}

\usepackage{xcolor}

\usepackage{tabularx}

\usepackage{optidef}

\newbool{short}

\boolfalse{short} % use to get full version

\DeclareMathOperator*{\argmin}{arg\,min}

\usepackage{mathtools}
\mathtoolsset{showonlyrefs}

\usepackage[utf8]{inputenc}
\usepackage[english]{babel}

\renewcommand{\emph}{\textit}

\theoremstyle{definition}
\newtheorem{theorem}{Theorem}
\newtheorem{corollary}{Corollary}
\newtheorem{assumption}{Assumption}
\newtheorem{lemma}{Lemma}
\newtheorem{definition}{Definition}
\newtheorem{problem}{Problem}
\newtheorem{proposition}{Proposition}

\usepackage{csquotes}
\usepackage{float}
\usepackage{accents}

\def\BibTeX{{\rm B\kern-.05em{\sc i\kern-.025em b}\kern-.08em
    T\kern-.1667em\lower.7ex\hbox{E}\kern-.125emX}}

\newcommand{\R}{\mathbb{R}}
\newcommand{\N}{\mathbb{N}}
\newcommand{\atimer}[1]{\tau_{c#1}}
\newcommand{\ctimer}[1]{\tau_{g#1}}
\newcommand{\dtimer}{\tau_{d}}
\newcommand{\HFO}{\mathcal{H}_{\text{FO}}}
\newcommand{\Htime}[2]{(t_{#1},{#2})}
\newcommand{\Hatime}[2]{(t_{#1},j{#2})}
\newcommand{\Hotime}[2]{({#1},{#2})}
\newcommand{\alphatime}[1]{\textcolor{black}{\bar{\alpha}\left(#1\right)}}
\newcommand{\xtime}[1]{\xi(#1)}

\newcommand{\dom}{\textnormal{dom }}
\newcommand{\HFOR}{\mathcal{H}_{\text{FO}}^{\rho}}
\newcommand{\n}[1]{\textnormal{#1}}

\newcommand{\twonorm}[1]{\left\|#1\right\|}

\newcommand{\paren}[1]{\left(#1\right)}

\newcommand{\blue}[1]{\textcolor{black}{#1}}
\newcommand{\green}[1]{\textcolor{black}{#1}}

\usepackage{subfiles} % Best loaded last in the preamble

\begin{document}

\title{\LARGE \bf{Autonomous Satellite Rendezvous 
via Hybrid Feedback Optimization}}
\author{
    \texorpdfstring{
        Oscar Jed R. Chuy\footnote{School of Electrical and Computer Engineering, Georgia Institute of Technology, Atlanta, GA, USA 30332. 
        Emails: \texttt{\{ochuy3,matthale\}@gatech.edu}.}, 
        Matthew T. Hale$^1$, 
        Vignesh Sivaramakrishnan\footnote{Space Vehicles Directorate, Air Force Research Laboratory, Kirtland AFB, NM, USA 87108. 
        Emails: \texttt{vignesh.sivaramakrishnan.ext@afresearchlab.com,~sean.phillips.9@spaceforce.mil}}, 
        Sean Phillips$^2$, and 
        Ricardo G. Sanfelice\footnote{School of Electrical and Computer Engineering,  University of California, Santa Cruz, CA, USA 95064.
        Email: \texttt{ricardo@ucsc.edu}.  \\
        \indent Approved for public release; distribution is unlimited. Public Affairs approval \#AFRL-2026-0259. The views and opinions presented herein are those of the author and do not necessarily represent the views of DoD, DAF or AFRL. Appearance of, or reference to, any commercial products or services does not constitute DoD endorsement.\\
            \indent \ \  Chuy, Hale, and Sanfelice were supported by AFOSR under grant FA9550-19-1-0169.
            Chuy and Hale were supported by ONR under grants N00014-21-1-2495, N00014-22-1-2435, and N00014-26-1-2068,
            and by AFRL under grants FA8651-22-F-1052 and FA8651-23-F-A006. 
            Sanfelice was supported by NSF Grants no. CNS-2039054 and CNS-2111688, by AFOSR Grants nos. FA9550-23-1-0145, FA9550-23-1-0313, and FA9550-23-1-0678, by AFRL Grant nos. FA8651-22-1-0017 and FA8651-23-1-0004, by ARO Grant no. W911NF-20-1-0253, and by DoD Grant no. W911NF-23-1-0158.
            Sivaramakrishnan is supported in part by an appointment to the NRC Research Associateship Program at the Air Force Research Laboratory, Space Vehicles Directorate, administered by the Fellowships Office of the National Academies of Sciences, Engineering, and Medicine.
        }
        % \footnote{
        %     Approved for public release; distribution is unlimited. Public Affairs approval \#AFRL-2026-0259. The views and opinions presented herein are those of the author and do not necessarily represent the views of DoD, DAF or AFRL. Appearance of, or reference to, any commercial products or services does not constitute DoD endorsement.\\
        %     \indent \ \  Chuy, Hale, and Sanfelice were supported by AFOSR under grant FA9550-19-1-0169.
        %     Chuy and Hale were supported by ONR under grants N00014-21-1-2495, N00014-22-1-2435, and N00014-26-1-2068,
        %     and by AFRL under grants FA8651-22-F-1052 and FA8651-23-F-A006. 
        %     Sanfelice was supported by NSF Grants no. CNS-2039054 and CNS-2111688, by AFOSR Grants nos. FA9550-23-1-0145, FA9550-23-1-0313, and FA9550-23-1-0678, by AFRL Grant nos. FA8651-22-1-0017 and FA8651-23-1-0004, by ARO Grant no. W911NF-20-1-0253, and by DoD Grant no. W911NF-23-1-0158.
        %     Sivaramakrishnan is supported in part by an appointment to the NRC Research Associateship Program at the Air Force Research Laboratory, Space Vehicles Directorate, administered by the Fellowships Office of the National Academies of Sciences, Engineering, and Medicine.
        % }
    }{Oscar Jed R. Chuy, Matthew T. Hale, Vignesh Sivaramakrishnan, Sean Phillips, Ricardo G. Sanfelice}
}

\begin{abstract}
     As satellites have proliferated, interest has increased in autonomous rendezvous, proximity operations, and docking (ARPOD). 
     A fundamental challenge in these tasks 
     is the uncertainties when operating in space, e.g., in 
     measurements of
     satellites' states, 
     which can make future states difficult to predict. Another challenge is that satellites' 
     onboard processors are typically
     much slower than their terrestrial counterparts. 
     Therefore, to address these challenges
     we propose to solve an ARPOD problem with \emph{feedback optimization}, which computes inputs to a system by measuring its outputs, feeding them into an optimization algorithm in the loop, and computing some number of iterations towards an optimal input. We focus on satellite rendezvous, and satellites' dynamics are modeled using the continuous-time Clohessy-Wiltshire equations, which are marginally stable.  We develop an asymptotically stabilizing controller for them, and we use discrete-time gradient descent in the loop to compute inputs to them. Then, we analyze the hybrid feedback optimization system formed by the stabilized Clohessy-Wiltshire equations with gradient descent in the loop. We show that this model is well-posed and that maximal solutions are both complete and non-Zeno. Then, we show that solutions converge exponentially fast to a ball around a rendezvous point, and we bound the radius of that ball in terms of system parameters. Simulations show that this approach provides up to a 98.4\% reduction in the magnitude of disturbances across a range of simulations, which illustrates the viability of hybrid feedback optimization for autonomous satellite rendezvous.     
\end{abstract}

\maketitle

\section{Introduction} \label{sec:intro}
Autonomous rendezvous, proximity operations, and docking (ARPOD) is a critical class of tasks for satellites in space~\cite{harpley2025space,robinson-smith2025spacex,petersonSafeRPODs}. 
The execution of ARPOD tasks allows for transportation of personnel, resupply, and servicing/repair \cite{WoffindenNavAutoRnds,BehrendtARPODMPC,JewisonBenchmarProb}. However, due to satellite payload restrictions and financial costs, computing power onboard satellites is limited~\cite{Scharf2003_1,Scharf2004_2,petersen2021challenge},
which restricts the speed of computations that are performed
when executing such tasks. 
Additionally, unanticipated modeling errors can cause poor performance and, potentially, cause mission failure,
in part because predictions of future satellite states are unreliable~\cite{robertson2003satellite,chern2006lesson}. 
What we seek is an ARPOD approach that (i) does not predict
future states,
(ii) rejects disturbances when they arise, and
(iii) has low computational complexity that accommodates the
slow processors often found onboard satellites. 

Therefore, in this paper we solve an ARPOD problem with \emph{feedback optimization}. 
Feedback optimization measures the outputs of a control system, feeds them
into a running optimization algorithm in the loop, and uses that algorithm to compute
new inputs to a system~\cite{HausT-Sinautoopt}. 
Feedback optimization avoids solving optimization problems offline in a feedforward configuration~\cite{hauswirth2024optimization,batchRLHanchen2020, processcntrl2010, OPSFfotz2000, KKTcontrollerAndrej, BIGNUCOLORadMV, ORTMANN2020106782}. 
%Instead, it measures the system outputs~\cite{hauswirth2018stab,hauswirth2024optimization,ORTMANN2020106782,HausT-Sinautoopt,processcntrl2010} and then uses those measurements to optimize inputs with an in-the-loop optimization algorithm.  
It offers inherent robustness to inaccurate system models and time-varying parameters, 
and eliminates the need for pre-computed set points or reference signals~\cite{hauswirth2024optimization, processcntrl2010}. 
It differs fundamentally from model-predictive
control (MPC) because it does not attempt to predict future states \cite{alexanderARPODGeoCM}. 
Such predictions
can be difficult under the uncertainties that arise in space-based autonomy, which
makes feedback optimization a natural fit. 
In general, feedback optimization also does not require an optimal input to be computed
exactly by the optimization algorithm in the loop. Instead, some number of steps can be taken to compute iterates that work toward
the optimal input, and one of those iterates can be used as a sub-optimal
input to a system.
This property helps accommodate slow processors onboard satellites,
which may not be able to exactly compute optimal inputs before
an input to a satellite is needed. 

We focus on a satellite rendezvous problem in which there are two satellites, the target and the chaser, 
that are in low earth orbit. 
The target satellite is uncontrolled and orbiting the earth, while
the chaser satellite is controlled. The goal is to drive the chaser 
to rendezvous with the target. 
Under mild conditions, we can linearize the chaser's dynamics to obtain the 
Clohessy-Wiltshire (CW) equations~\cite{CurtisOrbobook}, which
describe the motion of the chaser relative to the target. 
It is standard in feedback optimization to assume that a system 
is asymptotically stable 
before optimizing its inputs~\cite{giuseppeOFES}. However, the CW equations give a marginally stable system. 
Therefore, we first design a stabilizing feedback controller for the CW equations, and we show that only
a subset of available 
gains must be non-zero to set all closed-loop eigenvalues to desired values. We also derive
closed forms for each gain in terms of the desired closed-loop eigenvalues. 

%The inclusion of a pre-stabilized system is an approach seen in \cite{giuseppeOFES} which ensures even marginally stable systems like the CW equations, will readily allow for direct FO implementation.
%Feedback optimization disregards solving optimization problems in a feedforwards configuration seen in~\cite{hauswirth2024optimization,batchRLHanchen2020, processcntrl2010, OPSFfotz2000, KKTcontrollerAndrej, BIGNUCOLORadMV, ORTMANN2020106782} and instead measures the system outputs~\cite{hauswirth2018stab,hauswirth2024optimization,ORTMANN2020106782,HausT-Sinautoopt,processcntrl2010} and then uses those measurements to optimize inputs with an in-the-loop optimization algorithm. With FO, its inherent formulation leads to the following benefits: robustness to inaccurate system models and time-varying parameters, achieves constraint satisfaction with minimal model dependence, and importantly eliminates the need for pre-computed set points or reference signals~\cite{hauswirth2024optimization, processcntrl2010}. 

The stabilized CW equations 
that model the chaser 
are in continuous time, and we 
implement feedback optimization with 
discrete-time gradient descent running in the chaser's feedback loop.
We therefore 
model this setup 
as a hybrid system. 
We develop this model in the 
hybrid systems 
framework of~\cite{Hybridbook}, and we establish the existence and certain
properties of its solutions, which we use to characterize the long-term performance of
feedback optimization for the satellite rendezvous problem. 
To the best of our knowledge, this work is the first to apply feedback optimization
to an ARPOD problem.
In detail, the contributions of this paper are the following:
\begin{itemize}
    \item We design a controller to asymptotically stabilize the CW equations, and we give
    closed forms for the gains needed to attain desired closed-loop eigenvalues (Theorem~\ref{tm:stabGains}).    
    \item We develop a hybrid systems model of feedback optimization for the CW equations with gradient descent in the loop,
    we prove that it is well-posed, and we prove that all maximal solutions to it
    are both complete and non-Zeno (Lemma~\ref{lem:wellPosed}, Proposition~\ref{prop:maxsolComplete}). 
    \item We bound the steady-state rendezvous error of the chaser 
    satellite in terms of disturbances and other system parameters
    (Proposition~\ref{prop:completeHOCW}, Theorem~\ref{tm:globalcompleteHOCW}). 
    \item We prove that this implementation of feedback optimization is robust to perturbations, 
    in the sense that bounded perturbations to 
    the hybrid model produce bounded
    perturbations in the chaser's trajectories (Corollary \ref{cor:robust}).
    \item We show in simulations that hybrid feedback optimization successfully reduces
    the magnitude of disturbances 
    in satellite rendezvous 
    by up to 98.4\% (Section \ref{sec:simResults}).
\end{itemize}

A number of related approaches have been developed to solve ARPOD problems. 
For example, the authors in \cite{soderlund1,soderlund2} propose a switching controller that assures local asymptotic stability despite drift and under-actuation in the underlying satellite dynamics.
Researchers have also used nonlinear estimators and sliding mode control to execute ARPOD maneuvers in the presence of faulty thrusters and physical disturbances in space~\cite{henry20216}.
To meet real-time demands, developments in~\cite{BehrendtARPODMPC,behrendt2025timeconstrained} provide 
a model-predictive control (MPC) approach that is computationally time-constrained and robust to disturbances drawn from a Gaussian distribution. 
We differ from all of these works by considering a 
hybrid model in which a continuous-time chaser satellite is controlled by a discrete-time optimization algorithm it runs onboard. 
The feedback optimization controller that we present also has lower computational complexity than the aforementioned works
because it only requires some number of gradient descent
iterations to be performed when computing each of the chaser's inputs. 
To the best of our knowledge, this work is the first to apply
feedback optimization to an ARPOD problem. 

%Each of these approaches tackled different aspects of the ARPOD problem, our work wishes to make further improvements in comparison to these previous works with improved modeling and computational efficiency. Specifically, hybrid system models offer unique properties relevant in tackling the problem and our work allows for more efficiency by minimizing computations. 

% \textcolor{red}{Oscar, right here can you say something about where these approaches fall short?
% Otherwise the reader might say ``Well if all of these papers exist, then why do we need one more?''
% One option is to say that we want to use the fewest computations possible, e.g., for cube sats, so
% we want to not run an estimator or MPC or anything like that. We can also talk to Sean
% and Vignesh and see what they think.
% }

Within the feedback optimization literature, 
related work in \cite{giuseppeOFES, cothren2023perceptionbasedsampleddataoptimizationdynamical,chen25} considers a continuous-time system and discrete-time computations in a sampled-data feedback optimization configuration. Results in~\cite{giuseppeOFES,cothren2023perceptionbasedsampleddataoptimizationdynamical} show that large enough sample times guarantee closed-loop stability, and they provide practical stability guarantees under time-varying disturbances. 
Results in~\cite{chen25} show that when a closed-loop system's inputs change at a fixed rate, global exponential stability
is obtained.
The current paper differs because we develop a hybrid
model in the framework of \cite{Hybridbook}, 
which allows us to characterize system behavior at all times, 
rather than just at certain sample times. 
Moreover, we use the hybrid model to derive analytical
guarantees of robustness, including robustness to errors in the rate at which
inputs are applied to the chaser, which we allow to vary over time. 
Some of our modeling developments are related to those
in~\cite{chuy2025hybridsystemsmodelfeedback} by a subset of the authors
of the current paper, 
which also considered hybrid feedback optimization. 
However, the current paper differs because it must analyze
state dynamics that are subject to unknown, time-varying
disturbances (see Section~\ref{sec:stabDyn}), 
while developments
in~\cite{chuy2025hybridsystemsmodelfeedback} studied an unperturbed
state equation. 
%we can apply robustness guarantees and interconnect continuous and discrete time. We therefore develop a hybrid system model of feedback optimization with continuous-time dynamics and discrete-time optimization. 
%This work applies a previous hybrid system model of feedback optimization seen in \cite{chuy2025hybridsystemsmodelfeedback} to satellite docking. 

The rest of the paper is organized as follows. Section \ref{sec:prelim} provides background, and 
Section~\ref{sec:prob} gives a formal problem statement. 
Section \ref{sec:stabDyn} 
develops a stabilizing controller, and  
Section \ref{sec:hybridModel} develops the hybrid feedback optimization model of the
chaser satellite. 
Section \ref{sec:convAnalysis} analyzes the convergence
of the chaser to a rendezvous point. 
Section \ref{sec:simResults} presents simulations, and Section \ref{sec:results} concludes.

\section{Preliminaries and Background} \label{sec:prelim}
This section gives background on the Clohessy-Wiltshire equations,
feedback optimization, and hybrid systems. 

\subsection{Notation}
Let $\R$ denote the set of real numbers and $\N$ denote the set of non-negative integers. 
%For a differentiable function $\Phi : \R^m \times  \R^p \rightarrow \R$, let $\nabla_{\mathbf{u}}\Phi$ denote the partial derivative with respect to its first argument. 
Let $I_n$ denote the~$n \times n$ identity matrix and let~$\mathbbm{1}_{n}$ denote the all-ones vector of size~$n$.
Given scalars~$a_1, a_2, \ldots, a_n$, 
we use~$\textnormal{diag}(a_1, a_2, \ldots, a_n)$
to denote the diagonal matrix with~$a_1, a_2, \ldots, a_n$ on its main diagonal. 
The~$2$-norm of a vector~$\mathbf{x}$ is denoted~$\twonorm{\mathbf{x}}$.
We denote the Euclidean projection of a vector $\mathbf{v}$ onto 
a non-empty, compact, convex set
$\mathcal{Z}$ by $\Pi_{\mathcal{Z}}[\mathbf{v}] = \arg\min_{\mathbf{z} \in \mathcal{Z}} \|\mathbf{v}-\mathbf{z}\|$.
The distance from a point~$\mathbf{v} \in \mathbb{R}^n$ to a non-empty set~$\mathcal{A} \subseteq \mathbb{R}^n$
is denoted~$\|\mathbf{v}\|_{\mathcal{A}} := \inf_{\mathbf{a} \in \mathcal{A}} \|\mathbf{v} - \mathbf{a}\|$. 
We denote the diameter of a non-empty, compact, convex set
$\mathcal{Z}$ by~$d_{\mathcal{Z}}$. 
%Let $\lambda_{i}(N)$ denote the $i^{\n{th}}$  eigenvalue of the matrix $N$, and
Let~$\textnormal{eig}(M)$ denote the set of eigenvalues of a square matrix~$M$. 
Let~$\lambda_{max}(M)$ denote the largest eigenvalue
of a symmetric matrix~$M$, and let~$\lambda_{min}(M)$ 
denote its smallest eigenvalue. 
%For $r \geq 0$, we use $B_{r}(\tilde{\mathbf{x}}) := \{\mathbf{x}\in \R^n  : \|\mathbf{x} - \tilde{\mathbf{x}}\| \leq r\}$ to denote the closed Euclidean ball of radius $r$ about the point $\tilde{\mathbf{x}}$.
For a finite multi-set~$S$, we use~$\mu_{\max}(S)$ to denote the largest multiplicity
of an element of~$S$, i.e., the largest number of times an element 
appears in~$S$.

\begin{figure}[ht]
    \centering
    \includegraphics[width=0.75\linewidth]{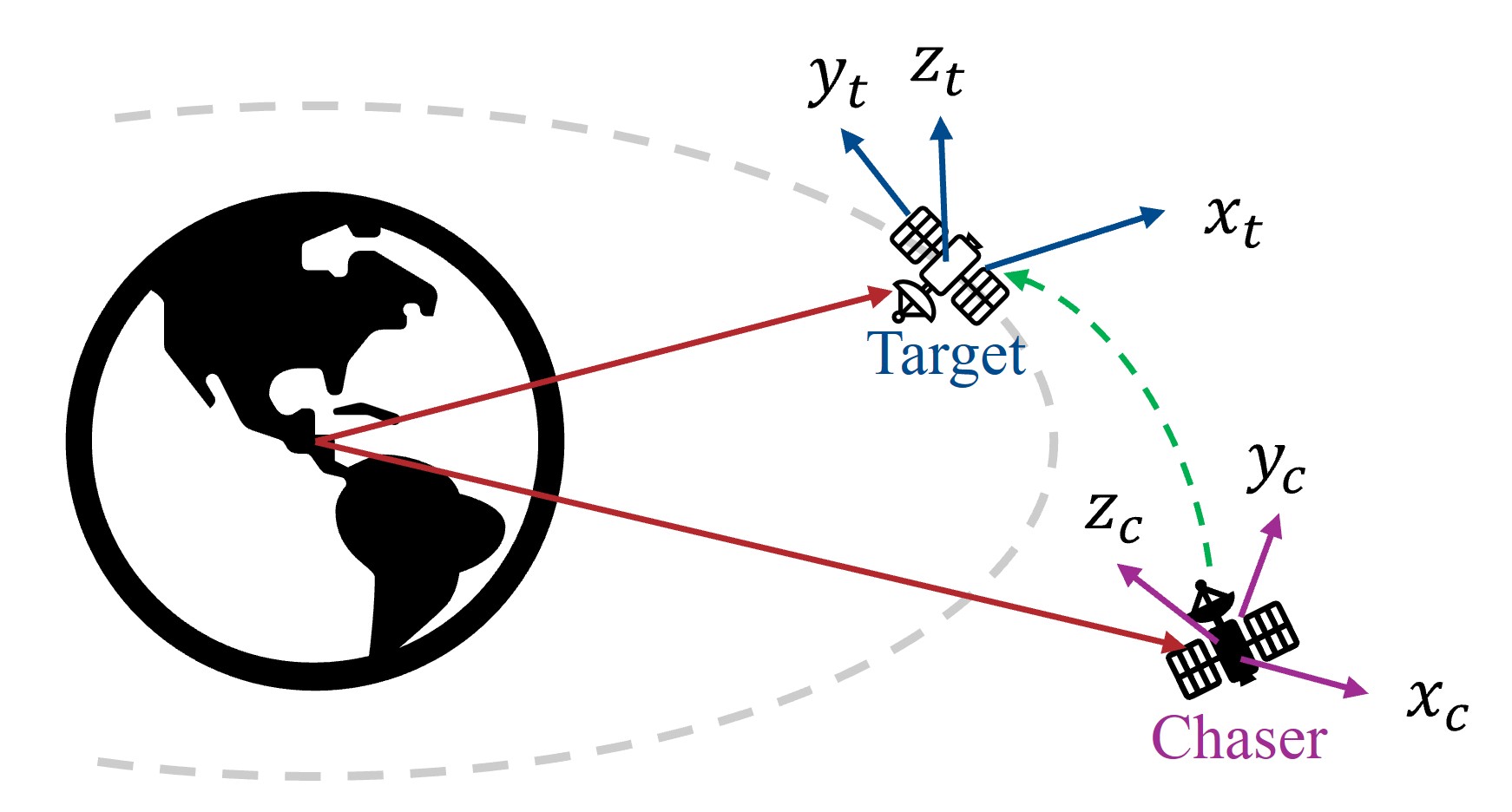}
    \caption{The target and chaser satellites are orbiting around the earth and the green line defines a rendezvous trajectory. The axes of each satellite are defined as follows: the radial 
    vectors~$x_t$ and~$x_c$ point away from the earth, 
    the tangential vectors
    $y_t$ and~$y_c$ point in the direction each satellite is moving along its orbit, and the out-plane 
    vectors~$z_t$ and~$z_c$ are 
    perpendicular to the orbital planes formed by~$x_t/y_t$
    and~$x_c/y_c$, respectively.}
    \label{fig:CWSetup}
\end{figure}

\subsection{Clohessy-Wiltshire Equations}
As described in the Introduction, we study the satellite rendezvous problem seen in Figure \ref{fig:CWSetup}.  
We consider the relative dynamics between a controlled chaser 
satellite and an uncontrolled target satellite. To 
model these dynamics, we make the following standard assumptions~\cite{petersonSafeRPODs,BehrendtARPODMPC}: 
\begin{enumerate}
    \item The target satellite is in an uncontrolled circular orbit around the earth. 
    \item The distance between the target and chaser satellites is much less than the orbital radius of the target satellite.
    \item The chaser satellite operates with bi-directional thrusters and external torque generators installed 
    along the axes aligning with its body-fixed coordinate frame. 
\end{enumerate}

Then one can obtain the dynamics for the chaser's position and velocity relative to the target.
We denote the chaser's position relative to the target by~$(x, y, z)^T \in \mathbb{R}^3$, and
linearizing these dynamics gives the  Clohessy–Wiltshire (CW) equations: 
\begin{equation} \label{eq:cwDynamics}
\left[
\begin{array}{c}
    \dot{x} \\
    \dot{y} \\
    \dot{z} \\
    \ddot{x} \\
    \ddot{y} \\
    \ddot{z} 
\end{array}
\right]
=  
\left[
    \begin{array}{cccccc}
        0 & 0 & 0 & 1 & 0 & 0 \\
        0 & 0 & 0 & 0 & 1 & 0 \\
        0 & 0 & 0 & 0 & 0 & 1 \\
        3w^2 & 0 & 0 & 0 & 2w & 0 \\
        0 & 0 & 0 & -2w & 0 & 0 \\
        0 & 0 & -w^2 & 0 & 0 & 0 
    \end{array} 
\right]
\left[
\begin{array}{c}
    {x} \\
    {y} \\
    {z} \\
    \dot{x} \\
    \dot{y} \\
    \dot{z} 
\end{array}
\right]
+
\frac{1}{m_c}
\left[
    \begin{array}{ccc}
        0 & 0 & 0  \\
        0 & 0 & 0  \\
        0 & 0 & 0  \\
        1 & 0 & 0  \\
        0 & 1 & 0  \\
        0 & 0 & 1  
    \end{array}
\right]
\left[
\begin{array}{c}
    {u_{x}} \\
    {u_{y}} \\
    {u_{z}} 
\end{array}
\right], 
\end{equation}
where $m_c$ is the mass of the chaser, $w := \sqrt{\mu/a^3}$, $\mu := 3.986 \times10^{14} m^3s^{-2}$ is the standard gravity parameter of the earth, and $a$ is the orbital radius of the satellite.
These equations model the translational motion of the chaser
relative to the target in the CW coordinate frame that is fixed
to the target satellite. 
The CW equations define a marginally stable system, and 
we compactly express these dynamics as
\begin{equation} \label{eq:cwdynamics}
\blue{\dot{\boldsymbol{x}}} = A_{\n{CW}}\blue{\boldsymbol{x}} + B_{\n{CW}}\mathbf{u}, 
\end{equation}
where~$\blue{\boldsymbol{x}} = (x, y, z, \dot{x}, \dot{y}, \dot{z})^T \in \mathbb{R}^6$ and the matrices
$A_{\n{CW}} \in \R^{6 \times 6}$ and $B_{\n{CW}} \in \R^{6 \times 3}$
can be read off from~\eqref{eq:cwDynamics}. 
% We will design a hybrid
% feedback optimization model by pairing the CW equations with a discrete-time
% gradient descent algorithm in the loop. 

\subsection{Feedback Optimization Background} \label{ss:fobackground}
Suppose we have the linear time  invariant (LTI) system 
\begin{align} \label{eq:ltisystem}
    \begin{split}
        \blue{\dot{\mathbf{\chi}}} &= A\blue{\mathbf{\chi}} + B\mathbf{v} \\
        \mathbf{y} &= \Psi\mathbf{x} + \mathbf{d},
    \end{split}
\end{align}
where
$\blue{\mathbf{\chi}} \in \R^n$ is the system's state, 
$\mathbf{v} \in \R^m$ is its input,
and~$\mathbf{y} \in \R^p$ 
is its output. 
The mapping~$t \mapsto \mathbf{d}(t) \in \mathbb{R}^p$ 
models an unknown time-varying disturbance.

In feedback optimization, it is standard to take~$A$ to be Hurwitz
and we suppose that is the case here. 
If the mapping~$t \mapsto \mathbf{d}(t)$ were known, 
then to optimize the system's steady-state behavior, at each $t$ one could drive 
the pair~$(\mathbf{v}, \mathbf{y})$ to 
a solution of 
% \begin{subequations} \label{eq:fo_genform}
%     \begin{align}
%         \min_{\mathbf{v}(t),\mathbf{y}(t)} \quad & \Phi\big(\mathbf{v}(t),\mathbf{y}(t)\big) \label{eq:phidef} \\
%         \text{subject to} \quad & \mathbf{y}(t) = H\mathbf{v}(t) + \mathbf{d}(t), \,\, \mathbf{v}(t) \in \mathcal{U}, \,\, \mathbf{y}(t) \in \R^p, \label{eq:ddef} 
%     \end{align}
% \end{subequations}
\begin{mini!}|l|
    {\mathbf{v}^{ss},\mathbf{y}^{ss}}
    {\Phi(\mathbf{v}^{ss},\mathbf{y}^{ss})}
    {\label{eq:fo_genform}}
    {\label{eq:phidef}}
    \addConstraint{\mathbf{y}^{ss} = H\mathbf{v}^{ss} + \mathbf{d}, \,\, \mathbf{v}^{ss} \in \mathcal{U}, \,\, \mathbf{y}^{ss}
    \in \R^p,}{\label{eq:ddef} }
\end{mini!}
where~$\Phi : \R^m \times \R^p \to \R$ is strongly convex in~$(\mathbf{v}^{ss}, \mathbf{y}^{ss})$
and~$\mathcal{U} \subseteq \mathbb{R}^m$ is a non-empty, compact, convex set. 
These properties ensure the problem in~\eqref{eq:phidef}-\eqref{eq:ddef} has a unique solution. 

For Hurwitz~$A$, if~$\mathbf{d} \equiv 0$, then the steady-state 
input-to-output map 
of the system in~\eqref{eq:ltisystem}
would be~$\mathbf{v}^{ss} \mapsto H\mathbf{v}^{ss}$, where $H:= -\Psi A^{-1}B$.
We allow for~$\mathbf{d} \neq 0$, and therefore we consider 
the mapping $\mathbf{v}^{ss} \mapsto H\mathbf{v}^{ss} + \mathbf{d}$,
which is equal to the steady-state mapping with its value perturbed by~$\mathbf{d}$. 
If the disturbance~$\mathbf{d}$ were known, then one could plug the expression for~$\mathbf{y}^{ss}$ in~\eqref{eq:ddef}
into~\eqref{eq:phidef} and solve the set-constrained problem
\begin{mini!}|l|
    {\mathbf{v}^{ss}}
    {\Phi(\mathbf{v}^{ss},H\mathbf{v}^{ss} + \mathbf{d})}
    {\label{eq:fo_genform2}}
    {\label{eq:phidef2}}
    \addConstraint{\mathbf{v}^{ss} \in \mathcal{U}}{\label{eq:ddef2} }. 
\end{mini!}
Then\blue{,} one could, for example, solve this problem using a standard projected gradient descent law
in which each gradient descent step is projected onto the constraint set~$\mathcal{U}$ to enforce
the satisfaction of the set constraint. 
We denote the~$k^{\n{th}}$ iterate of this projected gradient descent law 
by~$\mathbf{v}^{ss}_k$. Then\blue{,} the first iterate would take the form 
\begin{equation} \label{eq:newbiggrad}
\mathbf{v}^{ss}_{1} = \Pi_{\mathcal{U}}\Big[\mathbf{v}^{ss}_{0} -\gamma\Big(\nabla_{\mathbf{u}} \Phi\big(\mathbf{v}^{ss}_0, H\mathbf{v}^{ss}_0 + \mathbf{d}\big) + H^T
\nabla_{\mathbf{y}}\Phi\big(\mathbf{v}^{ss}_0, H\mathbf{v}^{ss}_0 + \mathbf{d}\big)\Big)\Big],
\end{equation}
where~$\gamma > 0$ is a stepsize, $\nabla_{\mathbf{u}} \Phi$ denotes the gradient of~$\Phi$ with respect to its first
argument, and $\nabla_{\mathbf{y}} \Phi$ denotes the gradient of~$\Phi$
with respect to its second argument. 

However, we consider unknown disturbances~$t \mapsto \mathbf{d}(t)$, which means that 
\blue{the value of $\mathbf{y}^{ss}$ cannot be computed from~$\mathbf{v}^{ss}$
with~\eqref{eq:ddef}.} 
Then\blue{,} the update
law in~\eqref{eq:newbiggrad} \blue{also} cannot be executed as written. 
Instead, 
rather than computing~$\mathbf{y}^{ss}$, 
we will sample the actual values of the output~$\mathbf{y}$ and use those
values to optimize over~$\mathbf{v}^{ss}$. 
%To do so, one can sample a value of the output~$\mathbf{y}$ and use that value to optimize the input. 
Let~$\mathbf{y}_s$ denote a sampled
value of the output. 
Although~$\mathbf{y}_s$ need not be sampled from a system
at steady state, we will use the standard approximation
in feedback optimization that treats~$\mathbf{y}_s$
as coming from a system at steady state~\cite{hauswirth2024optimization, wang2023decentralized}. 
This approximation is implemented
by setting~$\mathbf{y}_s = H\mathbf{v}^{ss}_0 + \mathbf{d}$,
where~$\mathbf{d}$ is the disturbance at the time~$\mathbf{y}_s$
is sampled
and~$\mathbf{v}_0^{ss}$ is the input at the time~$\mathbf{y}_s$
is sampled. 

Then\blue{,} we can solve the problem in~\eqref{eq:fo_genform}  
with a projected gradient descent law whose~$(k+1)^{\n{th}}$ iterate takes the form
\begin{equation} \label{eq:gdUpdateLaw}
    \mathbf{v}^{ss}_{k+1} = \Pi_{\mathcal{U}}\Big[\mathbf{v}^{ss}_{k} -\gamma\Big(\nabla_{\mathbf{u}} \Phi\big(\mathbf{v}^{ss}_k, \mathbf{y}_s\big) + H^T
\nabla_{\mathbf{y}}\Phi\big(\mathbf{v}^{ss}_k, \mathbf{y}_s\big)\Big)\Big]. 
\end{equation}

In the feedback optimization literature, it is standard to formulate a closed-loop system between the plant in \eqref{eq:ltisystem} and the optimization algorithm 
in \eqref{eq:gdUpdateLaw}. 
%, the latter of which is used as a controller. 
Then the closed-loop interconnected system is 
\begin{align} \label{eq:cdltisystem}
    \text{Plant: }&
    \begin{cases}
        \blue{\dot{\mathbf{\chi}}} &= A\blue{\mathbf{\chi}} + B\mathbf{v} \\
        \mathbf{y} &= \Psi\blue{\mathbf{\chi}} + \mathbf{d},
    \end{cases}\\
    \text{Controller:}&
    \begin{cases} 
        \mathbf{v}^{ss}_{k+1} = \Pi_{\mathcal{U}}\Big[\mathbf{v}^{ss}_{k} -\gamma\Big(\nabla_{\mathbf{u}} \Phi\big(\mathbf{v}^{ss}_k, \mathbf{y}_s\big) + H^T
\nabla_{\mathbf{y}}\Phi\big(\mathbf{v}^{ss}_k, \mathbf{y}_s\big)\Big)\Big],
    \end{cases}
\end{align}
where at certain points in time we set~$\mathbf{y}_s = \mathbf{y}$
and~$\mathbf{v} = \mathbf{v}^{ss}_k$. 
Feedback optimization does not require the gradient descent update law to converge to an optimum before
using one of its iterates as the input to the plant. Our analysis below allows for this possibility 
onboard the chaser satellite as well. 
%which allows for the interconnection of the continuous-time plant with a discrete-time optimization algorithm as its controller. 

\subsection{Hybrid System Background}
Using the framework of~\cite{Hybridbook} a hybrid system $\mathcal{H}$ takes the form
\begin{equation}\label{eq:hybridmodel}
    \mathcal{H} := 
    \begin{cases}
        \begin{aligned}
            \dot{\zeta} &\in F(\zeta) \quad &\zeta \in C \\
            \zeta^+ &\in G(\zeta) \quad &\zeta \in D
        \end{aligned}
    \end{cases},
\end{equation}
where $\zeta \in \R^{n}$ is the system's state vector and the maps~$F$ and~$G$ are set-valued in general. 
The \blue{map} $F$ defines the flow map and governs the continuous dynamics within the flow set $C$, while 
$G$ defines the jump map, which models the system's discrete behavior within the jump set~$D$.

\begin{definition}[Hybrid Basic Conditions \cite{Hybridbook}]
\label{def:hybridcond}
     A hybrid system $\mathcal{H}$ 
    with data $(C, F, D, G)$ satisfies the hybrid basic conditions if 
    \begin{enumerate}
        \item $C$ and $D$ are closed subsets of $\R^n$; \label{hybridcond_one}
        \item $F: \R^n\rightrightarrows\R^n$ is outer semicontinuous\footnote{A set-valued mapping $M:\R^m\rightrightarrows \R^n$ is outer semicontinuous (osc) at $x \in \R^m$ if  for every sequence of points $\{x_i\}_{i \in \N}$ convergent to $x$ and any convergent sequence of points~$\{y_i\}_{i \in \N}$ with $y_i \in M(x_i)$, one has $y\in  M(x)$, where $\lim_{i\rightarrow\infty}y_i  =  y$ \cite{Hybridbook}.}, and locally bounded\footnote{A set-valued mapping $M:\R^m\rightrightarrows \R^n$ is locally bounded at $x \in \R^m$ if there is a neighborhood $U_x$  of $x$ such that $M(U_x)\subset \R^n$ is bounded \cite{Hybridbook}.} relative to $C$, $C\subset \text{dom}~F,$ and $F(\zeta)$ is convex for every $\zeta \in C$;  \label{hybridcond_two}
        \item $G: \R^n\rightrightarrows\R^n$ is outer semicontinuous and locally bounded relative to $D$, and $D\subset \text{dom}~G$. \label{hybridcond_three}
    \end{enumerate}
\end{definition}

If a hybrid system satisfies the hybrid basic conditions, then it is well-posed by \cite[Theorem 6.30]{Hybridbook},
which implies that errors in the models of~\blue{$C$, $D$,~$F$, and~$G$}
up to a certain threshold produced bounded
changes in the trajectories produced by the system over finite hybrid time horizons. 
For a hybrid system $\mathcal{H}$, its solutions, denoted by~$\phi$, are hybrid arcs that can in general be  maximal\footnote{A solution $\phi$ to $\mathcal{H}$ is maximal if there does not exist another solution $\psi$ to $\mathcal{H}$ such that dom~$\phi$ is a proper subset of dom~$\psi$ and $\phi(t,j)=\psi(t,j)$ for all $(t,j)\in \textnormal{dom } \phi$~\cite{Hybridbook}.}, complete\footnote{The solution $\phi$ is complete if~$ \textnormal{dom }\phi$ is unbounded, i.e., if $\textnormal{length}(\textnormal{dom }\phi) = \sup_t \textnormal{dom }\phi + \sup_j \textnormal{dom }\phi = \infty$~\cite{Hybridbook}.}, and Zeno\footnote{The solution $\phi$ is Zeno if it is complete and $\sup_t \textnormal{dom } \phi<\infty$~\cite{Hybridbook}.}. 
Complete solutions are defined over arbitrarily long \blue{hybrid} time horizons, and
Zeno behavior implies that solutions undergo an infinite number of jumps in finite time, i.e., a system's states stop flowing in finite time. 

\section{Problem Statement} \label{sec:prob} 
As described in the Introduction, 
the goal of this work is to develop a computation-in-the-loop approach to autonomous
satellite rendezvous that 
(i) does not require predictions of future states,
(ii) allows for in-the-loop computations 
of inputs 
to be slow and hence to produce sub-optimal inputs,
and
(iii) provides robustness to unmodeled disturbances.
In general, rendezvous of the chaser with the target
at the exact same location would cause a collision and is therefore
undesirable. Moreover, under the above conditions we do not expect exact 
convergence to a rendezvous point. Instead, we seek
to show approximate convergence to a Euclidean ball about a desired
rendezvous point.
Accounting for all of these factors, the problem we solve
in this paper is the following. 

\begin{problem} \label{prob:theone}
Develop a controller that drives
the chaser satellite to asymptotically 
approximately 
rendezvous with the target
satellite without attempting to predict future states, while
using only onboard computations. 
Show that this controller
is robust to perturbations in the chaser's dynamics. 
\end{problem}

\section{Stabilization of the  Clohessy-Wiltshire Dynamics} \label{sec:stabDyn}
% \blue{[$\HFO$ $F$, $G$, $C$, $D$ with $A_{\textnormal{stab}}$ and $B_{\textnormal{stab}}$]}
% In this section we solve Problem~\ref{prob:controllerdesign}. 
We consider the CW equations from~\eqref{eq:cwdynamics}
with full-state feedback and unknown disturbances in output
measurements. This setup is equivalent to~\eqref{eq:ltisystem}
with~$A = A_{\n{CW}}$, $B = B_{\n{CW}}$, and~$\Psi = I_6$. 
Then
\begin{align}
\dot{\mathbf{x}}  &= A_{\n{CW}}\mathbf{x} + B_{\n{CW}}\mathbf{v} \\
\mathbf{y} &= \mathbf{x} + \mathbf{d},
\end{align}
where~$\mathbf{d} \in \mathbb{R}^6$ is an unknown time-varying disturbance. 
We implement a controller of the form~$\mathbf{v}=-K\mathbf{y} + \mathbf{u}$, 
where~$\mathbf{u} \in \mathbb{R}^3$ 
is a new input that will be computed by the optimization algorithm in the loop,
and
\begin{equation}
K := \left[\begin{array}{cccccc}
k_1 & k_2 & k_3 & k_4 & k_5 & k_6 \\
k_7 & k_8 & k_9 & k_{10} & k_{11} & k_{12} \\
k_{13} & k_{14} & k_{15} & k_{16} & k_{17} & k_{18} 
\end{array}\right]
\in \mathbb{R}^{3 \times 6} 
\end{equation}
is a matrix of gains.
Then\blue{,} the closed-loop system we consider takes the form 
\begin{align} \label{eq:CWltisystem}
    \dot{\mathbf{x}} &= A_{\n{stab}}\mathbf{x} + B_{\n{stab}}\mathbf{u} - B_{\n{stab}}K\mathbf{d} \\
    \mathbf{y} &= \mathbf{x} + \mathbf{d},
\end{align}
where $B_{\n{stab}} := B_{\n{CW}}$ and $A_{\n{stab}} := A_{\n{CW}} - B_{\n{CW}}K$. Explicitly, 
\begin{equation} \label{eq:A_Stab}
 A_{\n{stab}} = 
    \left[
    \begin{array}{cccccc}
        0 & 0 & 0 & 1 & 0 & 0 \\
        0 & 0 & 0 & 0 & 1 & 0 \\
        0 & 0 & 0 & 0 & 0 & 1 \\
        3w^2 - k_{1} & - k_{2} & - k_{3} & - k_{4} & 2w - k_{5} & - k_{6} \\
        - k_{7} & - k_{8} & - k_{9} & -2w - k_{10} & - k_{11} & - k_{12} \\
        - k_{13} & - k_{14} & -w^2 - k_{15} & - k_{16} & - k_{17} & - k_{18} \\
    \end{array}
    \right].
\end{equation}
%To note, these controller can be set to achieve certain system parameters but not all controllers will perform optimally or require a black-box function. The following subsection will discuss gain computation in further detail.
%While~$18$ gains appear in~$K$, we will show that
%only~$6$ of these gains to non-zero values. 
%\subsection{Gain Selection}
% \blue{[How to choose $K$ with the block diagonal matrix]}
%In this subsection we address Problem \ref{prob:controllerdesign} and show that we can solve for controller gains based on desired eigenvalues much like pole placement but under some assumptions. 
We suppose that a multi-set~$\Lambda_{\n{des}}$ of~$6$ desired closed-loop eigenvalues is given. 

\begin{assumption} \label{as:UnqeEig}
    All elements in $\Lambda_{\n{des}}$ are real and negative.
\end{assumption}

We denote the elements of~$\Lambda_{\n{des}}$ by~$\lambda_1, \lambda_2, \ldots, \lambda_6$ and
without loss of generality we suppose that
\begin{equation} \label{eq:lambda_order}
\lambda_6 \leq \lambda_5 \leq \cdots \leq \lambda_1 < 0. 
\end{equation}
Although the matrix~$K$ has~$18$ elements in it, 
we only set~$8$ of them to non-zero values, and they
will be used in Theorem~\ref{tm:stabGains}  
to set all closed-loop eigenvalues
equal to desired values.
The non-zero gains we use 
are~$k_1$, $k_4$, $k_5$, $k_8$, $k_{10}$, $k_{11}$, $k_{15}$, and~$k_{18}$,
and all other gains are set to zero. 
We set~$k_5 = 2w$ and~$k_{10} = -2w$. 
%For ease of indexing we relabel the other non-zero  gains as~$c_1 := k_1$, $c_2 := k_4$, $c_3 := k_8$, $c_4 := k_{11}$, $c_5 := k_{15}$, $c_6 := k_{18}$. 
% Then
%     \begin{equation} \label{eq:A_stabapprox}
%         A_{\n{stab}} := 
%         \left[
%         \begin{array}{cccccc}
%             0 & 0 & 0 & 1 & 0 & 0 \\
%             0 & 0 & 0 & 0 & 1 & 0 \\
%             0 & 0 & 0 & 0 & 0 & 1 \\
%             3w^2 - k_{1} & 0 & 0 & -k_4 & 0 & 0 \\
%             0 & -k_8 & 0 & 0 & -k_{11} & 0 \\
%             0 & 0 & -w^2 - k_{15} & 0 & 0 & -k_{18} \\
%         \end{array}
%         \right].
% \end{equation}
%

Next we choose~$k_1$, $k_4$, $k_8$, $k_{11}$, $k_{15}$, and~$k_{18}$ so that~$\textnormal{eig}\big(A_{\n{stab}}\big) = \Lambda_{\n{des}}$. 
We can do so because the CW equations are completely controllable. 
While standard pole placement
techniques could be used to numerically assign gains in~$K$, we 
use fewer gains to provide a closed-form relationship between the gains and the desired
closed-loop eigenvalues. We use this relationship in our later analysis to quantify how
the choices of the closed-loop eigenvalues affect the 
chaser satellite's
rate of convergence
to a compact ball about the rendezvous point
with the target.
%That bound can then be used to select eigenvalues in~$\Lambda_{\n{des}}$ to achieve a desired rate of convergence. 

\begin{theorem}[Stabilizing Controller Gains] \label{tm:stabGains}
    Consider~$A_{\n{stab}}$ from~\eqref{eq:A_Stab}, let a multi-set of desired eigenvalues~$\Lambda_{\n{des}} = \{\lambda_1, \lambda_2, \ldots, \lambda_6\}$ be given, and suppose Assumption~\ref{as:UnqeEig} holds.  
    Then setting
    \begin{equation} \label{eq:Kdef}
        K = \left[\begin{array}{cccccc}
            3w^2 + \lambda_1\lambda_2 & 0 &  0 & -\lambda_1 - \lambda_2 & 2w &  0  \\
             0 & \lambda_3\lambda_4 & 0 &  -2w & -\lambda_3 - \lambda_4 & 0 \\
             0 &  0 & -w^2 + \lambda_5\lambda_6 & 0 &  0 & -\lambda_5 - \lambda_6 \\
        \end{array}\right]
    \end{equation}    
    enforces~$\textnormal{eig}(A_{\n{stab}}) = \Lambda_{\n{des}}$. 
\end{theorem}
\begin{proof}
    We first compute the eigenvalues of~$A_{\n{stab}}$ in terms of~$k_1$, $k_4$, $k_8$, $k_{11}$, $k_{15}$, and~$k_{18}$. 
    The matrix~$\lambda I_6 - A_{\n{stab}}$ has a block~$2 \times 2$ structure, namely 
    \begin{equation}
        \lambda I_6 - A_{\n{stab}} = \left(\begin{array}{cc} 
                                                    \lambda I_3 & -I_3 \\
                                                    \textnormal{diag}(k_1 - 3w^2, k_8, w^2 + k_{15}) & \textnormal{diag}(\lambda + k_4, \lambda + k_{11}, \lambda + k_{18})
                                                     \end{array}\right). 
    \end{equation}
    Using Schur's formula~\cite[Fact 6.5.28]{bernsteinbook}, we have
    \begin{multline}
        \det(\lambda I_6 - A_{\n{stab}}) = \det(\lambda I_3)\det\Big(\textnormal{diag}(\lambda + k_4, \lambda + k_{11}, \lambda + k_{18}) \\
         + \textnormal{diag}(k_1 - 3w^2, k_8, w^2 + k_{15})\big(\lambda I_3\big)^{-1}\Big).
    \end{multline}
    Expanding gives
    \begin{multline}
        \det(\lambda I_6 - A_{\n{stab}}) 
        %&= \lambda^3 \det\Big(\textnormal{diag}(\lambda + k_4, \lambda + k_{11}, \lambda + k_{18}) + \frac{1}{\lambda}\textnormal{diag}(k_1 - 3w^2, k_8, w^2 + k_{15})\Big) \\
           = \\ 
           \det\Big(\textnormal{diag}(\lambda^2 + k_4\lambda, \lambda^2 + k_{11}\lambda, \lambda^2 + k_{18}\lambda) 
         + \textnormal{diag}(k_1 - 3w^2, k_8, w^2 + k_{15})\Big), 
         %&= \det\big(\textnormal{diag}(\lambda^2 - \lambda c_2 - c_1,
         %                              \lambda^2 - \lambda c_4 - c_3,
         %                              \lambda^2 - \lambda c_6 - c_5)\big),
    \end{multline}
    where we have used the fact that~$\det(cM) = c^3\det(M)$
    for a scalar~$c$ and~$3 \times 3$ matrix~$M$. 
    Then the eigenvalues
    of~$A_{\n{stab}}$ are the solutions to the equation    
    \begin{equation}
        \big(\lambda^2 + k_4\lambda + k_1 - 3w^2\big)\big(\lambda^2 + k_{11}\lambda + k_8\big)\big(\lambda^2 + k_{18}\lambda + k_{15} + w^2\big) = 0. 
    \end{equation}

    Solving this equation gives the eigenvalues of~$A_{\n{stab}}$ in pairs as the following: 
    %Suppose $A_{\n{stab}}$ was approximated to~\eqref{eq:A_stabapprox}, we can take the eigenvalues of the matrix which give the following solution in terms of the gains $c_i$:
    %\begin{equation}
        $\lambda_{1,2} = -\frac{1}{2}k_4 \pm \frac{1}{2}\sqrt{k_4^2 - 4(k_1 - 3w^2)}$, 
        $\lambda_{3,4} = -\frac{1}{2}k_{11} \pm \frac{1}{2}\sqrt{k_{11}^2 - 4k_8}$, 
        and
        $\lambda_{5,6} = -\frac{1}{2}k_{18} \pm \frac{1}{2}\sqrt{k_{18}^2 - 4(w^2 + k_{15})}$.
    %\end{equation}
    Starting with $\lambda_{1}$ we see that~$\lambda_{1} = -\frac{1}{2}k_4 + \frac{1}{2}\sqrt{k_4^2 - 4(k_1 - 3w^2)}$,
    and solving for~$k_1$ gives
        %\frac{1}{2}\sqrt{c_2^2 + 4c_1} &= \frac{1}{2}c_2 -  \lambda_{2}\\
        %\sqrt{c_2^2 + 4c_1} &= c_2 -  2\lambda_{2}\\
        %4c_1 &= (c_2 -  2\lambda_{2})^2 - c_{2}^2\\
        %&= c_{2}^2 -4c_{2}\lambda_{2} + 4\lambda_{2} - c_{2}^2\\
    $k_{1} = 3w^2 - \lambda_{1}^2 - k_4\lambda_{1}$.   
    We plug this expression into the expression for $\lambda_{2}$ to find
    %\begin{equation}
        $\lambda_{2} %&= \frac{1}{2}c_2 + \frac{1}{2}\sqrt{c_2^2 + 4c_1}\\
        = -\frac{1}{2}k_4 - \frac{1}{2}\sqrt{k_4^2 - 4(k_1 - 3w^2)} 
        = -\frac{1}{2}k_4 - \frac{1}{2}\big(2\lambda_1 + k_4\big)$. 
    %\end{equation}
    Then~$k_4  = -\lambda_{1} - \lambda_{2}$.
    %Since we've solved for $c_1$ and $c_2$ as a function of $\lambda_1$ and $\lambda_2$ then we can apply a similar procedure for the following gains. This allows us to implement a form of pole placement where for the desired eigenvalues we can come up with some gains following some assumptions the restrict the form of the matrix and the choice of eigenvalues.
    Solving for~$k_8$, $k_{11}$, $k_{15}$, and~$k_{18}$ proceeds similarly. 
\end{proof}

\section{Hybrid Modeling of Feedback Optimization} \label{sec:hybridModel}
In this section, 
we develop a hybrid model of feedback optimization for the satellite rendezvous problem,
and we show that it is well-posed
and that maximal solutions are both complete and non-Zeno.
We first formulate the system model. 
As stated in Section~\ref{ss:fobackground}, feedback optimization is used to optimize the steady-state
behavior of a system, and systems are typically assumed to be asymptotically stable 
when using feedback optimization. 
We will therefore implement feedback optimization for the
stabilized CW equations 
%From \eqref{eq:CWltisystem} we were able to define our stabilized CW dynamics as
\begin{subequations} \label{eq:cw_stab}
\begin{align} 
    \dot{\mathbf{x}} &= A_{\n{stab}}\mathbf{x} + B_{\n{stab}}\mathbf{u} - B_{\n{stab}}K\mathbf{d} \label{eq:cw_stab_x} \\
    \mathbf{y} &= \mathbf{x} + \mathbf{d}, \label{eq:cw_stab_y} 
\end{align}
\end{subequations}
where we suppose that desired eigenvalues in~$\Lambda_{des}$ have been specified, 
$K$ is from~\eqref{eq:Kdef}, and~$\mathbf{d}$ is an unknown time-varying disturbance. 
At each time~$t$, we seek to drive the input and output of this system 
to the solution to
% \begin{subequations} \label{eq:fo_genform_u}
%     \begin{align}
%         \min_{\mathbf{u},\mathbf{y}} \quad & \Phi\big(\mathbf{u},\mathbf{y}\big) \label{eq:phidef_u}\\
%         \text{subject to} \quad & \mathbf{y} = H_{\n{stab}}\mathbf{u} + \mathbf{d}, \,\, \mathbf{u} \in \mathcal{U}, \,\, \mathbf{y} \in \R^p, \label{eq:ddef_u} 
%     \end{align}
% \end{subequations}
\begin{mini!}|l|
    {\mathbf{u},\mathbf{y}}
    {\Phi\big(\mathbf{u},\mathbf{y}\big)}
    {\label{eq:fo_genform_u}}
    {\label{eq:phidef_u}}
    \addConstraint{\mathbf{y} = H_{\n{stab}}\mathbf{u} + \mathbf{d}, \,\, \mathbf{u} \in \mathcal{U}, \,\, \mathbf{y} \in \R^6,}{\label{eq:ddef_u} }
\end{mini!}
where~$H_{\n{stab}} := -A_{\n{stab}}^{-1}B_{\n{stab}}$,
and the constraint~$\mathbf{y} = H_{\n{stab}}\mathbf{u} + \mathbf{d}$ is 
equal to the steady-state input-to-output mapping associated with
the system in~\eqref{eq:cw_stab_x}-\eqref{eq:cw_stab_y} with
values perturbed by~$\mathbf{d}$. 
We consider objective functions~$\Phi$ of the form 
\begin{equation}\label{eq:quadOBJ} 
        \Phi(\mathbf{u},\mathbf{y}_{s}) = \frac{1}{2}\mathbf{u}^\top Q_{\mathbf{u}}\mathbf{u} + \frac{1}{2}(\mathbf{y}_{s} - \hat{\mathbf{y}})^\top Q_{\mathbf{y}}(\mathbf{y}_s - \hat{\mathbf{y}}), 
\end{equation}
where $Q_{\mathbf{u}} \in \R^{3 \times 3}$ and $Q_{\mathbf{y}} \in \R^{6 \times 6}$ are symmetric and positive definite, and $\hat{\mathbf{y}} \in \mathbb{R}^{6}$ is the nominal point at which
the chaser will rendezvous with the target. 
\blue{The chaser should stop moving once it has done so, and therefore
we set~$\hat{y}_4 = \hat{y}_5 = \hat{y}_6 = 0$.} 
% Setting~$\hat{\mathbf{y}} = 0$ may drive the chaser to collide
% with the target, and therefore we consider~$\hat{\mathbf{y}} \neq 0$
% so that the chaser approaches the target without colliding. 

\subsection{Hybrid Modeling of Flow and Jump Sets}
We now develop a hybrid model of 
feedback optimization implemented for
the chaser's dynamics in~\eqref{eq:cw_stab} with discrete-time gradient descent in the loop. 
Some of these modeling developments are related to those in prior work~\cite{chuy2025hybridsystemsmodelfeedback} by a subset of the authors
of the current paper. The current paper goes beyond~\cite{chuy2025hybridsystemsmodelfeedback}
because it implements a pre-stabilizing feedback controller design, propagates output disturbances through the dynamics as a result of that feedback controller, and 
quantifies the impact of 
these disturbances on satellite rendezvous. 

We begin by defining the state vector and flow and jump sets for the model of the chaser. 
For feedback optimization, the state~$\mathbf{x} \in \mathbb{R}^6$ evolves in continuous time according to~\eqref{eq:cw_stab_x},
and~$\mathbf{x}$ is a state in the hybrid model we develop. 
Samples of the output are taken over time, and each sample is fed into an optimization algorithm
that computes new values of~$\mathbf{u} \in \mathbb{R}^3$. We use~$\mathbf{y}_s \in \mathbb{R}^6$ to denote the sampled output,
and both~$\mathbf{y}_s$ and~$\mathbf{u}$ are states of the hybrid system as well. 
The value of~$\mathbf{y}_s$ is plugged into a projected 
gradient descent algorithm of the form of~\eqref{eq:gdUpdateLaw}
that runs in the loop and computes iterates that progress
toward the optimal input. 
We denote the current gradient descent iterate by~$\mathbf{z} \in \mathbb{R}^3$, which
is also a state of the hybrid system. 
To track the amount of continuous time required to compute
an iterate, we introduce a timer~$\ctimer{} \in \mathbb{R}$ that 
tracks the amount of time remaining until the next gradient descent iteration
is completed. It
counts down with unit rate, and an iteration of gradient descent
has been completed when it reaches zero. 
Between the times at which samples~$\mathbf{y}_s$ are taken,
the state~$\mathbf{x}$ continues flowing with the most recently applied
input held constant. 
Therefore, we introduce a timer~$\atimer{} \in \mathbb{R}$ that tracks the amount of time
remaining until the input to the system~$\mathbf{u}$ is changed.
This timer also counts down with unit rate, and the input~$\mathbf{u}$ is set
equal to the most recent optimization iterate~$\mathbf{z}$ when~$\atimer{}$ reaches zero. 
Additionally, to convert the time-dependent disturbance~$t \mapsto \mathbf{d}(t)$
into a state-dependent disturbance, 
we introduce a timer~$\dtimer\in \R$ that begins
at zero and counts up with unit rate. 
Each timer is a state of the hybrid model we develop.

We assemble the states of the hybrid system into the vector
%The states of the hybrid system include~$\mathbf{x} \in \R^6$,~$\mathbf{u} \in \R^3$ and~$\mathbf{y}_s \in \R^6$ which are, respectively, the state, input, and value of the most recent he sampled output of the LTI system in~\eqref{eq:CWltisystem}.  The full state is defined as
\begin{equation} \label{eq:zeta}
    \zeta := 
    \left(\begin{array}{ccccccc}\!
    \mathbf{x}^{\top} &
    \mathbf{u}^{\top} &
    \mathbf{y}_s^{\top} &
    \mathbf{z}^{\top} &
    \atimer{} &
    \ctimer{} &
    \dtimer
    \end{array}\!\right)^{\top}
    \in \mathcal{X} := \R^{21}.
\end{equation}
We define~$\atimer{,\max}$ as the maximum amount of time that can elapse between two consecutive
changes in the chaser's input, and 
it defines the control cadence of the chaser satellite. 
We define~$\ctimer{,\n{comp}}$ as the amount of time required
to complete one iteration of gradient descent. A jump occurs only when 
$\ctimer{}$ or $\atimer{}$ 
has reached zero. Otherwise the system's states continue flowing. 
%They dictate when the states jump as encompassed by certain criteria which are defined in the flow set~$C$ and jump set~$D$.
Then we define the flow and jump sets as
\begin{align} 
        C &:= \big\{\zeta \in \mathcal{X} \mid \atimer{} \in [0,\atimer{,\max}], \ctimer{} \in [0,\ctimer{,\n{comp}}]\big\} \label{eq:Cdef} \\
        D &:= \big\{\zeta \in \mathcal{X} \mid \atimer{} = 0 \textnormal{ or } \ctimer{} = 0\big\}, \label{eq:Ddef}
\end{align}
respectively. 
%where $\atimer{,max} > 0$ is the maximum amount of time between changes in the  input and $\ctimer{,\n{comp}} > 0$ is the allotted amount of time needed to perform a single gradient descent iteration.

\subsection{Flow Map Definition}
The flow map models the continuous-time evolution of~$\zeta$. 
The state~$\mathbf{x}$ flows according to the 
stabilized CW dynamics in~\eqref{eq:cw_stab_x},
$\ctimer{}$ and~$\atimer{}$ count down with unit rate, and 
$\dtimer$ counts up with unit rate. 
Other states only change at jump times,
and the flow map is 
\begin{equation} \label{eq:fwmap}
    F(\zeta) :=
    \left(\begin{array}{c}
    A_{\n{stab}}\mathbf{x} + B_{\n{stab}}\mathbf{u} - B_{\n{stab}}K\mathbf{d}(\dtimer) \\
     \mathbf{0} \\
     \mathbf{0} \\
     \mathbf{0} \\
    {-1} \\
    {-1} \\
    {1}
    \end{array}\right)
    \quad \blue{\textnormal{ for all } \zeta \in C}.
\end{equation}

\subsection{Jump Map Definition} \label{ss:jump}
    The jump map has three cases as defined in \eqref{eq:Ddef}: (i) $\ctimer{} = 0$ with~$\atimer{} > 0$, 
    (ii) $\atimer{} = 0$ with~$\ctimer{} > 0$, and (iii) $\atimer{} = \ctimer{} = 0$. 

    In case~(i), the hybrid system completes a single gradient descent step of the form~$\mathbf{z}^+ = \Pi_{\mathcal{U}}\big[\mathbf{z} - \gamma(\nabla_{\mathbf{u}} \Phi(\mathbf{z}, \mathbf{y}_s) + H^T\nabla_{\mathbf{y}}\Phi(\mathbf{z}, \mathbf{y}_s))\big]$, 
    where~$\nabla_{\mathbf{u}}\Phi$ denotes the derivative of~$\Phi$ with respect to its first argument,
    $\nabla_{\mathbf{y}}\Phi$ denotes the derivative of~$\Phi$ with respect to its second argument,
    and~$\gamma > 0$ is a stepsize. 
    Since~$\atimer{} > 0$, that new iterate is not applied as the input to the system.
    The timer $\ctimer{}$ resets to $\ctimer{,\n{comp}}$, and all other states remain unchanged.
    Then, the jump map for this case is  
    \begin{equation} \label{eq:g1def}
       G_1(\zeta) :=
                \left(\begin{array}{c}
                \mathbf{x} \\
                \mathbf{u} \\
                \mathbf{y}_s \\
                \Pi_{\mathcal{U}}\Big[\mathbf{z} - \gamma\big(\nabla_{\mathbf{u}} \Phi\big(\mathbf{z}, \mathbf{y}_s\big) + H^T\nabla_{\mathbf{y}}\Phi\big(\mathbf{z}, \mathbf{y}_s\big)\big)\Big] \\
                \atimer{} \\
                \ctimer{,\n{comp}} \\
                \dtimer
                \end{array} \right) \quad \blue{\textnormal{ for all } \zeta \in D_1},
    \end{equation}
    where~$D_1 := \{\zeta\in \mathcal{X}: \ctimer{} = 0\}$. 
    
    During case (ii), when~$\atimer{} = 0$ and~$\ctimer{} > 0$, 
    the input~$\mathbf{u}$ is set equal to the most recent optimization iterate~$\mathbf{z}$, 
    a new output~$\mathbf{y}_s$ is sampled, and the timer~$\atimer{}$ resets to some number in the interval~$[\atimer{,\min},\atimer{,\max}]$, where~$0 < \atimer{,\min} \leq  \atimer{,\max}$. 
    This interval is used to model indeterminacy in the amount of time that elapses between changes in the
    input. 
    All other states remain unchanged, and the jump map for this case is
    \begin{equation} \label{eq:g2def}
    G_2(\zeta) :=
        \left(\begin{array}{c}
                \mathbf{x} \\
                \mathbf{z} \\
                H_{\n{stab}}\mathbf{u} + \mathbf{d}(\dtimer) \\
                \mathbf{z} \\
                {[\atimer{,\min}, \atimer{,\max}]} \\
                \ctimer{}\\
                \dtimer
        \end{array} \right) \quad \blue{\textnormal{ for all } \zeta \in D_2},
    \end{equation}
    where $D_2 := \{\zeta\in \mathcal{X}: \atimer{} = 0\}$, and, as is standard in feedback optimization, we approximate the output~$\mathbf{y}_{s} = \Psi \mathbf{x} + \mathbf{d}$ 
    as coming from the perturbed steady-state mapping
    \begin{equation} \label{eq:hstab_approx}
        \mathbf{y}_{s} = H_{\n{stab}}\mathbf{u} + \mathbf{d},
    \end{equation}
    where~$H_{\n{stab}}$ is from~\eqref{eq:ddef_u}. 
    The mapping~$G_2$ has a set-valued element in the form of~$[\atimer{,\min}, \atimer{,\max}]$,
    which allows~$\atimer{}$ to reset to any value in this interval \blue{and
    models indeterminacy in the amount of time that elapses between changes in the input.} 
        
    In case (iii), where both $\atimer{} = 0$ and $\ctimer{} = 0$, we combine cases~(i) and~(ii), and 
    the system executes either~$G_1$ then~$G_2$ or~$G_2$ then~$G_1$. 
    The full jump map~$G$ is 
    \begin{equation} \label{eq:jpmap}
            \zeta^+ \in G(\zeta) := 
            \begin{cases}
                G_1(\zeta)
                \quad \text{if }
                \atimer{} > 0 \textnormal{ and } \ctimer{} = 0 \quad \textnormal{ Case (i)} \\
                G_2(\zeta)
                \quad \text{if }
                \atimer{} = 0 \textnormal{ and } \ctimer{} > 0 \quad \textnormal{ Case (ii)} \\
                G_3(\zeta) 
                \quad \text{if }
                \atimer{} = 0 \textnormal{ and } \ctimer{} = 0 \quad \textnormal{ Case (iii)}, 
            \end{cases}
    \end{equation}
    where $G_3(\zeta) := G_1(\zeta) \cup G_2(\zeta)$. 
    % The jump set~$D$ is equal to 
    % \begin{equation} \label{eq:jpset}
    %     D := D_1\cup D_2,
    % \end{equation}
    % where~$D_1 := \{\zeta\in \mathcal{X}: \ctimer{} = 0\}$ and $D_2 := \{\zeta\in \mathcal{X}: \atimer{} = 0\}$. 
    Then, the hybrid model of feedback optimization onboard
    the chaser satellite is 
    \begin{equation}\label{eq:hybridFO}
        \HFO:= (C,F,D,G),
    \end{equation}
    where $C$ is from \eqref{eq:Cdef}, $F$ is from \eqref{eq:fwmap}, $D$ is from \eqref{eq:Ddef}, and $G$ is from \eqref{eq:jpmap}.
    We illustrate this hybrid model in Figure~\ref{fig:HFOLogic}.

    \begin{figure}
    \centering
    \includegraphics[width=0.9\linewidth]{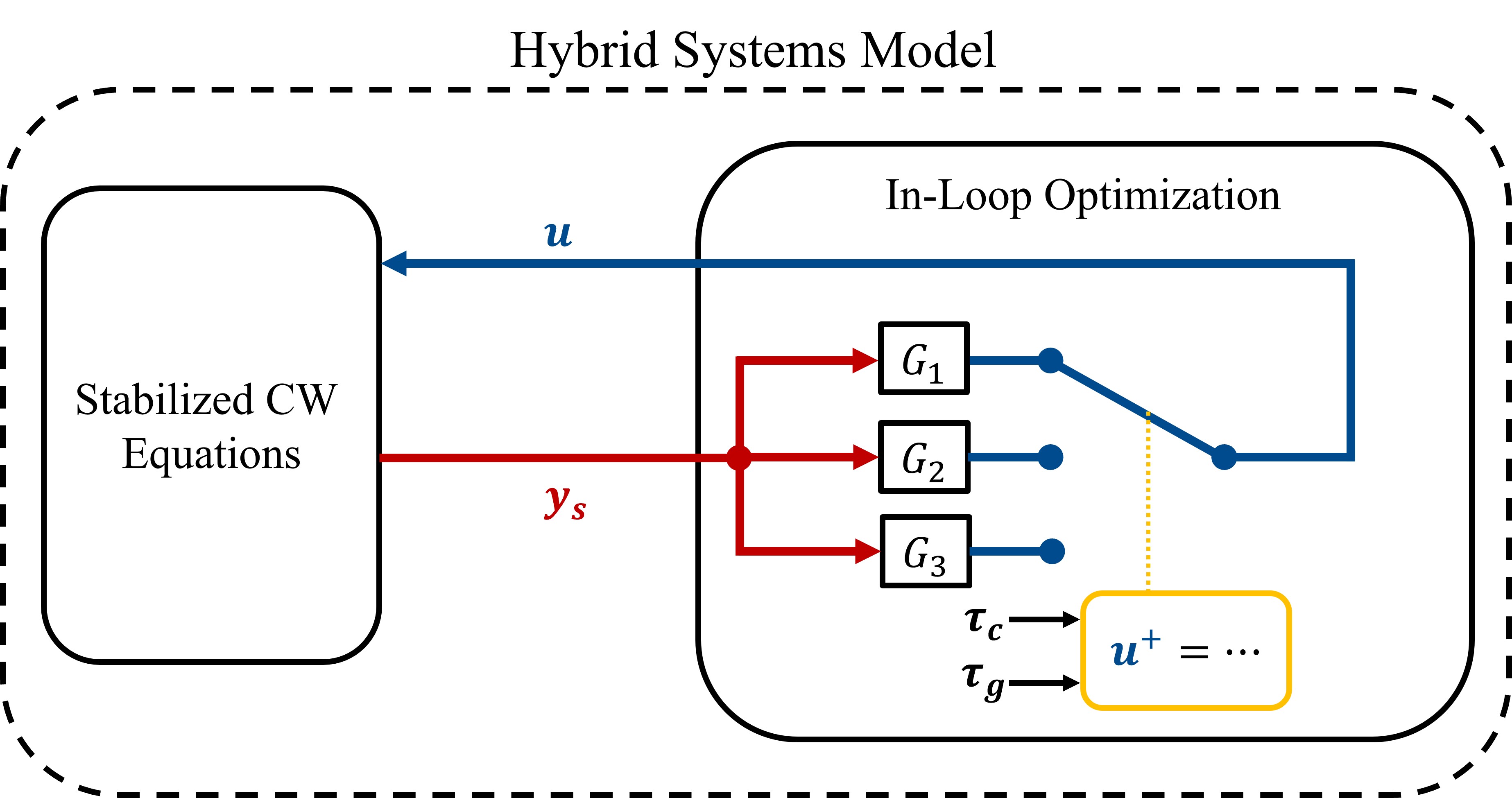}
    \caption{Diagram of the hybrid system model of feedback optimization 
    $\HFO$ from~\eqref{eq:hybridFO} 
    with the stabilized CW equations and in-the-loop optimization. The optimization 
    algorithm receives the sampled output~$\mathbf{y}_s$ and 
    computes the next input~$\mathbf{u}$ 
    with timing determined by the timers
    $\atimer{}$ and $\ctimer{}$ 
    reaching zero.}
    \label{fig:HFOLogic}
\end{figure}

\subsection{Intermittent Sampling and Optimization in \texorpdfstring{$\HFO$}{}}
\blue{We impose the following assumption on the relative rates of computation
of inputs and changes in the values of inputs applied to the chaser.}

\begin{assumption} \label{as:timescale}
\blue{        There exists~$\ell \in \mathbb{N}$ with~$\ell \geq 1$
        such that~$\ell\ctimer{,\n{comp}} \leq \atimer{,\min}$. }
    \end{assumption}

\blue{
    Assumption \ref{as:timescale} implies that there are at least $\ell$ gradient descent iterations, i.e., case (i) jumps, performed between consecutive changes in the 
    chaser's 
    input. Mathematically, it ensures that the optimization algorithm
    takes at least one step towards the
    optimal input before each input is applied to the chaser.
    The value of~$\ell$ is determined by a chaser satellite's onboard processor and will vary across
    implementations. Without this assumption, a satellite could never perform
    a computation and never change its input. 
    }

To analyze 
the relationship between sampling and optimization in~$\HFO$, 
\blue{consider an initial 
condition~$\phi(0,0) = \nu \in \R^{21}$ that satisfies
\begin{equation} \label{eq:initconds}
\atimer{}(0,0) \in [\atimer{,\min},\atimer{,\max}],~\ctimer{}(0, 0) = \tau_{g,\n{comp}},~\dtimer(0,0) = 0,~\textnormal{and}~\mathbf{z}_{0}\Hotime{0}{0} = \mathbf{u}\Hotime{0}{0}, 
\end{equation}
and consider a solution~$\phi$ to~$\mathcal{H}_{FO}$ from such initial condition.}
%We use~$\alpha(i)$ to denote the number of gradient descent iterations that are performed when computing the~$(i+1)^{\n{th}}$ value of the input~$\mathbf{u}$. 
The behavior of sampling outputs and optimizing inputs is as follows. 
Starting from hybrid time~$(0, 0)$ the input~$\mathbf{u}(0, 0)$
is applied to the system and held constant. 
While that input is applied,
the system performs 
some number
of gradient descent steps that we denote by~$\alpha(0)$
(which are $\alpha(0)$ case (i) jumps) before a new input is applied. 
Each case (i) jump is triggered by~$\ctimer{}$ reaching zero.
The new input
is applied when a case (ii) jump occurs, 
\blue{which} is triggered by~$\atimer{}$ reaching zero.
\blue{This case (ii) jump is the~$\big(\alpha(0) + 1\big)^{th}$ jump,} 
% occurs at hybrid  time~\blue{$\big(t_{\alpha(0)}, \alpha(0)\big)$}
and it changes the value of the input to
$\mathbf{u}\Hotime{t_{\alpha(0) + 1}}{\alpha(0) + 1}$, 
where the value~$\alpha(0) + 1$ 
has accounted for the~$\alpha(0)$ case (i) jumps
and the one case (ii) jump that have occurred.
When computing the value of the 
input~$\mathbf{u}\Hotime{t_{\alpha(0) + 1}}{\alpha(0) + 1}$, 
% the final iterate that is computed is denoted~$z_{\alpha(0)}\Hotime{t_{\alpha(0)}}{\alpha(0)}$, and
the~$k^{\n{th}}$ iterate 
for any~$k \in \{0, \ldots, \alpha(0)\}$ 
is denoted $\mathbf{z}_{k}\Htime{k}{k}$. 
%This notation allows for the ease in tracking how many iterations are completed when computing each input. \blue{For all other values of $\mathbf{z}$, it is held constant between iterations.}
When computing $\mathbf{u}\Htime{\alpha(0) + 1}{\alpha(0) + 1}$, we denote the $\alpha(0)^{\n{th}}$ iterate (which
is the last iterate) by $\mathbf{z}_{\alpha(0)}\Htime{\alpha(0)}{\alpha(0)}$.
\blue{The first two changes in the input are shown in Figure~\ref{fig:hybridtime},
\green{where~$\atimer{}$ is shown jumping to different values when it jumps
because it can jump to any value in the interval~$[\atimer{,\min}, \atimer{,\max}]$}.
}

\begin{figure}[h]
    \centering
    \includegraphics[width=1\linewidth]{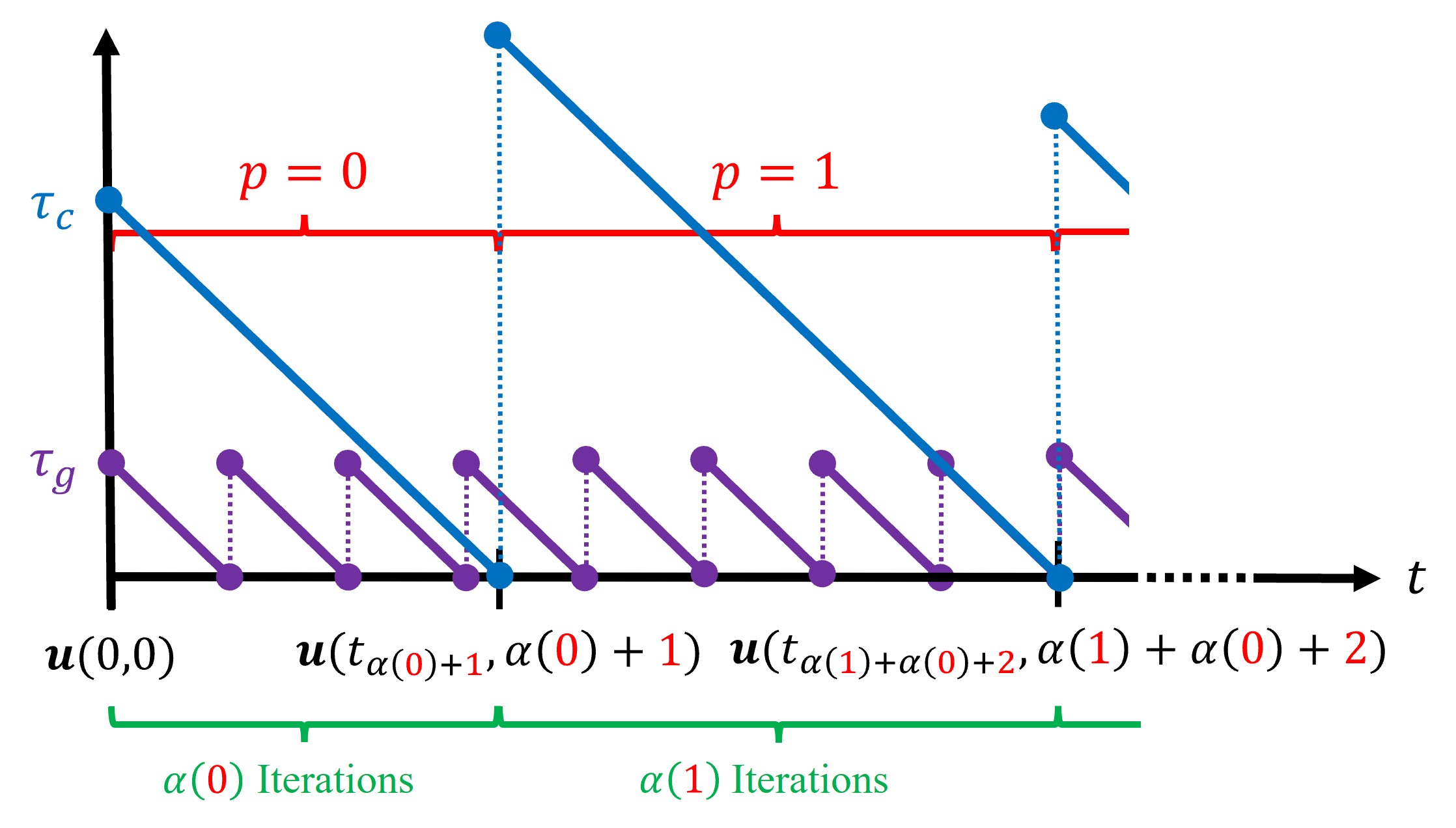}
    \caption{\blue{A visual representation of the evolution of 
    the input~$\mathbf{u}$. There are $\alpha(p)$ iterations of gradient descent when computing the $(p+1)^{\textnormal{th}}$ input. 
    The first two changes in the input occur at the hybrid times~$(t_{\alpha(0) + 1}, \alpha(0) + 1)$
    and~$(t_{\alpha(0) + \alpha(1) + 2}, \alpha(0) + \alpha(1) + 2)$. 
    The timer~$\atimer{}$ reaches zero for the second time at the same
    time that~$\ctimer{}$ reaches zero. This occurrence can trigger either a case~(i)
    jump followed by a case~(ii) jump or a case~(ii) jump followed by a case~(i) jump. 
    } 
    %If~$\atimer{}$ and $\ctimer{}$ reach zero,~$\mathbf{u}\Htime{\alpha(0)+p}{\alpha(0)+p}$ is able to account whether~$G_1$ or~$G_2$ is first performed by the change in the $\alpha(0)$ and $\alpha(1)$ values. 
    }
    \label{fig:hybridtime}
\end{figure}

When~$\atimer{}$ reaches zero for the first time, 
several operations occur in addition to the input changing,
and they are modeled in the jump map~$G_2$ in~\eqref{eq:g2def}.
%the case~(ii) jump occurs \blue{which is reflected in the hybrid time $\Htime{\alpha(0) + 1}{\alpha(0) + 1}$}. 
%\blue{The following actions are have been defined in \eqref{eq:g2def}:} 
The output $\mathbf{y}_s\Htime{\alpha(0) + 1}{\alpha(0) + 1}$ is sampled, and
we approximate it as coming from the perturbed
steady-state mapping in~\eqref{eq:hstab_approx}. 
The input~$\mathbf{u}$ is set equal to the most recent optimization iterate, i.e., 
$\mathbf{u}\Htime{\alpha(0) + 1}{\alpha(0) + 1} = \mathbf{z}_{\alpha(0)}\Htime{\alpha(0)}{\alpha(0)}$, and
that iterate is also used as the initial iterate when computing the next input
so that $\mathbf{z}_{0}\Htime{\alpha(0)  + 1}{\alpha(0) + 1} = \mathbf{z}_{\alpha(0)}\Htime{\alpha(0)}{\alpha(0)}$. 

We use~$\alpha(p)$ for~$p \in \mathbb{N}$ to denote
the number of gradient descent iterations
that are generated when computing the~$(p+1)^{\n{th}}$ input
to the system. 
We define $\alphatime{p}$ for~$p \in \mathbb{N}$ 
as the total number of gradient descent iterations that have been completed 
for computing any input 
up until the $p^{\n{th}}$ jump in $\mathbf{u}$. 
Then~$\alphatime{0} = 0$ and $\alphatime{p} = \sum^{p-1}_{i=0} \alpha(i)$.
Figure~\ref{fig:hybridtime} provides an illustration of how inputs change over time. 
%Figure \ref{fig:hybridtime} illustrates the changes in the input and their relationship to hybrid time. 

We can iterate the above analysis to find a general expression for the~$(k+1)^{\n{th}}$ iterate
when finding the~$(p+1)^{\n{th}}$ input to the system. That iterate is computed after
the~$p^{\n{th}}$ case~(ii) jump 
has occurred 
and after~$k+1$ case~(i) jumps have occurred after
that case~(ii) jump. 
It takes the form 
%We can generalize an arbitrary $k+1$ case (i) jumps with $p$ case (ii) jumps such that the optimization iterate takes the form
\begin{multline} \label{eq:bigiteration}
    \mathbf{z}_{k+1}\Htime{\alphatime{p} + p + k + 1}{\alphatime{p} + p + k + 1} = \\\Pi_{\mathcal{U}}\big[\mathbf{z}_{k}\Htime{\alphatime{p} + p + k}{\alphatime{p} + p + k} 
    - \gamma\big(Q_{\mathbf{u}}\mathbf{z}_{k}\Htime{\alphatime{p} + p + k}{\alphatime{p} + p + k} 
     \\ + H^{\top}Q_{\mathbf{y}}\big(\mathbf{y}_s\Htime{\alphatime{p} + p + k}{\alphatime{p} + p + k} - \hat{\mathbf{y}} \big)\big)\big]. 
\end{multline}

%where for ease of notation we have used
%\begin{multline}
%    \nabla_{\mathbf{u}}\Phi\big({\mathbf{z}}_{k}\Htime{\alphatime{p} + p + k}{\alphatime{p} + p + k}\big) = \\
%    Q_{\mathbf{u}}\mathbf{z}_{k}\Htime{\alphatime{p} + p + k}{\alphatime{p} + p + k} 
%     + H^{\top}Q_{\mathbf{y}}\big(\mathbf{y}_s\Htime{\alphatime{p} + p + k}{\alphatime{p} + p + k} - \hat{\mathbf{y}} \big).
%\end{multline}
The hybrid time~$\Htime{\alphatime{p} + p + k}{\alphatime{p} + p + k}$
accounts for~$\bar{\alpha}(p)$ total gradient descent
iterations that have been computed up to the~$p^{\n{th}}$
change in the input, $p$ changes in the input itself,
and~$k$ gradient descent iterations that have been computed so far
in the course of computing the~$(p+1)^{\n{th}}$ input. 

The $(p+1)^{\n{th}}$ input to the system is 
\begin{equation} \label{eq:c2inptiterate}
    \mathbf{u}\Htime{\alphatime{p + 1} + p + 1}{\alphatime{p + 1} + p + 1} =
    \mathbf{z}_{\alpha(p)}\Htime{\alphatime{p} + \alpha(p) + p}{\alphatime{p} + \alpha(p) + p},
\end{equation}
i.e., it is set equal to the most recently
computed iterate from the optimization algorithm
at the time that the change in the input occurs. 
When computing iterates for
the~$(p+2)^{\n{th}}$ input, 
the initial iterate is denoted 
\begin{equation}
\mathbf{z}_{0}\Htime{\alphatime{p + 1} + p + 1}{\alphatime{p + 1} + p + 1}. 
\end{equation}
%which equals $\mathbf{u}\Htime{\alphatime{p + 1} + p + 1}{\alphatime{p + 1} + p + 1}$.
That is, the initial iterate when computing the next input is equal to the final iterate
that was computed when finding the previous input.

\subsection{Basic Properties of the Hybrid Model of Feedback Optimization}
% \blue{[Well-Posed, Maximal Solutions are Complete: Move to a Later section]}
We impose the following assumption on the disturbance~$\mathbf{d}$.

\begin{assumption} \label{as:d}
\blue{There exists a non-empty compact set~$\mathcal{D} \subseteq \mathbb{R}^6$
such that~$\mathbf{d}(t) \in \mathcal{D}$ for all~$t \geq 0$.
The mapping~$t \mapsto \mathbf{d}(t)$ is differentiable 
and there exists some~$\bar{\mathbf{d}} \geq 0$ such
that~$\big\|\dot{\mathbf{d}}(t)\big\| \leq \bar{\mathbf{d}}$ for all~$t \geq 0$.} 
\end{assumption}

The following results show that solutions
to~$\HFO$ exist over arbitrarily long hybrid time domains. 

\begin{lemma} \label{lem:wellPosed}
    \blue{Consider the system~$\HFO$ in~\eqref{eq:hybridFO} 
    and suppose that Assumption~\ref{as:d} holds. Then~$\HFO$
    is well-posed in the sense that it satisfies Definition \ref{def:hybridcond}.}
\end{lemma}
\begin{proof}
    % See Lemma~1 in \cite{chuy2025hybridsystemsmodelfeedback}. 
    See \ref{Ap:wellPosed}
\end{proof}

Lemma \ref{lem:wellPosed} guarantees well-posedness of the system, which is 
used in the forthcoming results. 
%The model~$\HFO$ in~\eqref{eq:hybridFO} differs from the one studied in~\cite{chuy2025hybridsystemsmodelfeedback} because it includes a feedback controller that propagates disturbances into the state equation. 
Next we establish that solutions to~$\HFO$ exist 
\blue{for hybrid times~$(t, j)$ where~$t$ and~$j$ can grow arbitrarily large.}

\begin{proposition}[Completeness of Maximal Solutions]\label{prop:maxsolComplete}
    Consider the hybrid feedback optimization model $\HFO$ from \eqref{eq:hybridFO}
    \blue{and suppose Assumption~\ref{as:d} holds.}
    Then, from every point in $C\cup D$ there exists a nontrivial solution \blue{to $\HFO$, and
    all maximal solutions are both complete and non-Zeno.} 
\end{proposition}
\begin{proof}
    % See Proposition~2 in \cite{chuy2025hybridsystemsmodelfeedback}. 
    See \ref{Ap:maxsolComplete}
\end{proof}

The non-Zeno property guarantees that the system will keep flowing and 
that the computation and application of inputs will continue indefinitely. 
%This additionally tells us that since $\HFO$ attains the same properties as $\HFO$ \cite[Lemma and Proposition 1]{chuy2025hybridsystemsmodelfeedback} then the conclusions and theorems seen later can follow a similar structure.

\section{Convergence Analysis} \label{sec:convAnalysis}
This section analyzes how the chaser approaches the target when using feedback optimization. 
For simplicity of notation we define~$\xtime{p,k} := \bar{\alpha}(p) + p + k$. Then 
\begin{equation} \label{eq:xtime_def}
\Htime{\xtime{p,k}}{\xtime{p,k}} = \Htime{\alphatime{p} + p + k}{\alphatime{p} + p + k}, 
\end{equation} 
i.e., $\Htime{\xtime{p,k}}{\xtime{p,k}}$ denotes the hybrid time
at which the optimization algorithm has completed
the~$k^{\n{th}}$ iteration when computing the~$(p+1)^{\n{th}}$ input
to the system. 
Then we may write~\eqref{eq:bigiteration} as 
\begin{multline}
    \mathbf{z}_{k+1}\Htime{\xtime{p,k+1}}{\xtime{p,k+1}} =
    \Pi_{\mathcal{U}}\big[\mathbf{z}_{k}\Htime{\xtime{p,k}}{\xtime{p,k}} \\
    - \gamma\big(Q_{\mathbf{u}}\mathbf{z}_{k}\Htime{\xtime{p,k}}{\xtime{p,k}}
    + H^\top Q_{\mathbf{y}}\big(\mathbf{y}_s\Htime{\xtime{p,k}}{\xtime{p,k}} - \hat{\mathbf{y}} \big)\big)\big]. 
\end{multline}
%where 
%%\begin{equation}
%    $\nabla_{\mathbf{u}}\Phi\big(\mathbf{z}_{k}\Htime{\xtime{p,k}}{\xtime{p,k}}\big) = 
%    Q_{\mathbf{u}}\mathbf{z}_{k}\Htime{\xtime{p,k}}{\xtime{p,k}}
%    \\+ H^\top Q_{\mathbf{y}}\big(\mathbf{y}_s\Htime{\xtime{p,k}}{\xtime{p,k}} - \hat{\mathbf{y}} \big)$. 
%%\end{equation}

The following lemma relates successive iterates that are used to compute the inputs to the 
chaser satellite. 
In it, we use the constants
\begin{equation} \label{eq:LqDef}
L := \lambda_{\max}(Q_{\mathbf{u}} + H_{\n{stab}}^{\top}Q_{\mathbf{y}}H_{\n{stab}})
\quad \textnormal{and} \quad q:= 1-2\gamma\lambda_{\min}(Q_u) + \gamma^2L^2 \in (0,1),
\end{equation}
where~$Q_{\mathbf{u}}$ and~$Q_{\mathbf{y}}$
are from~\eqref{eq:quadOBJ}
and~$H_{\n{stab}}$ is from~\eqref{eq:ddef_u}. 

\begin{lemma}[Input Convergence Rate] \label{lem:inptconvrate}
    Consider the hybrid system~$\mathcal{H}_{FO}$ in~\eqref{eq:hybridFO} and the objective in~\eqref{eq:quadOBJ}. 
    Suppose that the gradient descent algorithm uses a stepsize~$\gamma \in \paren{0,\frac{2}{\lambda_{\min}(Q_u) + L}}$, where~$L$ is from~\eqref{eq:LqDef}.  
    Let~$\phi$ denote a maximal solution to~$\mathcal{H}_{FO}$ with initial condition~$\phi(0, 0) = \nu$ that satisfies~\eqref{eq:initconds}.
    For any~$(t, j) \in \textnormal{dom } \phi$, set~$P = \max\{p \in \mathbb{N} : \bar{\alpha}(p) + p \leq j\}$.     
    Then, for any~$p \in \{0, \ldots, P\}$ the state $z$ obeys
    \begin{multline} \label{eq:inputbound_norm}
        \big\|\mathbf{z}_{\alpha({p})}\Htime{\xtime{p,\alpha(p)}}{\xtime{p,\alpha(p)}} - \mathbf{z}^*\Htime{\xtime{p,0}}{\xtime{p,0}}\big\| \\
        \leq q^{\frac{\alpha({p})-1}{2}}\big\|\mathbf{z}_{1}\Htime{\xtime{p,1}}{\xtime{p,1}} - \mathbf{z}^*\Htime{\xtime{p,0}}{\xtime{p,0}}\big\|,
    \end{multline}
    \green{where $q$ is from~\eqref{eq:LqDef} and
    $\mathbf{z}^*\Htime{\xtime{p,0}}{\xtime{p,0}} = \argmin\limits_{\mathbf{u} \in \mathcal{U}} \Phi\big(\mathbf{u}, \mathbf{y}_{s}\Htime{\xtime{p,0}}{\xtime{p,0}}$
    is a function of the hybrid time time~$\Htime{\xtime{p,0}}{\xtime{p,0}}$ because it depends on
    the sampled output at the same time, namely~$\mathbf{y}_{s}\Htime{\xtime{p,0}}{\xtime{p,0}}$. 
    }
\end{lemma}

\begin{proof}
    The steps of the proof follow those of a standard proof in the convex optimization literature for the minimization of a strongly convex function using gradient descent, e.g., \cite[Ch. 1, Thm. 2]{polyak84}.     
\end{proof}

\subsection{Complete Hybrid Convergence}
At time~$t$, let~$\big(\tilde{\mathbf{u}}(t), \tilde{\mathbf{y}}(t)\big)$
denote the solution to~\eqref{eq:fo_genform_u}.
\green{Setting~$\dot{\mathbf{x}} \equiv 0$ in the CW equations and
setting~$\mathbf{u}$ and~$\mathbf{y}$ equal to~$\tilde{\mathbf{u}}(t)$
and~$\tilde{\mathbf{y}}(t)$, respectively, we obtain
%The satisfaction of that input-output relationship happens
%when~$\dot{\mathbf{x}} = 0$. 
%In the CW equations, we set~$\dot{\mathbf{x}} = 0$, $\mathbf{u} = \tilde{\mathbf{u}}(t)$,
%and~$\mathbf{y} = \tilde{\mathbf{y}}(t)$ to find the optimal state value
the state
\begin{equation} \label{eq:xtildedef}
\tilde{\mathbf{x}}(t) = -A_{\n{stab}}^{-1}B_{\n{stab}}\tilde{\mathbf{u}}(t)
+ A_{\n{stab}}^{-1}B_{\n{stab}}K\mathbf{d}(t),
\end{equation}
which simultaneously gives
the CW equations 
the input~$\tilde{\mathbf{u}}(t)$, the output~$\tilde{\mathbf{y}}(t)$,
and the input-output relation~$\tilde{\mathbf{y}}(t) = H_{\n{stab}}\tilde{\mathbf{u}}(t) + \mathbf{d}(t)$
as required by~\eqref{eq:fo_genform_u}. 
}
% is the optimal value of the 
% chaser 
% state~$\mathbf{x}$ at flow time~$t$
% because the condition~$\mathbf{x}(t, j) = \tilde{\mathbf{x}}(t)$ 
% would imply that the CW equations simultaneously
% have input~$\tilde{\mathbf{u}}(t)$, output~$\tilde{\mathbf{y}}(t)$,
% and input-output relation~$\tilde{\mathbf{y}}(t) = H_{\n{stab}}\tilde{\mathbf{u}}(t) + \mathbf{d}(t)$
% as required by~\eqref{eq:fo_genform_u}. 
The value of~$\tilde{\mathbf{x}}(t)$ in general differs from that of
the nominal rendezvous point~$\hat{y}$ 
precisely because~$\tilde{\mathbf{x}}(t)$ accounts for costs on the input
in~$\Phi$ and the disturbance~$\mathbf{d}$. 
The value of~$\tilde{\mathbf{x}}(t)$
can be interpreted as the chaser's \green{chosen} rendezvous point at time~$t$.
This section 
% solves Problem~\ref{prob:stateConverge} by bounding
bounds the distance between $\mathbf{x}(t, j)$ and $\tilde{\mathbf{x}}(t)$
as a function of~$t$. 
% We define $\mathcal{Y} = \{\mathbf{y} \in \mathbb{R}^6 : \mathbf{y} = H_{\n{stab}}\mathbf{u} + \mathbf{d}, \mathbf{u} \in \mathcal{U}, \mathbf{d} \in \mathcal{D}\}$.
% The set~$\mathcal{U}$ is compact by definition 
% and under Assumption~\ref{as:d} the set~$\mathcal{D}$
% is compact. Then, under Assumption~\ref{as:d} the set~$\mathcal{Y}$ is compact  
% as well.

For a solution~$\phi$ to~$\mathcal{H}_{FO}$, 
at each hybrid time~$(t, j) \in \dom \phi$
we can bound the distance 
between~$\phi$ and the set 
\begin{equation} \label{eq:closedsetA}
\mathcal{A}(t) := \big\{\tilde{\mathbf{x}}(t)\big\} \times \mathcal{U} \times \mathbb{R}^6 \times \mathcal{U} \times [0, \atimer{,\max}] \times [0, \ctimer{,\n{comp}}] \times \R,
\end{equation}
where 
% %\begin{equation} \label{eq:r}
%     $r = \mu_{\max}(\Lambda_{\n{des}})\frac{|\lambda_6|}{m_{c}|\lambda_1|^2}\big(d_{\mathcal{U}} + d_{\mathcal{U}}q^{\frac{\ell}{2}} + \big\|A_{\n{stab}}^{-1}\big\|\big\|K\big\|\bar{\mathbf{d}}\big)$,
% %\end{equation}
% $\bar{\mathbf{d}}$ is from Assumption~\ref{as:d}, 
% and
$\tilde{\mathbf{x}}(t)$ is from~\eqref{eq:xtildedef}.
By definition,~$\|\phi(t, j)\|_{\mathcal{A}(t)} = \|\mathbf{x}(t, j) - \tilde{\mathbf{x}}(t)\|$. 
We will bound~$(t, j) \mapsto \|\phi(t, j)\|_{\mathcal{A}(t)}$ for each solution to~$\HFO$, which will characterize the error between
the chaser's position~$\mathbf{x}(t,j)$ and
the \green{chosen} rendezvous point~$\tilde{\mathbf{x}}(t)$, while accounting for the full dynamics of~$\HFO$. 

\begin{proposition}[Convergence of $\HFO$] \label{prop:completeHOCW} 
    Consider the hybrid system $\HFO$ from \eqref{eq:hybridFO}, 
    %where based on~$\Lambda_{\n{des}}$ we implement \eqref{eq:Kdef},
    with the objective function from \eqref{eq:quadOBJ} and stepsize~$\gamma \in \paren{0,\frac{2}{\lambda_{\min}(Q_u) + L}}$,
    where~$Q_{\mathbf{u}}$ is from~\eqref{eq:quadOBJ} and~$L$ is from~\eqref{eq:LqDef}. 
    Suppose that Assumptions \ref{as:UnqeEig}, \ref{as:d},
    and \ref{as:timescale} hold. 
    %With the objective function from \eqref{eq:quadOBJ} and stepsize~$\gamma \in \paren{0,\frac{2}{\lambda_{\min}(Q_u) + L}}$,~$Q_{\mathbf{u}}$ and~$Q_{\mathbf{y}}$ is from~\eqref{eq:quadOBJ} and~$L := \lambda_{\max}(Q_{\mathbf{u}} + H^{\top}Q_{\mathbf{y}}H)$. 
    For each maximal solution~$\phi$ to~$\mathcal{H}_{FO}$ with initial condition~$\phi(0, 0) = \nu$ that satisfies~\eqref{eq:initconds}, for each~$(t, j) \in \textnormal{dom } \phi$, we have
    % \begin{multline}\label{eq:tmCompleteBound}
    %     \|\phi(t, j)\|_{\mathcal{A}(t)}
    %     \leq \mu_{\max}(\Lambda_{\n{des}})\frac{|\lambda_6|}{|\lambda_1|}\exp\big({-\big|\lambda_{1}\big|t}\big) \big\|\phi\Hotime{0}{0}\big\|_{\mathcal{A}(t)}\\
    %     + \bigg[\mu_{\max}(\Lambda_{\n{des}})\frac{|\lambda_6|}{|\lambda_1|} - 1\bigg] \frac{\mu_{\max}(\Lambda_{\n{des}})\frac{|\lambda_6|}{|\lambda_1|}d_{\mathcal{U}}}{m_{c}\big|\lambda_{1}\big|}\Big(1+q^{\frac{\ell}{2}}\Big)\exp\big({-\big|\lambda_{1}\big|}t\big)\\
    %     + \bigg[\mu_{\max}(\Lambda_{\n{des}})\frac{|\lambda_6|}{|\lambda_1|}\big(1 + t\big) - 1\bigg] \frac{\mu_{\max}(\Lambda_{\n{des}})\frac{|\lambda_6|}{|\lambda_1|}\big\|A_{\n{stab}}^{-1}\big\|\big\|K\big\|\bar{\mathbf{d}}}{m_{c}\big|\lambda_{1}\big|}\exp\big({-\big|\lambda_{1}\big|}t\big), 
    % \end{multline}
    \begin{multline}\label{eq:tmCompleteBound}
        \|\phi(t, j)\|_{\mathcal{A}(t)}
        \leq \mu_{\max}(\Lambda_{\n{des}})\frac{|\lambda_6|}{|\lambda_1|}\exp\big({-\big|\lambda_{1}\big|}t\big) \big\|\phi\Hotime{0}{0}\big\|_{\mathcal{A}(t)}\\
        + \frac{\mu_{\max}(\Lambda_{\n{des}})|\lambda_6|d_{\mathcal{U}}}{m_{c}\big|\lambda_{1}\big|^2}\big[2 - \exp(-|\lambda_1|\atimer{,\max}) -\exp\big({-\big|\lambda_{1}\big|}t\big)\big] \\
        + \frac{\mu_{\max}(\Lambda_{\n{des}})|\lambda_6| q^{\frac{\ell}{2}} d_{\mathcal{U}}}{m_{c}\big|\lambda_{1}\big|^2}\Big[1 - \exp\big({\big|\lambda_{1}\big|}\atimer{,\min}\big)\exp\big({-\big|\lambda_{1}\big|}t\big)\Big] \\
        + \frac{\mu_{\max}(\Lambda_{\n{des}}) |\lambda_6|}{m_{c}|\lambda_1|^2}\big\|A_{\n{stab}}^{-1}\big\|\big\|K\big\|\bar{\mathbf{d}} \Big[1 + \big(\mu_{\max}(\Lambda_{\n{des}})|\lambda_6|t - 1\big)\exp\big({-\big|\lambda_{1}\big|}t\big)\Big],
    \end{multline}
    where 
    $\bar{\mathbf{d}}$ is from Assumption~\ref{as:d}, 
    $\ell \geq 1$ is from Assumption~\ref{as:timescale}, 
    $\mathcal{A}(t)$ is from~\eqref{eq:closedsetA}, 
    and
    $q \in \paren{0,1}$ is from~\eqref{eq:LqDef}.   
    In particular, each such solution satisfies 
    \begin{multline} 
        \limsup_{\substack{(t, j) \in \dom \phi \\ t+j\rightarrow\infty}}\twonorm{\phi\Hotime{t}{j}}_{\mathcal{A}(t)} \leq \\ \mu_{\max}(\Lambda_{\n{des}})\frac{|\lambda_6|}{m_{c}|\lambda_1|^2}  
        \Big(2d_{\mathcal{U}} -d_{\mathcal{U}}\exp({-\big|\lambda_{1}\big|}\atimer{,\max}) 
        + d_{\mathcal{U}}q^{\frac{\ell}{2}} + \big\|A_{\n{stab}}^{-1}\big\|\big\|K\big\|\bar{\mathbf{d}}\Big). 
    \end{multline}
\end{proposition}
\begin{proof}
    See~\ref{Ap:prop2Proof}.
\end{proof}

Proposition~\ref{prop:completeHOCW} \green{shows} that a chaser satellite approaches a ball around its \green{chosen} rendezvous
point exponentially quickly with rate given by~$\exp(-|\lambda_1|t)$.  
At time~$t$ 
that \green{chosen} rendezvous point is~$\tilde{\mathbf{x}}(t)$, 
and the ball around it has radius~$r$, which depends on the 
eigenvalues in~$\Lambda_{\n{des}}$, which
are user-specified, as well as the parameter~$\ell$ from Assumption~\ref{as:timescale},
which is a function of the speed
of the chaser's onboard processor. 

\subsection{Convergence and Robustness}
    Proposition~\ref{prop:completeHOCW} requires initial conditions that satisfy~\eqref{eq:initconds}. 
    We next present a convergence result for solutions to~$\HFO$ in which 
    all initial conditions are arbitrary except for~$\dtimer(0, 0)$, which is
    simply used to represent time itself and hence begins at~$\dtimer(0, 0) = 0$. 
    This is both our main result on convergence and a step toward establishing a robustness result below in Corollary~\ref{cor:robust}. 
    %For arbitrary initial conditions that violate~\eqref{eq:initconds}, we aim to look an overall global convergence result for the solutions. Since \eqref{eq:tmCompleteBound} takes a form found similarly to \cite{chuy2025hybridsystemsmodelfeedback} for bounding its solutions to its optimal steady-state value, then the following results can be found similarly. 
    
    \begin{theorem}[Global Convergence of $\HFO$] \label{tm:globalcompleteHOCW}
    Consider the hybrid system $\HFO$ from \eqref{eq:hybridFO}, 
    %where based on~$\Lambda_{\n{des}}$ we implement \eqref{eq:Kdef}, 
    and suppose that Assumptions~\ref{as:UnqeEig}, 
    \ref{as:d}, 
    and~\ref{as:timescale} hold. 
    Suppose the objective from~\eqref{eq:quadOBJ} 
    is used with 
    a stepsize~$\gamma \in \paren{0,\frac{2}{\lambda_{\min}(Q_u) + L}}$, 
    where~$L$ is from~\eqref{eq:LqDef} and
    $Q_{\mathbf{u}}$ is from~\eqref{eq:quadOBJ}. For each maximal solution~$\phi$ to~$\mathcal{H}_{FO}$ with initial condition~$\phi(0, 0)$ that satisfies~$\dtimer(0,0) = 0$, 
    for each~$(t, j) \in \textnormal{dom } \phi$,  
        \begin{multline}
            \|\phi(t, j)\|_{\mathcal{A}(t)}
            \leq \mu_{\max}(\Lambda_{\n{des}})\frac{|\lambda_6|}{|\lambda_1|}\exp\big({-\big|\lambda_{1}\big|}t\big) \big\|\phi\Hotime{0}{0}\big\|_{\mathcal{A}(t)}\\
            + \frac{\mu_{\max}(\Lambda_{\n{des}})|\lambda_6|d_{\mathcal{U}}}{m_{c}\big|\lambda_{1}\big|^2}\big[2 - \exp(-2|\lambda_1|\atimer{,\max}) -\exp\big({-\big|\lambda_{1}\big|}t\big)\big] \\
            + \frac{\mu_{\max}(\Lambda_{\n{des}})|\lambda_6| q^{\frac{\ell}{2}} d_{\mathcal{U}}}{m_{c}\big|\lambda_{1}\big|^2}\Big[1 - \exp\big({\big|\lambda_{1}\big|}\atimer{,\min}\big)\exp\big({-\big|\lambda_{1}\big|}t\big)\Big] \\
            + \frac{\mu_{\max}(\Lambda_{\n{des}}) |\lambda_6|}{m_{c}|\lambda_1|^2}\big\|A_{\n{stab}}^{-1}\big\|\big\|K\big\|\bar{\mathbf{d}} \Big[1 + \big(\mu_{\max}(\Lambda_{\n{des}})|\lambda_6|t - 1\big)\exp\big({-\big|\lambda_{1}\big|}t\big)\Big],
        \end{multline}
     where $\bar{\mathbf{d}}$ is from Assumption~\ref{as:d}, $\ell \geq 1$ is from Assumption~\ref{as:timescale}, $\mathcal{A}(t)$ is from~\eqref{eq:closedsetA}, and
     $q \in \paren{0,1}$ is from~\eqref{eq:LqDef}.
     In particular, each such solution
     satisfies 
     \begin{multline}
        \limsup_{\substack{(t, j) \in \dom \phi \\ t+j\rightarrow\infty}} \|\phi(t, j)\|_{\mathcal{A}(t)} \leq \\
        \mu_{\max}(\Lambda_{\n{des}})\frac{|\lambda_6|}{m_{c}|\lambda_1|^2}\Big(2d_{\mathcal{U}} -d_{\mathcal{U}}\exp\big({-2\big|\lambda_{1}\big|}\atimer{,\max}) + d_{\mathcal{U}}q^{\frac{\ell}{2}} + \big\|A_{\n{stab}}^{-1}\big\|\big\|K\big\|\bar{\mathbf{d}}\Big).
     \end{multline}
    \end{theorem}
    \begin{proof}
        See~\ref{Ap:theorem2Proof}.         
    \end{proof}

    Theorem~\ref{tm:globalcompleteHOCW} shows that a result like Proposition~\ref{prop:completeHOCW}
    holds for arbitrary initial conditions \green{(except~$\dtimer$)}, i.e., the chaser satellite approaches an error 
    ball of known size about its \green{chosen} rendezvous point~$\tilde{\mathbf{x}}(t)$
    and it does so exponentially quickly
    with rate given by~$\exp(-|\lambda_1|t)$, regardless of its initial state.
    \blue{It is possible to 
    \green{choose the desired closed-loop eigenvalues in~$\Lambda_{des}$ to}
    bound the radius of the error ball by any~$\eta > 0$. 
    Given~$\eta > 0$, 
    a straightforward calculation shows that first choosing~$\lambda_6$ such that
    \begin{equation} \label{eq:l6bound}
        \big|\lambda_6\big| \geq \frac{\mu_{\max}(\Lambda_{\n{des}})\Big(2d_{\mathcal{U}}  + d_{\mathcal{U}}q^{\frac{\ell}{2}} + \big\|A_{\n{stab}}^{-1}\big\|\big\|K\big\|\bar{\mathbf{d}}\Big)}{m_{c}\eta}
    \end{equation}
    and then choosing~$\lambda_1$ such that
    \begin{equation} \label{eq:l1bound}
        \big|\lambda_1\big| \geq \sqrt{\frac{\mu_{\max}(\Lambda_{\n{des}})\Big(2d_{\mathcal{U}}  + d_{\mathcal{U}}q^{\frac{\ell}{2}} + \big\|A_{\n{stab}}^{-1}\big\|\big\|K\big\|\bar{\mathbf{d}}\Big)\big|\lambda_6\big|}{m_{c}\eta}}
    \end{equation}  
    ensures that~$\limsup_{\substack{(t, j) \in \dom\phi \\ t + j \to \infty}} \|\phi(t, j)\|_{\mathcal{A}(t)} \leq \eta$.
    These bounds may require~$|\lambda_6|$ and~$|\lambda_1|$ to be large, \green{which would 
    also require~$|\lambda_2|$, $|\lambda_3|$, $|\lambda_4|$, and~$|\lambda_5|$ to be large.
    These large eigenvalues can make the gains~$k_1$, $k_4$, $k_8$, $k_{11}$, $k_{15}$, and~$k_{18}$ large 
    when using the controller in Theorem~\ref{tm:stabGains}, which
    can result in large
    inputs being applied to the system.  
    However,} several properties of a satellite rendezvous problem can make these inputs smaller. 
    Both bounds depend on~$\bar{\mathbf{d}}$, which bounds the rate of change of
    the disturbance~$\mathbf{d}$. Slowly varying disturbances have smaller
    values of~$\bar{\mathbf{d}}$ and hence allow for smaller values of~$|\lambda_6|$ and~$|\lambda_1|$.
    Both bounds also depend on~$q^{\frac{\ell}{2}}$, which can be made smaller by executing
    faster computations onboard the chaser satellite. 
    These bounds are sufficient to have asymptotic error
    bounded by~$\eta$, and
    in Section~\ref{sec:simResults} we show
    that modest asymptotic error is incurred even without satisfying them. 
    } 
    
    The radius of the 
    ball in Theorem~\ref{tm:globalcompleteHOCW}
    contains the term~$-d_{\mathcal{U}}\exp\big({-2\big|\lambda_{1}\big|}\atimer{,\max})$ 
    where Proposition~\ref{prop:completeHOCW}
    contains the term
    $-d_{\mathcal{U}}\exp({-\big|\lambda_{1}\big|}\atimer{,\max})$.
    The ball about~$\tilde{\mathbf{x}}(t)$ in Theorem~\ref{tm:globalcompleteHOCW}
    is therefore larger, \blue{which implies that the conditions in~\eqref{eq:l6bound}
    and~\eqref{eq:l1bound} are sufficient to bound asymptotic error by~$\eta$
    using either result.} 

\subsection{Robustness of Global Hybrid Convergence} \label{ss:robustness}
    % In this subsection, we solve Problem~\ref{prob:robustness}. 
    Theorem~\ref{tm:globalcompleteHOCW} not only characterizes system behavior from
    arbitrary initial conditions, but also enables analysis of the robustness of the system~$\HFO$
    to various perturbations. We consider perturbations that represent unplanned variations in 
    the control cadence, namely the
    timing with which
    inputs to the chaser are changed, as well as 
    unplanned variations in 
    the timing with which computations are completed. 
    We first consider the perturbed domain of the flow map, which is
    \begin{equation} \label{eq:CWCrhodef}
        C_\rho := \big\{\zeta \in \mathcal{X} : \atimer{}\in[0,\atimer{,\max} + \theta_{c,\max}], 
                \ctimer{}\in[0,\ctimer{,\n{comp}} + \theta_{g,\n{comp}}]\big\},
    \end{equation}
    where~$\theta_{c,\max} \in (-\atimer{,\max}, \infty)$ and~$\theta_{g,\n{comp}} \in (-\ctimer{,\n{comp}}, \infty)$.
    These perturbations allow~$\atimer{}$ to take values larger than~$\atimer{,\max}$ and 
    also allow it to reset to any positive value outside of~$[\atimer{,\min},\atimer{,\max}]$ at jumps. 
    Similarly, they allow~$\ctimer{}$ to take values larger than~$\ctimer{,\n{comp}}$ and
    for it to reset to any positive value at jumps. 
    Jumps are still triggered when at least one of these timers reaches zero, 
    and we use $D_\rho := D$ for the perturbed system model. 

    For the flow map, we allow errors in the rates at which 
    the
    timers~$\ctimer{}$ and~$\atimer{}$ 
    count down. 
    We define the perturbed flow map as
    \begin{equation}
        F_\rho(\zeta) :=
        \left(\begin{array}{cc}
            A_{\n{stab}} \mathbf{x} + B_{\n{stab}} \mathbf{u} - B_{\n{stab}}K\mathbf{d}(\dtimer) \\   
                 \mathbf{0} \\
                 \mathbf{0} \\
                 \mathbf{0} \\
            -1 + \kappa_{c} \\
            -1 + \kappa_{g} \\
            1
        \end{array}\right), 
    \end{equation}
    where
    %~$\hat{A} \in \mathbb{R}^{6 \times 6}$ are modeling errors in the matrix~$A_{\n{stab}}$, $\hat{B} \in \mathbb{R}^{6 \times 3}$ are modeling errors in the matrix~$B_{\n{stab}}$, 
    $\kappa_c \in (-\infty, 1)$ models error in the rate at which~$\atimer{}$ counts down, and ${\kappa_g \in (-\infty, 1)}$ 
    models error in the rate at which~$\ctimer{}$ counts down. 

    We now define each case of the perturbed jump map. 
    For case (i), we account for the fact that $\ctimer{}$ could reset to a value other than $\ctimer{,\n{comp}}$ after a gradient descent iteration is completed, and we define~$G_{1,\rho}$ as
    \begin{equation}
        G_{1,\rho}(\zeta) := 
        \left(\begin{array}{cc}
             \mathbf{x}\\
             \mathbf{u}\\
             \mathbf{y}_{s}\\
            \Pi_\mathcal{U}\bigg[\mathbf{z} - \gamma\Big(\nabla_{\mathbf{u}} \Phi\big(\mathbf{z}, \mathbf{y}_s\big) + H^T\nabla_{\mathbf{y}}\Phi\big(\mathbf{z}, \mathbf{y}_s\big)\Big)\bigg]\\
            \atimer{} \\
            \ctimer{,\n{comp}} + \theta_{g,\n{comp}} \\
            \dtimer
        \end{array}\right),
    \end{equation}
    where~$\theta_{g,\n{comp}}$ is from~\eqref{eq:CWCrhodef}. For case (ii) we define~$G_{2,\rho}$ as
    \begin{equation}
        G_{2,\rho}(\zeta) := 
        \left(\begin{array}{cc}
            \mathbf{x}\\
             \mathbf{z}\\
            H_{\n{stab}} \mathbf{u} +  \mathbf{d}(\dtimer)\\
             \mathbf{z}\\
            \left[\atimer{,\min} + \theta_{c,\min},\atimer{,\max} + \theta_{c,\max}\right]\\
            \ctimer{} \\
            \dtimer
        \end{array}\right).
    \end{equation}
    % Here,~$\hat{H} \in \mathbb{R}^{p \times m}$ models perturbations to the matrix~$H$, which includes~$\hat{A}$ and~$\hat{B}$ as errors which were described above, and errors in the output map~$\Psi$. 
    The interval to which~$\atimer{}$ is reset is perturbed with constants~$\theta_{c,\min} \in (-\atimer{,\min},\infty)$ and~$\theta_{c,\max} \in (-\atimer{,\max},\infty)$ that satisfy $0 < \atimer{,\min} + \theta_{c,\min} \leq \atimer{,\max} + \theta_{c,\max}$, which ensures that~$\atimer{}$ is reset to a non-empty interval, though both endpoints of the interval can be perturbed. 
    For case (iii) we define~$G_{3,\rho}(\zeta) := G_{1,\rho}(\zeta) \cup G_{2,\rho}(\zeta)$,
    and we define the perturbed jump map as
    \begin{equation}
        G_\rho(\zeta) := 
        \begin{cases}
            G_{1,\rho}(\zeta) \quad \text{if }
            \atimer{} > 0 \textnormal{ and } \ctimer{} = 0 \quad \textnormal{ Case (i)} \\
            G_{2,\rho}(\zeta) \quad \text{if } \atimer{} = 0 \textnormal{ and } \ctimer{} > 0 \quad \textnormal{ Case (ii)} \\
            G_{3,\rho}(\zeta) \quad \text{if } \atimer{} = 0 \textnormal{ and } \ctimer{} = 0 \quad \textnormal{ Case (iii)}
        \end{cases}.
    \end{equation}
    We also define
    \begin{equation} \label{eq:newrhodef}
        \rho = \max\{
        \theta_{g,\n{comp}},  
        \kappa_c, \kappa_g, \theta_{c,\min}, \theta_{c,\max}
        \}
    \end{equation}
to be the maximum size of any perturbation. 

The full perturbed hybrid system model is defined as
\begin{equation} \label{eq:HFOR}
    \HFOR := 
   \begin{cases}
       \dot{\zeta} \in F_\rho(\zeta) & \zeta \in C_\rho\\
       \zeta^+\in G_\rho(\zeta) & \zeta \in D_\rho
   \end{cases}. 
\end{equation}

The notion of robustness that we analyze requires the following definition. 

    \begin{definition}[$(\tau, \epsilon)$-closeness {\cite[Definition 5.23]{Hybridbook}}] \label{def:teclose}
        Given $\tau,\epsilon > 0$, two hybrid arcs $\phi_1$ and $\phi_2$ are $(\tau,\epsilon)$-close if
        \begin{enumerate}
            \item for all $(t,j)\in\textnormal{dom }\phi_1$ with $t+j\leq \tau$ there exists $s$ such that $(s,j)\in \textnormal{dom }\phi_2$, $|t-s|<\epsilon$, and 
            \begin{equation}
                |\phi_1(t,j)-\phi_2(s,j)|<\epsilon;
            \end{equation}
            \item for all $(t,j)\in\textnormal{dom }\phi_2$ with $t+j\leq \tau$ there exists $s$ such that $(s,j)\in \textnormal{dom }\phi_1$, $|t-s|<\epsilon$, and 
            \begin{equation}
                |\phi_2(t,j)-\phi_1(s,j)|<\epsilon.
            \end{equation}
        \end{enumerate}
    \end{definition}
The following result characterizes robustness of hybrid feedback optimization for the satellite rendezvous problem. 

    \begin{corollary}[Robustness of $\HFO$]\label{cor:robust}
        Consider the hybrid system~$\HFOR$ with~$\rho$ as defined in~\eqref{eq:newrhodef}, and 
        suppose that 
        Assumptions~\ref{as:UnqeEig}, \ref{as:d}, and~\ref{as:timescale} hold. 
        Consider objectives of the form of~\eqref{eq:quadOBJ}, and suppose that the gradient
        descent algorithm uses a stepsize~$\gamma \in \paren{0,\frac{2}{\lambda_{\min}(Q_u) + L}}$. 
        Then, for every~$\epsilon > 0$ and~$\tau > 0$, there exists~$\delta > 0$  with the following property: for every solution~$\phi_{\delta}$ to~$\HFO^{\delta\rho}$, there exists a solution~$\phi$ to~$\HFO$ such that~$\phi_{\delta}$ and~$\phi$ are~$(\tau,\epsilon)$-close. 
    \end{corollary}
    \begin{proof}
        Lemma~\ref{lem:wellPosed} shows that~$\HFO$ is well-posed, and 
        Proposition~\ref{prop:maxsolComplete} shows its maximal solutions are complete.
        The result then follows from~\cite[Prop. 6.34]{hybridfeedcntrl}. 
        %Following a similar structure to the proof seen in \cite{chuy2025hybridsystemsmodelfeedback}, then the result follows. This holds since \eqref{eq:hybridFO} satisfies the same properties and holds a similar structure seen in \cite{chuy2025hybridsystemsmodelfeedback}. 
    \end{proof}

    Corollary~\ref{cor:robust} shows that, up to any time~$\tau$
    and for any tolerance~$\epsilon$, there exists a strictly positive
    perturbation to~$\HFO$ such that each solution of the perturbed
    system is~$(\tau, \epsilon)$-close to a solution that~$\HFO$ could have produced. 
    \blue{
    This result only holds over bounded hybrid time horizons, which means
    that one must select a time~$\tau$ up until which this result is applied.
    In practice, one often runs a satellite rendezvous controller until 
    some point in time, after which a different controller is used to complete
    a task after the chaser has rendezvoused with the target, e.g., the chaser
    may inspect the target satellite. Then one can set the time~$\tau$
    to be the amount of time \green{for which} the rendezvous controller will be used,
    and one is assured of robustness for the \green{entire} time that it is used. 
    }

\section{Simulation Results} \label{sec:simResults}
This section presents three sets of simulation results
for hybrid feedback optimization running onboard the chaser
satellite.
The first set of results simulates the system~$\HFO$ from~\eqref{eq:hybridFO}
from a fixed initial condition 
to illustrate its nominal behavior. 
The second simulates the perturbed system~$\HFOR$ 
from~\eqref{eq:HFOR} for various values
of the perturbation~$\rho$ to compare this behavior
to the nominal behavior.
The third simulates~$\HFO$ from~$20$ different
initial conditions to show that all of them
produce trajectories that achieve
the desired rendezvous behavior. 

\subsection{Problem Setup}
Suppose we consider the chaser dynamics from \eqref{eq:CWltisystem}, and the feedback optimization problem we solve is
% \begin{equation}
%     \begin{aligned}
%     \label{probform2}
%         \min_{\mathbf{u},\mathbf{y}} \quad & \Phi(\mathbf{u},\mathbf{y}):=\frac{1}{2} \mathbf{u}^\top Q_u \mathbf{u} + \frac{1}{2}(\mathbf{y}-\hat{\mathbf{y}})^\top Q_y (\mathbf{y} - \hat{\mathbf{y}})\\
%         \text{subject to} \quad & \mathbf{u} \in \mathcal{U}, \,\, \mathbf{y} \in \R^6, 
%     \end{aligned}
% \end{equation}
\begin{mini!}|l|
    {\mathbf{u},\mathbf{y}}
    {\Phi(\mathbf{u},\mathbf{y}):=\frac{1}{2} \mathbf{u}^\top Q_{\mathbf{u}} \mathbf{u} + \frac{1}{2}(\mathbf{y}-\hat{\mathbf{y}})^\top Q_{\mathbf{y}} (\mathbf{y} - \hat{\mathbf{y}})}
    {\label{probform2}}
    {}
    \addConstraint{\mathbf{u} \in \mathcal{U}, \,\, \mathbf{y} = H_{\n{stab}}\mathbf{u} + \mathbf{d},}{}
\end{mini!}
where~$Q_{\mathbf{u}} = 0.00005\cdot I_3$,~$Q_{\mathbf{y}} = \text{diag}(0.04,0.04,0.04,0.055,0.055,0.055)$, $\mathcal{U} = [-0.4,0.4]$,~$\mathbf{d}(t) = 5\sin(t) \cdot\mathbbm{1}_{6}$, and~$\hat{\mathbf{y}} = (100.0,100.0,100.0,0.0,0.0,0.0)^\top$.   

For simulations, the Hybrid Equations Toolbox (Version 3.0.0.76 \cite{Sanfelice_2013}) was used, along with the initial conditions
\begin{multline} \label{eq:initconds1}
    \mathbf{x}(0,0)= (1500, -1770, 3000, 1, 3.4, 1)^\top, \mathbf{u}(0,0) = (0, 0, 0)^\top, \ctimer{}(0,0) = 0.5
    \\ \mathbf{y}_s(0,0) = (1505, -1775, 3005, 6, 5.4, 6.2)^\top, \mathbf{z}(0,0) = (0, 0, 0)^\top, 
    ~\atimer{}(0,0) = 0.175,\\
    ~\dtimer(0,0) = 0, 
\end{multline}
where $\ctimer{,\n{comp}} = 0.5, \atimer{,\min} = 1.5 $, and $\atimer{,\max} = 2.0$, 
with stepsize~$\gamma = 0.1$.

For the stabilizing controller, we 
set
\begin{equation}
\Lambda_{\n{des}} = \{-0.0155, -0.0163, -0.0155, -0.0170, -0.0165, -0.0170\}, 
\end{equation}
which gives 
\begin{equation}
        \tilde{A}_{\n{stab}} = 
        \left[
        \begin{array}{cccccc}
            0 & 0 & 0 & 1 & 0 & 0 \\
            0 & 0 & 0 & 0 & 1 & 0 \\
            0 & 0 & 0 & 0 & 0 & 1 \\
            -0.0003 & 0 & 0 & -0.0318 & 0 & 0 \\
            0 & -0.0003 & 0 & 0 & -0.0325 & 0 \\
            0 & 0 & -0.0003 & 0 & 0 & -0.0335 \\
        \end{array}
        \right].
    \end{equation}

\begin{figure*}[ht]
    \centering
    \includegraphics[width=1\linewidth,keepaspectratio]{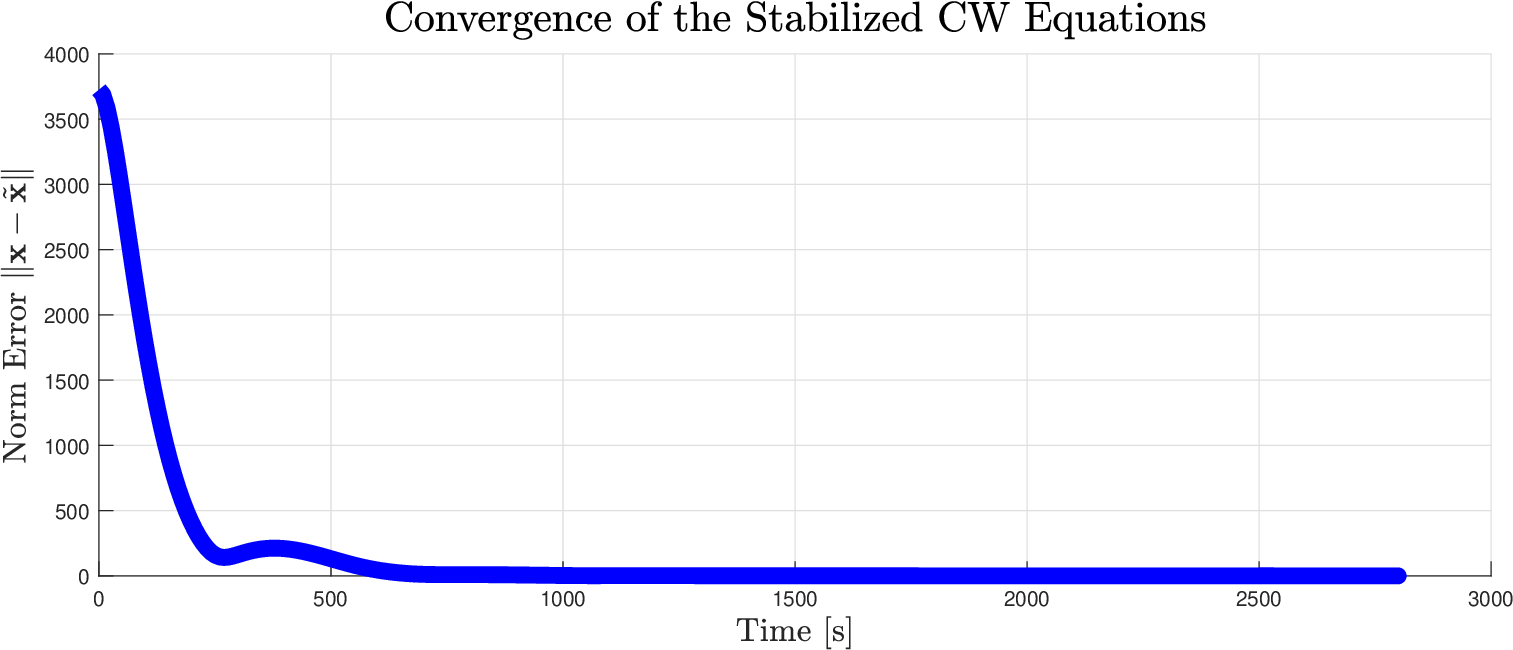}    
    \caption{ 
    Convergence of the state $\mathbf{x}$ under the dynamics $\HFO$ from the initial condition in~\eqref{eq:initconds1}. The value of
    $\mathbf{x}(t, j)$ is asymptotically close to the \green{chosen} rendezvous point~$\tilde{\mathbf{x}}(t)$,
    and the asymptotic value of~$\|\mathbf{x}(t, j) - \tilde{\mathbf{x}}(t)\|$ is 
    approximately~$8.2\cdot 10^{-2}$ meters.
    }
    \label{fig:Nominal}
\end{figure*}

\subsection{Simulation Results}
We first simulate the behavior of~$\HFO$ from the initial 
condition in~\eqref{eq:initconds1}.
As in Proposition~\ref{prop:completeHOCW}
and Theorem~\ref{tm:globalcompleteHOCW},
our focus is on evaluating how closely
the state of the chaser
approaches its \green{chosen} rendezvous
point. We therefore examine the value
of~$\|\mathbf{x}(t, j) - \tilde{\mathbf{x}}(t)\|$
as a function of~$t$. 

This value over time is plotted
in Figure~\ref{fig:Nominal}, which
shows that the chaser satellite converges quite closely
to the \green{chosen} rendezvous point, despite
the perturbations given by~$t \mapsto \mathbf{d}(t)$
and the fact that sub-optimal inputs to the system
are used over time. 
The trajectory shown in Figure~\ref{fig:Nominal}
represents the typical behavior we saw in our
simulations, and it shows
the success of feedback optimization
for satellite rendezvous. 
The limiting value of~$\|\mathbf{x}(t, j) - \tilde{\mathbf{x}}(t)\|$
is approximately~$0.082$, while~$\bar{\mathbf{d}} = 5$, which indicates
that feedback optimization provides 
a~$98.4$\% reduction in the magnitude of disturbances faced by the system. 

\begin{figure*}[ht]
    \centering
    \includegraphics[width=1\linewidth,keepaspectratio]{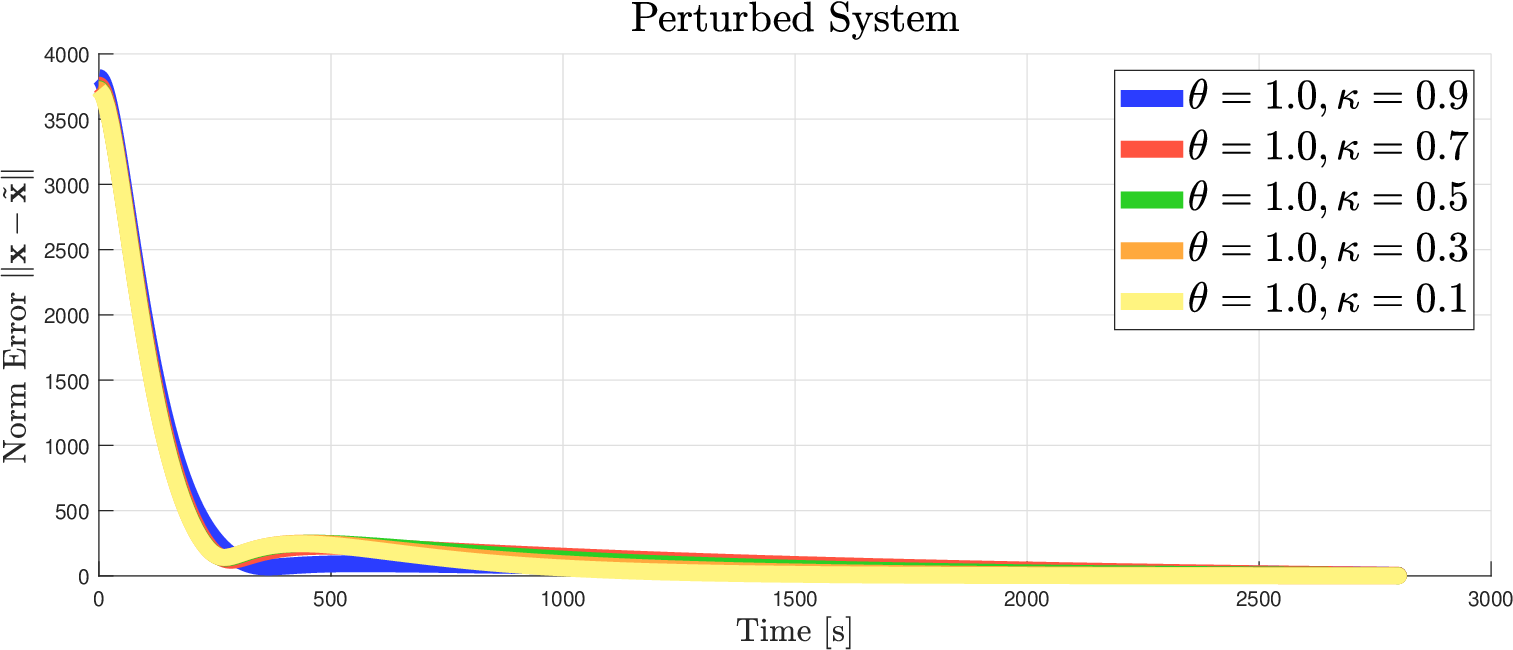}
    \caption{ Convergence of the 
    rendezvous error~$\|\mathbf{x}(t, j) - \tilde{\mathbf{x}}(t)\|$ 
    for the perturbed system $\HFOR$ from the initial condition in~\eqref{eq:initconds1} with $\theta = 1.0$
    and varying values of~$\kappa$. The perturbed system attains an asymptotic error of at most $\sim$$4.39$ meters even though there are 
    perturbations to the dynamics of~$\atimer{}$ and~$\ctimer{}$. 
    }
    \label{fig:P_1.0}
\end{figure*}
\begin{table}[ht]
    \centering
    \caption{Values of perturbations and resulting asymptotic error for the perturbed system~$\HFOR$.}
    \begin{tabular}{l|ccc|}
    \cline{2-4}
                            & \multicolumn{3}{c|}{Value of $\|\mathbf{x} - \tilde{\mathbf{x}}\|$}                               \\ \hline
    \multicolumn{1}{|l|}{}  & \multicolumn{1}{l|}{$\theta = -0.25$} & \multicolumn{1}{c|}{$\theta = 0.5$} & $\theta = 1.0$ \\ \hline
    \multicolumn{1}{|l|}{$\kappa = 0.1$} & \multicolumn{1}{c|}{0.13}  & \multicolumn{1}{c|}{0.14}  & 0.27  \\ \hline
    \multicolumn{1}{|l|}{$\kappa = 0.3$} & \multicolumn{1}{c|}{0.19}  & \multicolumn{1}{c|}{0.15}  & 0.31 \\ \hline
    \multicolumn{1}{|l|}{$\kappa = 0.5$} & \multicolumn{1}{c|}{0.43}  & \multicolumn{1}{c|}{0.27}  & 1.57  \\ \hline
    \multicolumn{1}{|l|}{$\kappa = 0.7$} & \multicolumn{1}{c|}{0.54}  & \multicolumn{1}{c|}{2.23}  & 4.39  \\ \hline
    \multicolumn{1}{|l|}{$\kappa = 0.9$} & \multicolumn{1}{c|}{1.22}  & \multicolumn{1}{c|}{4.32}  & 3.53  \\ \hline
    \end{tabular}
    \label{tab:perturbations}
\end{table}

\begin{figure*}[ht]
    \centering
    \includegraphics[width=1\linewidth,keepaspectratio]{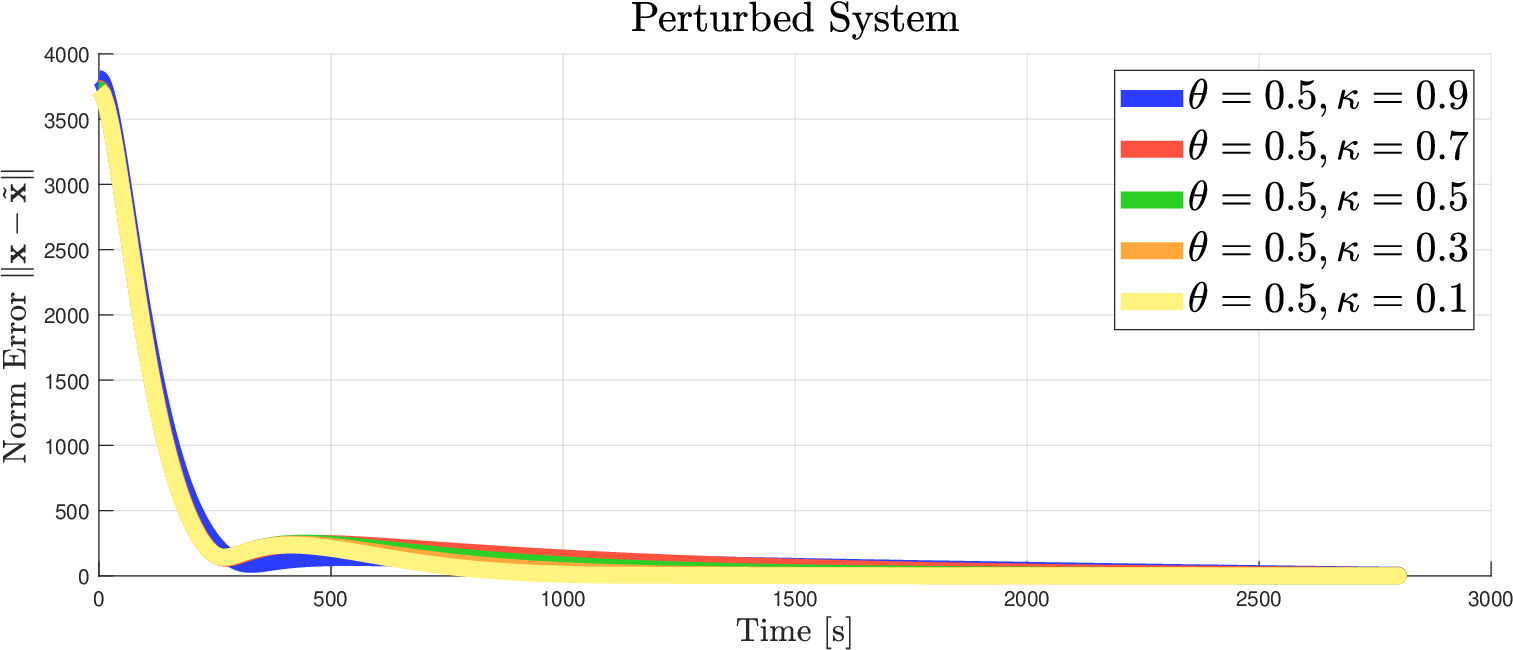}
    \caption{ Convergence of the
    rendezvous error~$\|\mathbf{x}(t, j) - \tilde{\mathbf{x}}(t)\|$ 
    for the perturbed system $\HFOR$ from the initial condition in~\eqref{eq:initconds1} with $\theta = 0.5$
    and varying values of~$\kappa$. 
    The perturbed system attains an asymptotic error of at most $\sim$$4.32$ meters even though there are perturbations to the dynamics of~$\atimer{}$ and~$\ctimer{}$. 
    }
    \label{fig:P_0.5}
\end{figure*}

\begin{figure*}[ht]
    \centering
    \includegraphics[width=1\linewidth,keepaspectratio]{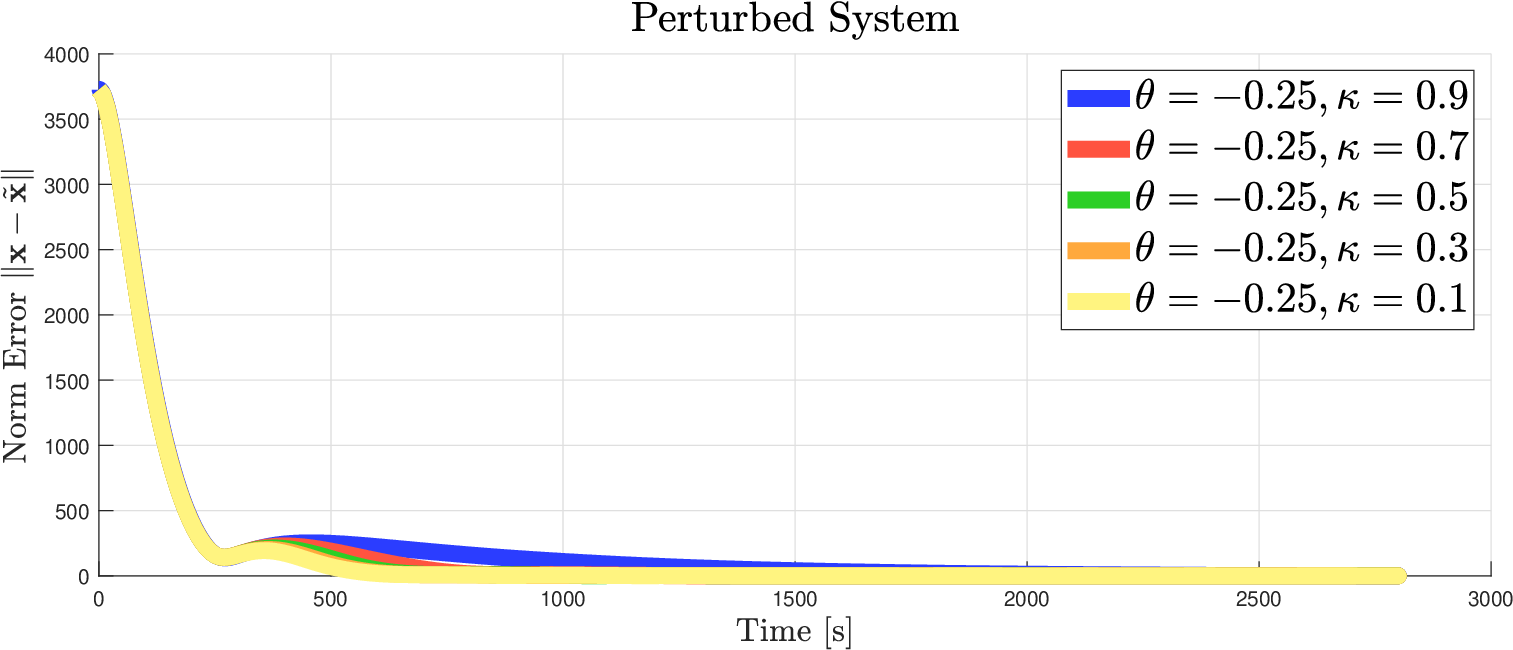}
    \caption{ Convergence of the
    rendezvous error~$\|\mathbf{x}(t, j) - \tilde{\mathbf{x}}(t)\|$  for the perturbed system $\HFOR$ from the initial condition in \eqref{eq:initconds1} with $\theta = -0.25$
    and varying values of~$\kappa$. 
    The perturbed system attains an asymptotic error of at most $\sim$$1.22$ meters even though there are perturbations to the dynamics of~$\atimer{}$ and~$\ctimer{}$.
    }
    \label{fig:P_-0.25}
\end{figure*}

We next consider the perturbed system model~$\HFOR$.
For the timers~$\atimer{}$ and~$\ctimer{}$ we use perturbations of the form $\theta_{c,\min} = \theta_{c,\max} = \theta_{g, \n{comp}} = \theta$ and $\kappa_c = \kappa_g = \kappa$. 
The values of~$\theta$ and~$\kappa$ used in simulations are shown in Table~\ref{tab:perturbations}, along with the asymptotic values of~$\|\mathbf{x}(t, j) - \tilde{\mathbf{x}}(t)\|$ for each run. 
\green{
A second-order regression of these values gives
%By computing~$a_0=0.256, a_1 = -{1.62}, a_2 = -{0.449}, a_3 = 4.131, a_4 = -{0.023}, \textnormal{ and }~a_5 = 3.364$, we can apply a second-order interpolation to show that~
\begin{equation}
\|\mathbf{x}(t, j) - \tilde{\mathbf{x}}(t)\|\approx 0.256 - 1.62\kappa - 0.449\theta + 4.131\kappa^2 - 0.023\theta^2 + 3.364\kappa\theta.
\end{equation}
}

Figure~\ref{fig:P_1.0} shows the behavior of~$\|\mathbf{x}(t, j) - \tilde{\mathbf{x}}(t)\|$ for~$\theta = 1$ and varying values of~$\kappa$.
The value~$\theta = 1$ means that the timers~$\atimer{}$ and~$\ctimer{}$ take one second longer to jump than
they do in the nominal system model. The positive values taken by~$\kappa$ mean
that both~$\atimer{}$ and~$\ctimer{}$ count down more slowly than they do in the nominal system model. 
Figure~\ref{fig:P_0.5} shows simulations
for~$\theta = 0.5$ and varying values of~$\kappa$. 
Figure~\ref{fig:P_-0.25} does the same for~$\theta = -0.25$ and varying values of~$\kappa$. 
The negative value of~$\theta$ implies that~$\atimer{}$
and~$\ctimer{}$ count down to zero faster than they do in the nominal model. 

Table~\ref{tab:perturbations_diff} gives the asymptotic values of~$\|\mathbf{x}(t, j) - \tilde{\mathbf{x}}(t)\|$
for each value of~$\rho = \max\{\theta, \kappa\}$. 
Table~\ref{tab:perturbations_diff}
shows that larger perturbations~$\rho$ produce larger asymptotic errors in the chaser's state, which is intuitive.
Table~\ref{tab:perturbations} shows that increasing~$\theta$ increases errors more rapidly than increasing~$\kappa$.
Intuitively, smaller values of both~$\theta$ and~$\kappa$ produce smaller errors because computations,
measurements of outputs, 
and
changes in the input happen more quickly for those smaller values, which allows the underlying optimization
algorithm to react more quickly to disturbances faced by the system. 

\begin{table}[ht]
    \centering
    \caption{Effects of perturbations~$\rho$ on the asymptotic value of~$\|\mathbf{x}(t, j) - \tilde{\mathbf{x}}(t)\|$.}
    \begin{tabular}{c|ccccc|}
    \cline{2-6}
    \multicolumn{1}{l|}{}                     & \multicolumn{5}{c|}{Asymptotic value of~$\|\mathbf{x}(t, j) - \tilde{\mathbf{x}}(t)\|$}                                                                      \\ \hline
    \multicolumn{1}{|c|}{$\rho$} & \multicolumn{1}{c|}{0.1} & \multicolumn{1}{c|}{0.3} & \multicolumn{1}{c|}{0.5} & \multicolumn{1}{c|}{0.7} & 0.9 \\ \hline
    \multicolumn{1}{|c|}{Value}               & \multicolumn{1}{c|}{0.05}    & \multicolumn{1}{c|}{0.11}    & \multicolumn{1}{c|}{0.35}    & \multicolumn{1}{c|}{2.15}    & 4.23    \\ \hline
    \end{tabular}
    \label{tab:perturbations_diff}
\end{table}

\begin{figure*}[ht]
    \centering
    \includegraphics[width=1\linewidth,keepaspectratio]{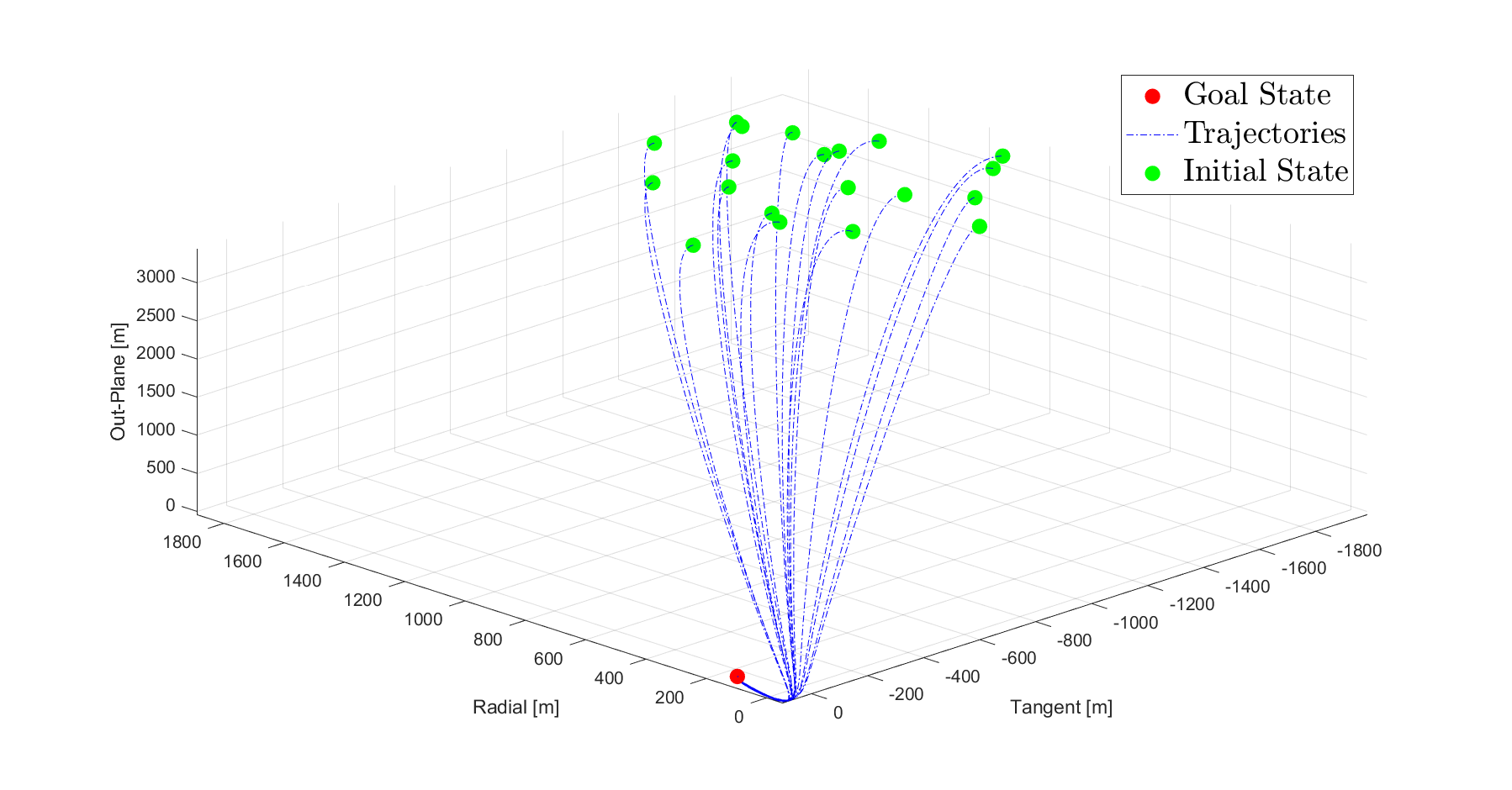}
    \caption{ The convergence of the state $\mathbf{x}$ under the dynamics of~$\HFO$ from~$20$ initial conditions satisfying \eqref{eq:initconds2}.
    In all cases, $\mathbf{x}(t, j)$ asymptotically approaches~$\tilde{\mathbf{x}}(t)$, and the maximum asymptotic
    error across all runs is~$0.81$ meters.   
    }
    \label{fig:RandInit}
\end{figure*}

For the third set of simulations 
we consider~$20$ initial conditions 
of the nominal system~$\HFO$ 
that satisfy
\begin{multline} \label{eq:initconds2}
    \mathbf{x}(0,0) \in [1000, 2000] \times [-1000, -2000] \times [-2500, -3500] \times [0.1, 4.0] \times 
    \\ [0.1, 4.0] \times [0.1, 4.0], \mathbf{u}(0,0) = \big(0, 0, 0\big)^\top, 
    ~\mathbf{y}_s(0,0) = \mathbf{x}(0,0) + 5,~\mathbf{z}(0,0) = \big(0, 0, 0\big)^\top,\\ 
    \atimer{}(0,0) = 0.175,~\ctimer{}(0,0) = 0.5,~\dtimer(0,0) = 0, 
\end{multline}
and we simulated trajectories from each. 
Figure~\ref{fig:RandInit} shows that all trajectories produce
asymptotically small values of~$\|\mathbf{x}(t, j) - \tilde{\mathbf{x}}(t)\|$,
with the largest among them being~$0.81$ meters,
which illustrates that the feedback optimization approach we use reliably
drives the chaser satellite close to its rendezvous point with
the target satellite. 

\section{Conclusion} \label{sec:results}
This paper developed a hybrid system model for feedback optimization 
in a satellite rendezvous problem that considers continuous-time dynamics 
with discrete-time optimization in the loop. 
The resulting hybrid model was shown to be well-posed, and all of its solutions
were shown to be complete and non-Zeno. Analytically, the solutions
to this system were shown to converge to a ball around a desired
rendezvous point, and in simulation the radius of that ball was shown
to be small, indicating close convergence to desired rendezvous points. 
Future work includes using hybrid feedback optimization for systems with nonlinear dynamics,
such as including attitude control with the CW equations.

\bibliographystyle{elsarticle-num}
\bibliography{Bib}

\appendix
\section{Basic Results on Hybrid Systems} \label{sec:properties}
%\subsection{Outer Semicontinuity of Jump Maps} \label{Ap:gosc}
For completeness we reproduce here two basic results from the study
of hybrid systems. 

\begin{lemma}[Lemma~A.33 in~\cite{hybridfeedcntrl}] \label{lem:OSLBcond}
     Given closed sets $D_1\subset \R^m$ and $D_2\subset \R^m$ and the set-valued maps $G_1:D_1 \rightrightarrows \R^n$ and $G_2:D_2\rightrightarrows\R^n$ that are outer semicontinuous and locally bounded relative to $D_1$ and $D_2$, respectively, the set-valued map $G:D\rightrightarrows \R^n$ given by
    \begin{equation}
    G(\zeta) := G_1(\zeta) \cup G_2(\zeta) 
             = 
        \begin{cases}
            G_1(\zeta) &\text{if } \zeta \in D_1 \backslash D_2 \\
            G_2(\zeta) &\text{if } \zeta \in D_2 \backslash D_1\\
            G_1(\zeta)\cup G_2(\zeta) &\text{if } \zeta \in D_1 \cap D_2
        \end{cases}\label{OSLBjump}
    \end{equation}
    for each $\zeta\in D$ is outer-semicontinuous and locally bounded relative to the closed set $D$. 
\end{lemma}

\begin{lemma}[Basic existence of solutions; Proposition~2.34 in \cite{hybridfeedcntrl}] \label{lem:complete}
     Let $\mathcal{H}=(C,F,D,G)$ satisfy Definition~\ref{def:hybridcond}. Take an arbitrary $\nu \in C \cup D.$ If $\nu \in D$ or \\
    (VC) there exists a neighborhood $U$ of $\nu$ such that for every $\zeta \in U \cap C$,
    \begin{equation*}
        F(\zeta) \cap T_C(\zeta) \neq \emptyset, 
    \end{equation*}
    then there exists a nontrivial solution $\phi$ to $\mathcal{H}$ with $\phi(0,0)=\nu$. 
    If (VC) holds for every $\nu \in C\backslash D$, then there exists a nontrivial solution to $\mathcal{H}$ from every initial point in $C \cup D$, and every maximal solution~$\phi$ 
    to~$\mathcal{H}$ 
    satisfies exactly one of the following conditions:
    \begin{enumerate}
        \item $\phi$ is complete;
        \item dom $\phi$ is bounded and the interval $I^J$, where $J = \sup_j \textnormal{dom } \phi$, has nonempty interior and $t\mapsto \phi(t,J)$ is a maximal solution to $\dot{z}\in F(z)$, in fact lim$_{t\rightarrow T}|\phi(t,J)|=\infty$, where $T = \sup_t \textnormal{dom } \phi$;
        \item $\phi(T,J)\in C\cup D $, where $(T,J)= \sup \textnormal{dom } \phi$. 
    \end{enumerate}
    Furthermore, if $G(D)\subset C \cup D$, then $3)$ above does not occur. 
\end{lemma} 

%-------------------------------------------------------------------
\section{Section~\ref{sec:hybridModel} Proofs} 
\subsection{Proof of Lemma \ref{lem:wellPosed}} \label{Ap:wellPosed}
    For the hybrid model $\HFO$ in~\eqref{eq:hybridFO}, the set~$C$ in~\eqref{eq:Cdef} and the set~$D$ in~\eqref{eq:Ddef} satisfy Condition~\ref{hybridcond_one} in Definition~\ref{def:hybridcond} by inspection. 
    The map~$F$ in~\eqref{eq:fwmap} is defined everywhere on~$C$ and outputs a singleton that
    is linear in~$\mathbf{x}$, $\mathbf{u}$, and~$\mathbf{d}$. 
    \blue{The disturbance~$\mathbf{d}$ that appears in~$F$ is bounded
    under Assumption~\ref{as:d}, and~$F$ therefore
    satisfies Condition~\ref{hybridcond_two} in Definition~\ref{def:hybridcond}.} 
    
    To satisfy Condition~\ref{hybridcond_three} for the map~$G$ in~\eqref{eq:jpmap}, we use Lemma~\ref{lem:OSLBcond} in~\ref{sec:properties}. 
    The jump map~$G_1$ in~\eqref{eq:g1def} is outer semicontinuous since the projection mapping~$\Pi_{\mathcal{U}}[\cdot]$ is continuous 
    and all other entries are singletons. 
    The jump map~$G_2$ in~\eqref{eq:g2def} is outer semicontinuous since its 
    only set-valued entry is a compact interval and all other entries are singletons.
    Since the jump map in~\eqref{eq:jpmap} has the structure of the jump map in Lemma~\ref{lem:OSLBcond} and the jump set in~\eqref{eq:Ddef} has the structure of the jump set in Lemma~\ref{lem:OSLBcond}, we can apply Lemma~\ref{lem:OSLBcond} to the map~$G$ in~\eqref{eq:jpmap}. 
    Then~$G$ is outer semicontinuous and locally bounded relative to the closed set $D$. 
    Then Condition~\ref{hybridcond_three} of Definition \ref{def:hybridcond} is satisfied. 
    Therefore, all conditions of Definition \ref{def:hybridcond} are satisfied, and~$\HFO$ is well-posed.

\subsection{Proof of Proposition \ref{prop:maxsolComplete}} \label{Ap:maxsolComplete}
    Since $\HFO$ is well-posed as shown in Lemma~\ref{lem:wellPosed}, the results follow from applying Lemma~\ref{lem:complete} in \ref{sec:properties}. 
    First we show that the condition (VC) holds for all~$\nu\in C\backslash D$. 
    Consider an arbitrary $\nu\in C\backslash D$, and let $U$ be a neighborhood of $\nu$.
    We show that~$F(\zeta) \cap T_{C}(\zeta) \neq \emptyset$ for each~$\zeta \in U\cap C$, where $T_{C}(\zeta)$ is the tangent cone of the set~$C$ at the point~$\zeta$. 
    Based on the flow map~$F$ from~\eqref{eq:fwmap}, only~$\dot{\mathbf{x}}$, $\dot{\atimer{}}$, $\dot{\ctimer{}}$,
    and~$\dot{\dtimer}$ are non-zero during flows. 
    In addition,~$C$ does not restrict the values of $\mathbf{x}$ or~$\dtimer$, which implies that
    all values of~$\dot{\mathbf{x}}$ and~$\dot{\dtimer}$ are allowable at all values of~$\mathbf{x}$ and~$\dtimer$, respectively. 
    On the other hand, the timers~$\atimer{}$ and~$\ctimer{}$ 
    take values in compact intervals, which implies that some directions are infeasible at some 
    values of~$\atimer{}$ and~$\ctimer{}$.
    Therefore, satisfaction of the condition~$F(\zeta)\cap T_C(\zeta) \neq \emptyset$
    is determined by the dynamics of~$\atimer{}$ and~$\ctimer{}$. 
    
    To show that $F(\zeta)\cap T_C(\zeta) \neq \emptyset$, we compute the tangent cone~$T_C(\zeta)$ as
    \begin{equation}
        T_{C}(\zeta) :=
        \begin{cases}
            \R^{18} \times \{-1\} \times \R \times \R   &\text{if $\atimer{} = \atimer{,\max}$} \text{ and $\ctimer{} \in (0, \ctimer{,\n{comp}})$}\\
            \R^{18} \times \{-1\} \times \{1\} \times \R &\text{if $\atimer{} = \atimer{,\max}$} \text{ and $\ctimer{} = 0$}\\
            \R^{18}\times \{1\} \times \{-1\} \times \R  &\text{if  $\atimer{} = 0$} \text{ and $\ctimer{} = \ctimer{,\n{comp}}$}\\
            \R^{18} \times \R \times \{-1\} \times \R  &\text{if $\atimer{} \in (\atimer{,\min}, \atimer{,\max})$} \text{ and $\ctimer{} = \ctimer{,\n{comp}}$}\\
            \R^{18}\times \{-1\} \times \{-1\} \times \R &\text{if $\atimer{} = \atimer{,\max}$ } \text{ and $\ctimer{} = \ctimer{,\n{comp}}$}\\
            \R^{18} \times \{1\} \times \{1\} \times \R  &\text{if $\atimer{} = 0$ } \text{ and $\ctimer{} = 0$}\\
            \R^{18} \times \R \times \R \times \R &\text{otherwise.}
        \end{cases}
    \end{equation}
    Then $F(\zeta)\in T_C(\zeta)$ for every~$\zeta \in C \backslash D$ and condition~(VC) from Lemma~\ref{lem:complete} is satisfied for every~$\nu \in C \backslash D$. 
    Then by Lemma~\ref{lem:complete}
    there exists a nontrivial solution $\phi$ to $\mathcal{H}_{FO}$ with the initial condition $\phi(0,0) = \nu$. 
    Let~$\mathcal{S}_{\mathcal{H}_{FO}}$ denote the set of all maximal solutions to~$\mathcal{H}_{FO}$. 
    Every solution $\phi \in \mathcal{S}_{\mathcal{H}_{FO}}$ satisfies one of the three conditions 
    in Lemma~\ref{lem:complete}. 

    The system~$\HFO$ satisfies~$G(D) \subset C \cup D$, which implies that $3)$ in Lemma~\ref{lem:complete} does not occur. 
    Regarding $2)$, the only components of the solution that change during flows are $\mathbf{x}, \atimer{}$, $\ctimer{}$, and~$\dtimer$. 
    Since $\atimer{}$ and $\ctimer{}$ take values in compact sets, they are bounded and cannot blow up to infinity. 
    Since~$\mathbf{u} \in \mathcal{U}$, which is a compact set, the input~$\mathbf{u}$ is bounded, and there exists some finite~$\mathbf{u}_{max} \geq 0$ such that~$\|\mathbf{u}\| \leq \mathbf{u}_{max}$ for all~$\mathbf{u} \in \mathcal{U}$.
    Similarly, $\mathbf{d} \in \mathcal{D}$, which is a compact set, and
    there exists some~$\mathbf{d}_{max} \geq 0$ such that~$\|\mathbf{d}\| \leq \mathbf{d}_{max}$
    for all~$\mathbf{d} \in \mathcal{D}$. 
    Since the mapping~$(\mathbf{x}, \mathbf{u}) \mapsto {A}_{\n{stab}}\mathbf{x} + B_{\n{stab}}\mathbf{u} - B_{\n{stab}}K\mathbf{d}$ 
    is globally Lipschitz and both~$\mathbf{u}$ and~$\mathbf{d}$ are bounded, the state~$\mathbf{x}$ does not blow up in finite time.
    Finally, since~$\dot{\dtimer} = 1$, the timer~$\dtimer$ cannot blow up in finite time. 
    Then Condition~$2)$ from Lemma~\ref{lem:complete} does not hold.
    Therefore, Condition~$1)$ from Lemma~\ref{lem:complete} holds and all maximal solutions to~$\HFO$ are complete.

    Regarding Zeno behavior, jumps in~$G$ are triggered only when
    either one of the timers~$\atimer{}$ or~$\ctimer{}$ has reached zero, and~$G$ resets
    any timers that have reached zero to non-zero values.
    This property implies that~$G(D) \cap D = \emptyset$, ruling out Zeno behavior by Proposition~2.34 in~\cite{hybridfeedcntrl}. 

%-------------------------------------------------------------------
\section{Section~\ref{sec:convAnalysis} Proofs} 
\subsection{Proof of Proposition \ref{prop:completeHOCW}} \label{Ap:prop2Proof}
In this proof any empty sums evaluate to zero by convention, and
we take $\alpha(-1) = 0$ for notational convenience. 
\blue{The high-level steps of the proof are:
\begin{enumerate}
    \item Bound the distance between the input $\mathbf{u}$ and the optimal input $\mathbf{u}^*$. \label{step:1}
    \item Give an expression for the difference~$\mathbf{x}\Hatime{}{} - \tilde{\mathbf{x}}(t)$
    in terms of integrals of the dynamics of~$x(t, j)$. \label{step:2}
    \item Bound each integrand that does not have a closed form integral. \label{step:2point5}
    \item Evaluate each reuslting integral and simplify to bound~$\|\mathbf{x}\Hatime{}{} - \tilde{\mathbf{x}}(t)\|$. \label{step:3}
\end{enumerate}
}

\blue{We begin with Step~\ref{step:1}.} 
By definition of~$\mathcal{A}(t)$ we have %, only the state~$\mathbf{x}$ in~$\phi$ affects the value of~$\|\phi(t, j)\|_{\mathcal{A}(t)}$. Then
    \begin{equation} \label{eq:t1normsequal}
        \|\phi(t, j)\|_{\mathcal{A}(t)} 
        = \big\|\mathbf{x}\Hatime{}{}\big\|_{\{\tilde{\mathbf{x}}(t)\}} = \big\|\mathbf{x}\Hatime{}{} - \tilde{\mathbf{x}}(t)\big\|.
    \end{equation}
    We define~$P = \max\{p \in \mathbb{N} : \bar{\alpha}(p) + p \leq j\}$,
    which is the number of times the input~$\mathbf{u}$ has changed up to jump time~$j$. 
    We now address step~\ref{step:1}, suppose~$P \geq 1$, then 
    using~\eqref{eq:xtime_def} in~\eqref{eq:c2inptiterate} we see that for each~$p \in \{1, \ldots, P\}$ we have 
    %\begin{equation}
        $\mathbf{u}\Htime{\xtime{p,0}}{\xtime{p,0}} = \mathbf{z}_{\alpha({p-1})}\Htime{\xtime{p-1,\alpha(p-1)}}{\xtime{p-1,\alpha(p-1)}}$.
        If~$P = 0$, then we have only
        $\mathbf{u}\Htime{\xtime{0,0}}{\xtime{0,0}} = \mathbf{u}(0,0) = \mathbf{z}_{0}(0, 0)$. 
    %\end{equation}
    %The iterates are being optimized to the optimizer which will be used as the input so
    For~$p \geq 1$ the optimal value for~$\mathbf{u}\Htime{\xtime{p,0}}{\xtime{p,0}}$ is 
    %\begin{equation}
        $\mathbf{u}^*\Htime{\xtime{p,0}}{\xtime{p,0}} = \mathbf{z}^*\Htime{\xtime{p-1,0}}{\xtime{p-1,0}}$,
    %\end{equation}
    where we define
    %\begin{equation}
        $\mathbf{z}^{*}\Htime{\xtime{p-1,0}}{\xtime{p-1,0}} 
        = \argmin\limits_{\mathbf{u} \in \mathcal{U}} \,  \Phi\big(\mathbf{u}, \mathbf{y}_{s}\Htime{\xtime{p-1,0}}{\xtime{p-1,0}}\big)$. 
    %\end{equation}
    From Lemma \ref{lem:inptconvrate} we know that for~$p \geq 1$ 
    \begin{multline} \label{eq:thm1_1}
        \big\|\mathbf{z}_{\alpha({p-1})}\big(\xtime{p-1,\alpha(p-1)}\big)
        -\mathbf{z}^*\Htime{\xtime{p-1,0}}{\xtime{p-1,0}}\big\| 
        \leq \\ q^{\frac{\alpha({p-1})-1}{2}}\|\mathbf{z}_{1}\Htime{\xtime{p-1,1}}{\xtime{p-1,1}}
        - \mathbf{z}^*\Htime{\xtime{p-1,0}}{\xtime{p-1,0}}\big\|.
    \end{multline}    
    For~$p \geq 1$ we also have 
    \begin{multline} \label{eq:thm1_2}
        \mathbf{z}_{1}\Htime{\xtime{p-1,1}}{\xtime{p-1,1}} 
        = \Pi_{\mathcal{U}}\bigg[\mathbf{z}_0\Htime{\xtime{p-1,0}}{\xtime{p-1,0}}
        \\ -\gamma\Big(\nabla_{\mathbf{u}} \Phi\big(\mathbf{z}_0\Htime{\xtime{p-1,0}}{\xtime{p-1,0}}, 
         \mathbf{y}_{s}\Htime{\xtime{p-1,0}}{\xtime{p-1,0}}\big)
         \\ + H^T\nabla_{\mathbf{y}} \Phi\big(\mathbf{z}_0\Htime{\xtime{p-1,0}}{\xtime{p-1,0}}, 
         \mathbf{y}_{s}\Htime{\xtime{p-1,0}}{\xtime{p-1,0}}\big)
         \Big)\bigg].
    \end{multline}
    In projected gradient descent the optimum is a fixed point so that
    \begin{multline} \label{eq:thm1_3}
        \mathbf{z}^*\Htime{\xtime{p-1,0}}{\xtime{p-1,0}}
        = \Pi_{\mathcal{U}}\Big[\mathbf{z}^*\Htime{\xtime{p-1,0}}{\xtime{p-1,0}}  
        \\ -\gamma\Big(\nabla_{\mathbf{u}} \Phi\big(\mathbf{z}^*\Htime{\xtime{p-1,0}}{\xtime{p-1,0}}, 
         \mathbf{y}_{s}\Htime{\xtime{p-1,0}}{\xtime{p-1,0}}\big)
         \\ + H^T\nabla_{\mathbf{y}} \Phi\big(\mathbf{z}^*\Htime{\xtime{p-1,0}}{\xtime{p-1,0}}, 
         \mathbf{y}_{s}\Htime{\xtime{p-1,0}}{\xtime{p-1,0}}\big)
         \Big)\bigg].
    \end{multline}
    Using the relation
    \begin{equation}
        \mathbf{y}_s\Htime{\xtime{p-1,0}}{\xtime{p-1,0}} = H_{\n{stab}} \mathbf{z}_0\Htime{\xtime{p-1,0}}{\xtime{p-1,0}} + \mathbf{d}(t_{\xtime{p-1,0}})
    \end{equation}
    from~\eqref{eq:hstab_approx},
    %where from \cite[Lemma 3]{chuy2025hybridsystemsmodelfeedback}, it holds that~$\gamma \in \paren{0,\frac{2}{\lambda_{\min}(Q_u) + L}}$.
    %Then since $\Pi_{\mathcal{U}}$ has a non expansive property, we can upper bound the term by removing $\Pi_{\mathcal{U}}$. Likewise, since we approximate $\mathbf{y}_s$ as $H\mathbf{u}$ then
    %\begin{equation} \label{eq:thm1_4}
    %    \mathbf{y}_{s}\Htime{\xtime{p-1,0}}{\xtime{p-1,0}} 
    %      = H\mathbf{u}\Htime{\xtime{p-1,0}}{\xtime{p-1,0}} + \mathbf{d}\Htime{\xtime{p-1,0}}{\xtime{p-1,0}}
    %      = H\mathbf{z}_{0}\Htime{\xtime{p-1,0}}{\xtime{p-1,0}}  + \mathbf{d}\Htime{\xtime{p-1,0}}{\xtime{p-1,0}}.
    %\end{equation}
    we combine~\eqref{eq:thm1_1},~\eqref{eq:thm1_2}, and~\eqref{eq:thm1_3} 
    %and~\eqref{eq:thm1_4} 
    %we use the fact that~$u\Htime{\xtime{p,0}}{\xtime{p,0}} = z_{\alpha(p-1)}\big(\xtime{p-1,\alpha(p-1)}\big)$ and~$u^*\Htime{\xtime{p,0}}{\xtime{p,0}} = z^*(\xtime{p-1,0})$
    %so that 
    to find that for~$p \geq 1$, 
    \begin{multline} \label{eq:thm1_5}
        \twonorm{\mathbf{u}\Htime{\xtime{p,0}}{\xtime{p,0}}  - \mathbf{u}^*\Htime{\xtime{p,0}}{\xtime{p,0}}}
        \leq q^{\frac{\alpha({p-1})-1}{2}} \big\| \mathbf{z}_{0}\Htime{\xtime{p-1,0}}{\xtime{p-1,0}}  \\ 
        -\gamma(Q_{\mathbf{u}} + H_{\n{stab}}^{\top}Q_{\mathbf{y}}H_{\n{stab}})\mathbf{z}_{0}\Htime{\xtime{p-1,0}}{\xtime{p-1,0}} 
        - \gamma H_{\n{stab}}^\top Q_{\mathbf{y}} (\mathbf{d}(t_{\xtime{p-1,0}}) - \hat{\mathbf{y}}) \\
        - \mathbf{z}^*\Htime{\xtime{p-1,0}}{\xtime{p-1,0}} 
        + \gamma\big(Q_{\mathbf{u}} + H_{\n{stab}}^TQ_{\mathbf{y}}H_{\n{stab}}\big)\mathbf{z}^*\Htime{\xtime{p-1,0}}{\xtime{p-1,0}} \\
        + \gamma H_{\n{stab}}^TQ_{\mathbf{y}}(\mathbf{d}(t_{\xtime{p-1,0}}) - \hat{\mathbf{y}}) \big\|,
    \end{multline}
    where we have substituted
    both~$\mathbf{u}\Htime{\xtime{p,0}}{\xtime{p,0}} = \mathbf{z}_{\alpha(p-1)}\big(\xtime{p-1,\alpha(p-1)}\big)$
    and~$\mathbf{u}^*\Htime{\xtime{p,0}}{\xtime{p,0}} = \mathbf{z}^*\Htime{\xtime{p-1,0}}{\xtime{p-1,0}}$. 
    We have also used the non-expansive property of~$\Pi_{\mathcal{U}}$, namely
    that~$\|\Pi_{\mathcal{U}}[u_1] - \Pi_{\mathcal{U}}[u_2]\| \leq \|u_1 - u_2\|$ for all~$u_1, u_2 \in \mathcal{U}$. 
    Combining like terms and applying the triangle inequality gives 
    \begin{multline}
        \twonorm{\mathbf{u}\Htime{\xtime{p,0}}{\xtime{p,0}} - \mathbf{u}^*\Htime{\xtime{p,0}}{\xtime{p,0}}}
        \leq  q^{\frac{\alpha({p-1})-1}{2}} \Big\|
        \mathbf{z}_{0}\Htime{\xtime{p-1,0}}{\xtime{p-1,0}} 
        \\ - \mathbf{z}^*\Htime{\xtime{p-1,0}}{\xtime{p-1,0}} \Big\| 
        \cdot \Big\|I_3 - \gamma(Q_{\mathbf{u}} + H_{\n{stab}}^TQ_{\mathbf{y}}H_{\n{stab}})\Big\|
    \end{multline}
    for~$p \geq 1$. 
    We have
    %\begin{equation}
        $\|\mathbf{z}_{0}\Htime{\xtime{p-1,0}}{\xtime{p-1,0}} 
        - \mathbf{z}^*\Htime{\xtime{p-1,0}}{\xtime{p-1,0}}\| \leq d_{\mathcal{U}}$
    %\end{equation}
    by definition of $d_{\mathcal{U}}$, 
    and a straightforward calculation shows~$\|I_3 - \gamma(Q_{\mathbf{u}} + H_{\n{stab}}^TQ_{\mathbf{y}}H_{\n{stab}})\| \leq q^{\frac{1}{2}}$. 
    Then
    \begin{equation} %\label{eq:thm1_inptdiff}
        \twonorm{\mathbf{u}\Htime{\xtime{p,0}}{\xtime{p,0}} - \mathbf{u}^*\Htime{\xtime{p,0}}{\xtime{p,0}}}
        %&\leq  q^{\frac{\alpha({p-1})-1}{2}} d_{\mathcal{U}}\Big\|I_3 - \gamma(Q_{\mathbf{u}} + H_{\n{stab}}^TQ_{\mathbf{y}}H_{\n{stab}})\Big\| \\
        \leq  q^{\frac{\alpha({p-1})}{2}} d_{\mathcal{U}}  \label{eq:u_convergence}
    \end{equation}
    %% where a straightforward calculation shows that~$\Big\|I_3 - \gamma(Q_{\mathbf{u}} + H_{\n{stab}}^TQ_{\mathbf{y}}H_{\n{stab}})\Big\| \leq q^{\frac{1}{2}}$. 
    % We additionally know that based on \cite[Lemma 2]{chuy2025hybridsystemsmodelfeedback} that
    % \mh{Go ahead and include those steps here. It's okay if the steps that we use here
    % are very similar to those in the other paper.}    
    for each~$p \geq 1$. 
    
    \blue{We now move to Step~\ref{step:2}.} The model of~$\mathcal{H}_{FO}$ applies piecewise constant inputs to the underlying chaser dynamics.
    We therefore integrate the chaser's dynamics over the intervals with constant inputs and compute the distance between~$\mathbf{x}(t, j)$ and~$\tilde{\mathbf{x}}(t)$
    for an arbitrary~$(t, j) \in \dom \phi$.
    For each~$p \in \{0, \ldots, P\}$ the input is equal to~$\mathbf{u}\Htime{\xtime{p,0}}{\xtime{p,0}}$ over intervals of the 
    form~$[t_{\xtime{p,0}}, t_{\xtime{p+1,0}}]$, and
    therefore we may write 
    \begin{multline} \label{eq:thm1__1}
        \mathbf{x}(t, j) - \tilde{\mathbf{x}}(t)         
         = \exp(A_{\n{stab}}t)\mathbf{x}\Hotime{0}{0} 
        + \sum^{P-1}_{p=0}\int_{t_{\xtime{p,0}}}^{t_{\xtime{p+1,0}}} \!\!\!\!\!\!\!\!\!\!\exp\big(A_{\n{stab}}(t_{\xtime{p+1,0}}-\tau)\big)
        d\tau \\ \cdot B_{\n{stab}}\mathbf{u}\Htime{\xtime{p,0}}{\xtime{p,0}}  
        + \int_{t_{\xtime{P,0}}}^{t} \exp\big(A_{\n{stab}}(t - \tau)\big) d\tau B_{\n{stab}}\mathbf{u}\Htime{\xtime{P,0}}{\xtime{P,0}} \\
        - A_{\n{stab}}^{-1}B_{\n{stab}}K\mathbf{d}(t)  
         + A_{\n{stab}}^{-1}B_{\n{stab}}\tilde{\mathbf{u}}(t)  
        - \int^{t}_{0}\exp\big(A_{\n{stab}}(t - \tau)\big)B_{\n{stab}}K\mathbf{d}(\tau)d\tau,
    \end{multline}
    where $\tilde{\mathbf{x}}(t) = A_{\n{stab}}^{-1}B_{\n{stab}}K\mathbf{d}(t) - A_{\n{stab}}^{-1}B_{\n{stab}}\tilde{\mathbf{u}}(t)$
    is from~\eqref{eq:xtildedef}. 
    Next, using integration by parts, we find 
    \begin{multline} \label{eq:integrationbyparts}
        - \int^{t}_{0}\exp\big(A_{\n{stab}}(t - \tau)\big)B_{\n{stab}}K\mathbf{d}(\tau)d\tau  
        = A_{\n{stab}}^{-1}B_{\n{stab}}K\mathbf{d}(t)\\ - \exp\big(A_{\n{stab}}t\big)A_{\n{stab}}^{-1}B_{\n{stab}}K\mathbf{d}(0) 
        - \int^{t}_{0}\exp\big(A_{\n{stab}}(t - \tau)\big)A_{\n{stab}}^{-1}B_{\n{stab}}K\dot{\mathbf{d}}(\tau)d\tau. 
    \end{multline}
    We use~\eqref{eq:integrationbyparts} in~\eqref{eq:thm1__1} and
    add and subtract the term 
    \begin{multline}
    \sum_{p=0}^{P-1} \int_{t_{\xtime{p,0}}}^{t_{\xtime{p+1,0}}} \exp\big(A_{\n{stab}}(t_{\xtime{p+1,0}} - \tau)\big)d\tau B_{\n{stab}} \mathbf{u}^*\Htime{\xtime{p,0}}{\xtime{p,0}} \\
    + \int_{t_{\xtime{P,0}}}^{t} \exp\big(A_{\n{stab}}(t - \tau)\big)d\tau B_{\n{stab}}\mathbf{u}^*\Htime{\xtime{P,0}}{\xtime{P,0}} \\
    + \int_0^t \exp\big(A_{\n{stab}}(t - \tau)\big)d\tau B_{\n{stab}} \tilde{\mathbf{u}}(t) + \exp(A_{\n{stab}}t)B_{\n{stab}}K\mathbf{d}(t) 
    \end{multline}
    to find  
    \begin{multline} 
        \mathbf{x}(t, j) - \tilde{\mathbf{x}}(t)        
         = \exp(A_{\n{stab}}t)\mathbf{x}\Hotime{0}{0} \\
        + \sum^{P-1}_{p=0}\int_{t_{\xtime{p,0}}}^{t_{\xtime{p+1,0}}} \exp\big(A_{\n{stab}}(t_{\xtime{p+1,0}}-\tau)\big)
        d\tau B_{\n{stab}}\big(\mathbf{u}\Htime{\xtime{p,0}}{\xtime{p,0}} \\ - \mathbf{u}^*\Htime{\xtime{p,0}}{\xtime{p,0}}\big) 
        + \int_{t_{\xtime{P,0}}}^{t} \exp\big(A_{\n{stab}}(t - \tau)\big) d\tau  B_{\n{stab}} \big( \mathbf{u}\Htime{\xtime{P,0}}{\xtime{P,0}} \\ - \mathbf{u}^*\Htime{\xtime{p,0}}{\xtime{p,0}}\big) 
        + \sum^{P-1}_{p=0}\int_{t_{\xtime{p,0}}}^{t_{\xtime{p+1,0}}} \exp\big(A_{\n{stab}}(t_{\xtime{p+1,0}}-\tau)\big)
        d\tau B_{\n{stab}} \\ \cdot \big(\mathbf{u}^*\Htime{\xtime{p,0}}{\xtime{p,0}} - \tilde{\mathbf{u}}(t)\big) 
        \!\!+ \! \int_{t_{\xtime{P,0}}}^{t} \!\!\!\!\!\!\!\exp\big(A_{\n{stab}}(t - \tau)\big) d\tau  B_{\n{stab}}\big(\mathbf{u}^*\Htime{\xtime{P,0}}{\xtime{P,0}} \\ - \tilde{\mathbf{u}}(t)\big) 
        + \exp\big(A_{\n{stab}}t\big)A_{\n{stab}}^{-1}B_{\n{stab}}\tilde{\mathbf{u}}(t)
        -\exp\big(A_{\n{stab}}t\big)A_{\n{stab}}^{-1}B_{\n{stab}}K\mathbf{d}(0) \\
        - \int^{t}_{0} \exp\big(A_{\n{stab}}(t - \tau)\big)A_{\n{stab}}^{-1}B_{\n{stab}}K\dot{\mathbf{d}}(\tau)d\tau,
    \end{multline}
    where we have also used
    %\begin{equation} \label{eq:integralofexp}
        $\int_{0}^{t}  \exp\big(A_{\n{stab}}(t - \tau)\big) d\tau = \Big(\exp\big(A_{\n{stab}}t\big) -  I_{6}\Big)A_{\n{stab}}^{-1}$.         
    %\end{equation}
    % \begin{multline} 
    %     \mathbf{x}(t, j) - \tilde{\mathbf{x}}(t)         
    %      = \exp(A_{\n{stab}}t)\mathbf{x}\Hotime{0}{0} \\
    %     + \sum^{P-1}_{p=0}\int_{t_{\alphatime{p}+p}}^{t_{\xtime{p+1,0}}} \exp\big(A_{\n{stab}}(t_{\xtime{p+1,0}}-\tau)\big)
    %     d\tau \\ 
    %     \cdot B_{\n{stab}}\Big(\mathbf{u}\Htime{\xtime{p,0}}{\xtime{p,0}} -  \mathbf{u}^*\Htime{\xtime{p,0}}{\xtime{p,0}}\Big) \\
    %     + \int_{t_{\xtime{P,0}}}^{t} \exp\big(A_{\n{stab}}(t - \tau)\big) d\tau  B_{\n{stab}} \Big( \mathbf{u}\Htime{\xtime{p,0}}{\xtime{p,0}} - \mathbf{u}^*\Htime{\xtime{p,0}}{\xtime{p,0}}\Big) \\
    %     + \sum^{P-1}_{p=0}\int_{t_{\alphatime{p}+p}}^{t_{\xtime{p+1,0}}} \exp\big(A_{\n{stab}}(t_{\xtime{p+1,0}}-\tau)\big)
    %     d\tau \\ \cdot B_{\n{stab}}\big(\mathbf{u}^*\Htime{\xtime{p,0}}{\xtime{p,0}} - \tilde{\mathbf{u}}(t)\Big) \\
    %     + \int_{t_{\xtime{P,0}}}^{t} \exp\big(A_{\n{stab}}(t - \tau)\big) d\tau 
    %      B_{\n{stab}}\Big(\mathbf{u}^*\Htime{\xtime{p,0}}{\xtime{p,0}} - \tilde{\mathbf{u}}(t)\Big)\\
    %     + \exp\big(A_{\n{stab}}t\big)A_{\n{stab}}^{-1}B_{\n{stab}}\tilde{\mathbf{u}}(t) 
    %     - \exp\big(A_{\n{stab}}t\big)d\tau B_{\n{stab}}K\mathbf{d}(t)\\
    %     - \int^{t}_{0} \Big(\exp\big(A_{\n{stab}}t\big) - \exp\big(A_{\n{stab}}(t - \tau)\big)\Big)A_{\n{stab}}^{-1}B_{\n{stab}}K\dot{\mathbf{d}}(\tau)d\tau,
    % \end{multline}    
    After taking the norm of $\mathbf{x}(t, j) - \tilde{\mathbf{x}}(t)$ we use
    \begin{equation}
    \int^{t}_{0} \exp\big(A_{\n{stab}}t\big)A_{\n{stab}}^{-1}B_{\n{stab}}K\dot{\mathbf{d}}(\tau)d\tau = \exp\big(A_{\n{stab}}t\big)A_{\n{stab}}^{-1}B_{\n{stab}}K\big(\mathbf{d}(t) - \mathbf{d}(0)\big)
    \end{equation}
    and~$\tilde{\mathbf{x}}(t) = -A_{\n{stab}}^{-1}B_{\n{stab}}\tilde{\mathbf{u}}(t) + A_{\n{stab}}^{-1}B_{\n{stab}} K \mathbf{d}(t)$, and then we
    apply the triangle inequality to find 
    \begin{multline} 
        \big\|\mathbf{x}(t, j) - \tilde{\mathbf{x}}(t) \big\|       
         \leq \big\|\exp(A_{\n{stab}}t)\big\|\big\|\mathbf{x}\Hotime{0}{0} - \tilde{\mathbf{x}}(t) \big\| \\
        + \frac{1}{m_{c}}\sum^{P-1}_{p=0}\int_{t_{\xtime{p,0}}}^{t_{\xtime{p+1,0}}} \big\|\exp\big(A_{\n{stab}}(t_{\xtime{p+1,0}}-\tau)\big)\big\|
        d\tau \big\|\mathbf{u}\Htime{\xtime{p,0}}{\xtime{p,0}} \\ 
        -  \mathbf{u}^*\Htime{\xtime{p,0}}{\xtime{p,0}} \big\| 
        + \frac{1}{m_{c}}\int_{t_{\xtime{P,0}}}^{t} \big\|\exp\big(A_{\n{stab}}(t - \tau)\big)\big\| d\tau 
         \big\| \mathbf{u}\Htime{\xtime{P,0}}{\xtime{P,0}} \\ 
         - \mathbf{u}^*\Htime{\xtime{p,0}}{\xtime{p,0}} \big\| 
        + \frac{1}{m_{c}}\sum^{P-1}_{p=0}\int_{t_{\xtime{p,0}}}^{t_{\xtime{p+1,0}}} \big\|\exp\big(A_{\n{stab}}(t_{\xtime{p+1,0}}-\tau)\big)\big\|
        d\tau \\ \cdot \big\|\mathbf{u}^*\Htime{\xtime{p,0}}{\xtime{p,0}} - \tilde{\mathbf{u}}(t)\big\| 
        + \frac{1}{m_{c}}\int_{t_{\xtime{P,0}}}^{t} \big\|\exp\big(A_{\n{stab}}(t - \tau)\big)\big\| d\tau 
        \big\|\mathbf{u}^*\Htime{\xtime{P,0}}{\xtime{P,0}} \\  - \tilde{\mathbf{u}}(t)\big\|
        + \frac{1}{m_{c}}\int^{t}_{0} \big\|\exp\big(A_{\n{stab}}t\big) - \exp\big(A_{\n{stab}}(t - \tau)\big)\big\|\big\|A_{\n{stab}}^{-1}\big\| \big\|K\big\|\big\|\dot{\mathbf{d}}(\tau)\big\|d\tau,
    \end{multline}
    where we have used~$\big\|B_{\n{stab}}\big\| = \frac{1}{m_{c}}$, which 
    \green{is a property of the CW equations that}
    follows from~$B_{\n{stab}}^TB_{\n{stab}} = \frac{1}{m_c^2}I_3$. 

    \blue{Now we turn to Step~\ref{step:2point5}.}
    From Assumption~\ref{as:d} we have $\big\|\dot{\mathbf{d}}(t)\big\| \leq \bar{\mathbf{d}}$, and
    since~$\mathcal{U}$
    is compact we have~$\|u_1 - u_2\| \leq d_{\mathcal{U}}$ for all~$u_1, u_2 \in \mathcal{U}$.
    Using these facts and~\eqref{eq:u_convergence} gives
    %We then simplify the bound to 
    % \begin{multline} 
    %     \big\|\mathbf{x}(t, j) - \tilde{\mathbf{x}}(t) \big\|       
    %      \leq \big\|\exp(A_{\n{stab}}t)\big\|\big\|\mathbf{x}\Hotime{0}{0} - \tilde{\mathbf{x}}(t) \big\| \\
    %     + \sum^{P-1}_{p=0}\int_{t_{\alphatime{p}+p}}^{t_{\xtime{p+1,0}}} \big\|\exp\big(A_{\n{stab}}(t_{\xtime{p+1,0}}-\tau)\big)\big\|
    %     d\tau \\ 
    %     \cdot \big\|\mathbf{u}\Htime{\xtime{p,0}}{\xtime{p,0}} -  \mathbf{u}^*\Htime{\xtime{p,0}}{\xtime{p,0}} \big\| \\
    %     + \int_{t_{\xtime{P,0}}}^{t} \big\|\exp\big(A_{\n{stab}}(t - \tau)\big)\big\| d\tau   \big\| \mathbf{u}\Htime{\xtime{p,0}}{\xtime{p,0}} - \mathbf{u}^*\Htime{\xtime{p,0}}{\xtime{p,0}} \big\| \\
    %     + \sum^{P-1}_{p=0}\int_{t_{\alphatime{p}+p}}^{t_{\xtime{p+1,0}}} \big\|\exp\big(A_{\n{stab}}(t_{\xtime{p+1,0}}-\tau)\big)\big\|
    %     d\tau \big\|\mathbf{u}^*\Htime{\xtime{p,0}}{\xtime{p,0}} - \tilde{\mathbf{u}}(t)\big\| \\
    %     + \int_{t_{\xtime{P,0}}}^{t} \big\|\exp\big(A_{\n{stab}}(t - \tau)\big)\big\| d\tau 
    %     \big\|\mathbf{u}^*\Htime{\xtime{p,0}}{\xtime{p,0}} - \tilde{\mathbf{u}}(t)\big\|\\
    %     + \int^{t}_{0} \big\|\exp\big(A_{\n{stab}}t\big) - \exp\big(A_{\n{stab}}(t - \tau)\big)\big\|\big\|A_{\n{stab}}^{-1}\big\|\big\|K\big\|\bar{\mathbf{d}}d\tau.
    % \end{multline}
    % By definition of $d_{\mathcal{U}}$ and \eqref{eq:thm1_inptdiff}, the bound simplifies to
    \begin{multline}  \label{eq:thm1_7}
        \big\|\mathbf{x}(t, j) - \tilde{\mathbf{x}}(t) \big\|       
        \leq \big\|\exp(A_{\n{stab}}t)\big\|\big\|\mathbf{x}\Hotime{0}{0} - \tilde{\mathbf{x}}(t)\big\| \\
        + \frac{1}{m_{c}}\int_{0}^{t_{\xtime{1,0}}} \big\|\exp\big(A_{\n{stab}}(t_{\xtime{1,0}}-\tau)\big)\big\| 
        d\tau \; d_{\mathcal{U}}\\
        + \frac{1}{m_{c}}\sum^{P-1}_{p=1}\int_{t_{\xtime{p,0}}}^{t_{\xtime{p+1,0}}} \big\|\exp\big(A_{\n{stab}}(t_{\xtime{p+1,0}}-\tau)\big)\big\| d\tau q^{\frac{\alpha({p-1})}{2}} \; d_{\mathcal{U}} \\
        + \frac{1}{m_{c}}\int_{t_{\xtime{P,0}}}^{t} \big\|\exp\big(A_{\n{stab}}(t - \tau)\big)\big\| d\tau   
        \; q^{\frac{\alpha({P-1})}{2}} d_{\mathcal{U}} 
        + \frac{1}{m_{c}}\int_{0}^{t} \big\|\exp\big(A_{\n{stab}}(t - \tau)\big)\big\| d\tau 
        \; d_{\mathcal{U}}\\
        + \frac{1}{m_{c}}\int^{t}_{0} \big\|\exp\big(A_{\n{stab}}t\big) - \exp\big(A_{\n{stab}}(t - \tau)\big)\big\|d\tau \big\|A_{\n{stab}}^{-1}\big\|\big\|K\big\|\bar{\mathbf{d}}, 
    \end{multline}
    where we have combined integrals, separated the~$p = 0$ term 
    from the sum, 
    and used~$\alpha(-1) = 0$.
    The last integrand can be bounded via 
    \begin{align} 
        \big\|\exp\big(A_{\n{stab}}t\big) - \exp\big(A_{\n{stab}}(t - \tau)\big)\big\| 
        &\leq \big\|\exp\big(A_{\n{stab}}(t - \tau)\big\|\big\|\exp\big(A_{\n{stab}}\tau\big) - I_6\big\| \\
        &\leq \big\|\exp\big(A_{\n{stab}}(t - \tau)\big)\big\|\big(\big\|\exp\big(A_{\n{stab}}\tau\big)\big\| + 1\big). \label{eq:thm1__2} 
    \end{align}
    From \cite[Theorem 2]{loan1975matrixExp}, we have
    \begin{equation} \label{eq:thm1_9}
         \big\|\exp(A_{\n{stab}}t)\big\| \leq 
         \mu_{\max}(\Lambda_{\n{des}}) \frac{|\lambda_6|}{|\lambda_1|} \exp\big({-|\lambda_1|}t\big),
    \end{equation}
    \green{which comes directly from the stabilized CW equations that we study.}
    %where~$\frac{|\lambda_6|}{|\lambda_1|}$ is the spectral condition number of~$A_{\n{stab}}$. 
    Using this bound with \eqref{eq:thm1__2} gives 
    % \begin{multline}
    %     \int^{t}_{0} \big\|\exp\big(A_{\n{stab}}t\big) - \exp\big(A_{\n{stab}}(t - \tau)\big)d\tau\\
    %     \leq \int^{t}_{0} \mu_{\max}(\Lambda_{\n{des}}) \frac{|\lambda_6|}{|\lambda_1|} \exp\big({-|\lambda_1|}(t - \tau)\big)\Big(\mu_{\max}(\Lambda_{\n{des}}) \frac{|\lambda_6|}{|\lambda_1|} \exp\big({-|\lambda_1|}\tau\big) + 1\Big)d\tau\\
    %     \leq \int^{t}_{0}\mu_{\max}(\Lambda_{\n{des}})^2 \Bigg(\frac{|\lambda_6|}{|\lambda_1|}\Bigg)^2\exp\big({-|\lambda_1|}t\big) + \mu_{\max}(\Lambda_{\n{des}}) \frac{|\lambda_6|}{|\lambda_1|} \exp\big({-|\lambda_1|}(t - \tau)\big)d\tau\\
    %     \leq \mu_{\max}(\Lambda_{\n{des}})^2 \Bigg(\frac{|\lambda_6|}{|\lambda_1|}\Bigg)^2t\exp\big({-|\lambda_1|}t\big) + \int^{t}_{0} \mu_{\max}(\Lambda_{\n{des}}) \frac{|\lambda_6|}{|\lambda_1|} \exp\big({-|\lambda_1|}(t - \tau)\big)d\tau\\
    %     \leq \mu_{\max}(\Lambda_{\n{des}})^2 \Bigg(\frac{|\lambda_6|}{|\lambda_1|}\Bigg)^2t\exp\big({-|\lambda_1|}t\big) + \mu_{\max}(\Lambda_{\n{des}}) \frac{|\lambda_6|}{|\lambda_1|} \int^{t}_{0}\exp\big({-|\lambda_1|}(t - \tau)\big)d\tau\\
    % \end{multline}
    \begin{multline}\label{eq:thm1__3}
        \int^{t}_{0} \big\|\exp\big(A_{\n{stab}}t\big) - \exp\big(A_{\n{stab}}(t - \tau)\big) \big\|d\tau 
        \leq \mu_{\max}(\Lambda_{\n{des}})^2 \Bigg(\frac{|\lambda_6|}{|\lambda_1|}\Bigg)^2t\exp\big({-|\lambda_1|}t\big) \\ 
        + \mu_{\max}(\Lambda_{\n{des}}) \frac{|\lambda_6|}{|\lambda_1|} \int^{t}_{0}\exp\big({-|\lambda_1|}(t - \tau)\big)d\tau.
    \end{multline}
    Assumption \ref{as:timescale} ensures $\alpha(p) \geq \ell$ for all $p \in \mathbb{N}$. Since~$q\in (0,1)$, % holds from \cite[Lemma 3]{chuy2025hybridsystemsmodelfeedback}, 
    we have $q^{\frac{\alpha({p-1})}{2}} \leq q^{\frac{\ell}{2}}$ for all~$p \geq 1$. 

    \blue{Now we address Step~\ref{step:3}.} 
    Using~\eqref{eq:thm1__3} in~\eqref{eq:thm1_7}, evaluating the integrals, and combining with~\eqref{eq:thm1_9} and~\eqref{eq:thm1__3} gives
    \begin{multline} \label{eq:secondbigone}
        \big\|\mathbf{x}(t, j) - \tilde{\mathbf{x}}(t) \big\|       
        \leq \mu_{\max}(\Lambda_{\n{des}})\frac{|\lambda_6|}{|\lambda_1|} \exp\big({-\big|\lambda_{1}\big|}t\big)\big\|\mathbf{x}\Hotime{0}{0} - \tilde{\mathbf{x}}(t)\big\| \\
        + \frac{\mu_{\max}(\Lambda_{\n{des}})\frac{|\lambda_6|}{|\lambda_1|}d_{\mathcal{U}}}{m_{c}\big|\lambda_{1}\big|}\Big[1-\exp\big({-\big|\lambda_{1}\big|}t_{\xtime{1,0}}\big)\Big] \\
        + \frac{\mu_{\max}(\Lambda_{\n{des}})\frac{|\lambda_6|}{|\lambda_1|} q^{\frac{\ell}{2}} d_{\mathcal{U}} }{m_{c}\big|\lambda_{1}\big|} \Big[1-\exp\big({-\big|\lambda_{1}\big|}(t - t_{\xtime{1,0}})\big)\Big] \\
        + \frac{\mu_{\max}(\Lambda_{\n{des}})\frac{|\lambda_6|}{|\lambda_1|}d_{\mathcal{U}} }{m_{c}\big|\lambda_{1}\big|} \Big[1-\exp\big({-\big|\lambda_{1}\big|}t\big)\Big]\\
        + \frac{\mu_{\max}(\Lambda_{\n{des}})^2\Big(\frac{|\lambda_6|}{|\lambda_1|}\Big)^2}{m_{c}}\big\|A_{\n{stab}}^{-1}\big\|\big\|K\big\|\bar{\mathbf{d}}t\exp\big({-|\lambda_1|}t\big) \\
        + \frac{\mu_{\max}(\Lambda_{\n{des}}) \frac{|\lambda_6|}{|\lambda_1|}}{m_{c}|\lambda_1|}\big\|A_{\n{stab}}^{-1}\big\|\big\|K\big\|\bar{\mathbf{d}} \Big[1-\exp\big({-\big|\lambda_{1}\big|}t\big)\Big].
    \end{multline}
    Next, using~$\atimer{,\min} \leq t_{\xtime{1,0}} \leq \atimer{,\max}$ gives
        the upper bound 
        \begin{multline} \label{eq:thm1___3}
           \big\|\mathbf{x}(t, j) - \tilde{\mathbf{x}}(t) \big\|       
            \leq \mu_{\max}(\Lambda_{\n{des}})\frac{|\lambda_6|}{|\lambda_1|} \exp\big({-\big|\lambda_{1}\big|}t\big)\big\|\mathbf{x}\Hotime{0}{0} - \tilde{\mathbf{x}}(t)\big\| \\
            + \frac{\mu_{\max}(\Lambda_{\n{des}})\frac{|\lambda_6|}{|\lambda_1|}d_{\mathcal{U}}}{m_{c}\big|\lambda_{1}\big|}\Big[1-\exp\big({-\big|\lambda_{1}\big|}\atimer{,\max}\big)\Big] \\
            + \frac{\mu_{\max}(\Lambda_{\n{des}})\frac{|\lambda_6|}{|\lambda_1|} q^{\frac{\ell}{2}} d_{\mathcal{U}}}{m_{c}\big|\lambda_{1}\big|}\Big[1 - \exp\big({-\big|\lambda_{1}\big|}(t - \atimer{,\min})\big)\Big] \\
            + \frac{\mu_{\max}(\Lambda_{\n{des}}) \frac{|\lambda_6|}{|\lambda_1|}d_{\mathcal{U}}}{m_{c}|\lambda_1|}\Big[1-\exp\big({-\big|\lambda_{1}\big|}t\big)\Big]\\
            + \frac{\mu_{\max}(\Lambda_{\n{des}})^2\Big(\frac{|\lambda_6|}{|\lambda_1|}\Big)^2}{m_{c}}\big\|A_{\n{stab}}^{-1}\big\|\big\|K\big\|\bar{\mathbf{d}}t\exp\big({-|\lambda_1|}t\big) \\
            + \frac{\mu_{\max}(\Lambda_{\n{des}}) \frac{|\lambda_6|}{|\lambda_1|}}{m_{c}|\lambda_1|}\big\|A_{\n{stab}}^{-1}\big\|\big\|K\big\|\bar{\mathbf{d}} \Big[1-\exp\big({-\big|\lambda_{1}\big|}t\big)\Big].
        \end{multline}
        By factoring we see that 
        \begin{equation} \label{eq:thm1___4}
            \exp\big({-\big|\lambda_{1}\big|}(t - \atimer{,\min}\big) = \exp\big({\big|\lambda_{1}\big|}\atimer{,\min}\big)\exp\big({-\big|\lambda_{1}\big|}t\big). 
        \end{equation}
    By definition of $\|\cdot\|_{\mathcal{A}(t)}$, it holds that 
    \begin{equation} \label{eq:thm1_14}
        \|x(0, 0) - \tilde{x}(t)\| = \|\phi(0, 0)\|_{\mathcal{A}(t)}, 
    \end{equation}
    and combining~\eqref{eq:t1normsequal},~\eqref{eq:thm1___3},~\eqref{eq:thm1___4} and~\eqref{eq:thm1_14} gives
    \begin{multline}\label{eq:thm1_17}
        \|\phi(t, j)\|_{\mathcal{A}(t)}
        \leq \mu_{\max}(\Lambda_{\n{des}})\frac{|\lambda_6|}{|\lambda_1|}\exp\big({-\big|\lambda_{1}\big|}t\big) \big\|\phi\Hotime{0}{0}\big\|_{\mathcal{A}(t)}\\
        + \frac{\mu_{\max}(\Lambda_{\n{des}})\frac{|\lambda_6|}{|\lambda_1|}d_{\mathcal{U}}}{m_{c}\big|\lambda_{1}\big|}\big[2 - \exp(-|\lambda_1|\atimer{,\max}) -\exp\big({-\big|\lambda_{1}\big|}t\big)\big] \\
        + \frac{\mu_{\max}(\Lambda_{\n{des}})\frac{|\lambda_6|}{|\lambda_1|} q^{\frac{\ell}{2}} d_{\mathcal{U}}}{m_{c}\big|\lambda_{1}\big|}\Big[1 - \exp\big({\big|\lambda_{1}\big|}\atimer{,\min}\big)\exp\big({-\big|\lambda_{1}\big|}t\big)\Big] \\
        + \frac{\mu_{\max}(\Lambda_{\n{des}}) \frac{|\lambda_6|}{|\lambda_1|}}{m_{c}|\lambda_1|}\big\|A_{\n{stab}}^{-1}\big\|\big\|K\big\|\bar{\mathbf{d}} \Big[1 + \big(\mu_{\max}(\Lambda_{\n{des}})|\lambda_6|t - 1\big)\exp\big({-\big|\lambda_{1}\big|}t\big)\Big].
    \end{multline}   
    Then
    \begin{multline}
        \limsup_{t+j \rightarrow \infty} \|\phi(t, j)\|_{\mathcal{A}(t)} \leq \mu_{\max}(\Lambda_{\n{des}})\frac{|\lambda_6|}{m_{c}|\lambda_1|^2} \cdot \\
        \big(2d_{\mathcal{U}} -d_{\mathcal{U}}\exp({-\big|\lambda_{1}\big|}\atimer{,\max}) 
        + d_{\mathcal{U}}q^{\frac{\ell}{2}} + \big\|A_{\n{stab}}^{-1}\big\|\big\|K\big\|\bar{\mathbf{d}}\big). 
    \end{multline}

\subsection{Proof of Theorem \ref{tm:globalcompleteHOCW}} \label{Ap:theorem2Proof}
        In this proof any empty sums evaluate to zero by convention, and
we take $\alpha(-1) = 0$ for notational convenience. 
        An arbitrary initial condition~$\phi(0, 0)$
        allows for arbitrary
        $\mathbf{z}_{0}\Hotime{0}{0} \in \mathcal{U}$, 
        $\ctimer{}(0, 0) \in [0,\tau_{g,\n{comp}}]$, and
        $\atimer{}(0,0) \in [0,\atimer{,\max}]$.
        In particular, the timer~$\atimer{}$
        can reach zero for the first time before the timer~$\ctimer{}$
        reaches zero~$\ell$ times, and
        the bound in Proposition~\ref{prop:completeHOCW} does not apply to such cases. 
        If that scenario does occur, then~$\atimer{}$ is reset to a value
        in the interval~$[\atimer{,\min}, \atimer{,\max}]$.
        After that reset, 
        Assumption~\ref{as:timescale}
        ensures that~$\ctimer{}$ reaches zero at least~$\ell$ times
        before~$\atimer{}$ reaches zero again.
        Then at least~$\ell$ gradient descent iterations will have
        been performed before the input~$\mathbf{u}$ changes
        for the second time. The second change in~$\mathbf{u}$
        occurs at hybrid time~$(t_{\xtime{2,0}},\xtime{2,0})$,
        and we next characterize the behavior of~$\HFO$ up to that hybrid time. 

        The input over the interval~$[0, t_{\xtime{1,0}}]$
        is~$\mathbf{u}(0,0)$, and to derive a worst-case bound 
        we suppose that zero gradient descent iterations
        are performed before the first jump in~$\atimer{}$.
        Then~$\mathbf{u}(t_{\xtime{1,0}}, \xtime{1,0})
        = \mathbf{u}(0, 0)$ because no computations have been
        performed to change the input. 
        Then~$\mathbf{u}(0, 0)$ is the input to the system
        over the interval~$[0, t_{\xtime{2,0}}]$.
        We know that~$\atimer{,\min}\leq t_{\xtime{2,0}} \leq 2\atimer{,\max}$ because, by continuous
        time~$t_{\xtime{2,0}}$, the timer~$\atimer{}$
        has reached zero from its initial condition in~$[0,\atimer{,\max}]$,
        then reset to some value in~$[\atimer{,\min}, \atimer{,\max}]$,
        then reached zero again. 
        %Supposing that zero gradient descent iterations have been completed during the interval~$[t_0, t_{\xtime{1,0}}]$, the input to the system is~$\mathbf{u}(0, 0)$ over the interval~$[t_0, t_{\xtime{2,0}}]$. 
        
        Repeating the same steps used to reach~\eqref{eq:thm1_7},
        using~$\alpha(p) \geq \ell$ for all~$p \geq 2$, and combining integrals
        we find
        % \begin{multline} \label{eq:thm2_1}
        %     \|\phi(t_{\alphatime{2}+2}, \alphatime{2}+2)\|_{\mathcal{A}(t_{\alphatime{2}+2})}
        %     \leq 
        %     \mu_{\max}(\Lambda_{\n{des}})\frac{|\lambda_6|}{|\lambda_1|}\exp\big({-\big|\lambda_{1}\big|\atimer{,\min}}\big) 
        %     \big\|\phi\Hotime{0}{0}\big\|_{\mathcal{A}(t)} \\
        %     + \frac{\mu_{\max}(\Lambda_{\n{des}})\frac{|\lambda_6|}{|\lambda_1|}}{m_{c}\big|\lambda_{1}\big|}\Big(2d_{\mathcal{U}} + \big\|A_{\n{stab}}^{-1}\big\|\big\|K\big\|\bar{\mathbf{d}}\Big)\\
        %     +  \frac{\mu_{\max}(\Lambda_{\n{des}})^2\Big(\frac{|\lambda_6|}{|\lambda_1|}\Big)^2}{m_{c}}\big\|A_{\n{stab}}^{-1}\big\|\big\|K\big\|\bar{\mathbf{d}}\atimer{,\min}\exp\big({-\big|\lambda_{1}\big|}\atimer{,\min}\big). 
        % \end{multline}
        \begin{multline}\label{eq:thm2_1}
            \big\|\mathbf{x}(t, j) - \tilde{\mathbf{x}}(t) \big\|       
            \leq \big\|\exp(A_{\n{stab}}t)\big\|\big\|\mathbf{x}\Hotime{0}{0} - \tilde{\mathbf{x}}(t)\big\| \\
            + \frac{1}{m_{c}}\int_{0}^{t_{\xtime{2,0}}} \big\|\exp\big(A_{\n{stab}}(t_{\xtime{2,0}} - \tau)\big)\big\| d\tau   
            \; d_{\mathcal{U}} \\
            + \frac{1}{m_{c}}\int_{t_{\xtime{2,0}}}^{t} \big\|\exp\big(A_{\n{stab}}(t - \tau)\big)\big\| d\tau   
            \; q^{\frac{\ell}{2}} d_{\mathcal{U}} \\ 
            + \frac{1}{m_{c}}\int_{0}^{t} \big\|\exp\big(A_{\n{stab}}(t - \tau)\big)\big\| d\tau 
            d_{\mathcal{U}}\\
            + \frac{1}{m_{c}}\int^{t}_{0} \big\|\exp\big(A_{\n{stab}}t\big) - \exp\big(A_{\n{stab}}(t - \tau)\big)\big\|d\tau \big\|A_{\n{stab}}^{-1}\big\|\big\|K\big\|\bar{\mathbf{d}},
        \end{multline}
        where the first integral does 
        %the~$p=0$ and~$p=1$ terms from the summation in~\eqref{eq:thm1_7} do 
        not have a factor of~$q^{\frac{\ell}{2}}$ in it
        in order
        to account for the case in which zero gradient
        descent iterations are performed before the first jump in~$\mathbf{u}$.

        Using~\eqref{eq:thm1_9} and~\eqref{eq:thm1__3} in~\eqref{eq:thm2_1} gives
        \begin{multline} \label{eq:thm2__1}
           \big\|\mathbf{x}(t, j) - \tilde{\mathbf{x}}(t) \big\|       
            \leq \mu_{\max}(\Lambda_{\n{des}})\frac{|\lambda_6|}{|\lambda_1|} \exp\big({-\big|\lambda_{1}\big|}t\big)\big\|\mathbf{x}\Hotime{0}{0} - \tilde{\mathbf{x}}(t)\big\| \\
            + \frac{\mu_{\max}(\Lambda_{\n{des}})\frac{|\lambda_6|}{|\lambda_1|}d_{\mathcal{U}}}{m_{c}\big|\lambda_{1}\big|}\Big[1-\exp\big({-\big|\lambda_{1}\big|}t_{\xtime{2,0}}\big)\Big] \\
            + \frac{\mu_{\max}(\Lambda_{\n{des}})\frac{|\lambda_6|}{|\lambda_1|} q^{\frac{\ell}{2}} d_{\mathcal{U}}}{m_{c}\big|\lambda_{1}\big|}\Big[1 - \exp\big({-\big|\lambda_{1}\big|}(t - t_{\xtime{2,0}})\big)\Big] \\
            + \frac{\mu_{\max}(\Lambda_{\n{des}}) \frac{|\lambda_6|}{|\lambda_1|}d_{\mathcal{U}}}{m_{c}|\lambda_1|}\Big[1-\exp\big({-\big|\lambda_{1}\big|}t\big)\Big] \\
            + \frac{\mu_{\max}(\Lambda_{\n{des}})^2\Big(\frac{|\lambda_6|}{|\lambda_1|}\Big)^2}{m_{c}}\big\|A_{\n{stab}}^{-1}\big\|\big\|K\big\|\bar{\mathbf{d}}t\exp\big({-|\lambda_1|}t\big) \\
            + \frac{\mu_{\max}(\Lambda_{\n{des}}) \frac{|\lambda_6|}{|\lambda_1|}}{m_{c}|\lambda_1|}\big\|A_{\n{stab}}^{-1}\big\|\big\|K\big\|\bar{\mathbf{d}} \Big[1-\exp\big({-\big|\lambda_{1}\big|}t\big)\Big].
        \end{multline}
        Next, using~$\atimer{,\min} \leq t_{\xtime{2,0}} \leq 2\atimer{,\max}$ gives
        the upper bound 
        \begin{multline} \label{eq:thm2_3}
           \big\|\mathbf{x}(t, j) - \tilde{\mathbf{x}}(t) \big\|       
            \leq \mu_{\max}(\Lambda_{\n{des}})\frac{|\lambda_6|}{|\lambda_1|} \exp\big({-\big|\lambda_{1}\big|}t\big)\big\|\mathbf{x}\Hotime{0}{0} - \tilde{\mathbf{x}}(t)\big\| \\
            + \frac{\mu_{\max}(\Lambda_{\n{des}})\frac{|\lambda_6|}{|\lambda_1|}d_{\mathcal{U}}}{m_{c}\big|\lambda_{1}\big|}\Big[1-\exp\big({-2\big|\lambda_{1}\big|}\atimer{,\max}\big)\Big] \\
            + \frac{\mu_{\max}(\Lambda_{\n{des}})\frac{|\lambda_6|}{|\lambda_1|} q^{\frac{\ell}{2}} d_{\mathcal{U}}}{m_{c}\big|\lambda_{1}\big|}\Big[1 - \exp\big({-\big|\lambda_{1}\big|}(t - \atimer{,\min})\big)\Big] \\
            + \frac{\mu_{\max}(\Lambda_{\n{des}}) \frac{|\lambda_6|}{|\lambda_1|}d_{\mathcal{U}}}{m_{c}|\lambda_1|}\Big[1-\exp\big({-\big|\lambda_{1}\big|}t\big)\Big]\\
            + \frac{\mu_{\max}(\Lambda_{\n{des}})^2\Big(\frac{|\lambda_6|}{|\lambda_1|}\Big)^2}{m_{c}}\big\|A_{\n{stab}}^{-1}\big\|\big\|K\big\|\bar{\mathbf{d}}t\exp\big({-|\lambda_1|}t\big) \\
            + \frac{\mu_{\max}(\Lambda_{\n{des}}) \frac{|\lambda_6|}{|\lambda_1|}}{m_{c}|\lambda_1|}\big\|A_{\n{stab}}^{-1}\big\|\big\|K\big\|\bar{\mathbf{d}} \Big[1-\exp\big({-\big|\lambda_{1}\big|}t\big)\Big].
        \end{multline}
        % By factoring we see that 
        % \begin{equation} \label{eq:thm2_4}
        %     \Big[\exp\big({-\big|\lambda_{1}\big|}(t - 2\atimer{,\max}\big)-\exp\big({-\big|\lambda_{1}\big|}t\big)\Big] = \big[\exp\big(2{\big|\lambda_{1}\big|}\atimer{,\max}\big) - 1\big]\exp\big({-\big|\lambda_{1}\big|}t\big),
        % \end{equation}
        Combining~\eqref{eq:t1normsequal},~\eqref{eq:thm1___4},~\eqref{eq:thm1_14}, and~\eqref{eq:thm2_3} gives
        \begin{multline} \label{eq:thm2__2}
            \|\phi(t, j)\|_{\mathcal{A}(t)}
            \leq \mu_{\max}(\Lambda_{\n{des}})\frac{|\lambda_6|}{|\lambda_1|}\exp\big({-\big|\lambda_{1}\big|}t\big) \big\|\phi\Hotime{0}{0}\big\|_{\mathcal{A}(t)}\\
            + \frac{\mu_{\max}(\Lambda_{\n{des}})\frac{|\lambda_6|}{|\lambda_1|}d_{\mathcal{U}}}{m_{c}\big|\lambda_{1}\big|}\big[2 - \exp(-2|\lambda_1|\atimer{,\max}) -\exp\big({-\big|\lambda_{1}\big|}t\big)\big] \\
            + \frac{\mu_{\max}(\Lambda_{\n{des}})\frac{|\lambda_6|}{|\lambda_1|} q^{\frac{\ell}{2}} d_{\mathcal{U}}}{m_{c}\big|\lambda_{1}\big|}\Big[1 - \exp\big({\big|\lambda_{1}\big|}\atimer{,\min}\big)\exp\big({-\big|\lambda_{1}\big|}t\big)\Big] \\
            + \frac{\mu_{\max}(\Lambda_{\n{des}}) \frac{|\lambda_6|}{|\lambda_1|}}{m_{c}|\lambda_1|}\big\|A_{\n{stab}}^{-1}\big\|\big\|K\big\|\bar{\mathbf{d}} \Big[1 + \big(\mu_{\max}(\Lambda_{\n{des}})|\lambda_6|t - 1\big)\exp\big({-\big|\lambda_{1}\big|}t\big)\Big],
        \end{multline}
        and the asymptotic result follows by taking
        the limit superior. 

\end{document}